%% file: mainv11.tex
\documentclass[onecolumn,lettersize,journal]{IEEEtran}

\usepackage[utf8]{inputenc}
\usepackage[T1]{fontenc}
\usepackage{url}
\usepackage{ifthen}
\usepackage{cite}
\usepackage[cmex10]{amsmath} 
\usepackage{subfigure}
\usepackage{amsmath}
\usepackage{amsfonts}
\usepackage{amssymb}
\usepackage{enumerate}
\usepackage{bm}
\usepackage{graphicx}
\usepackage{color}
\usepackage[bottom]{footmisc}
\usepackage{dsfont}
\usepackage{multirow} 
\usepackage{arydshln}
\usepackage[caption=false,font=footnotesize]{subfig}
\usepackage{standalone}
\usepackage{derivative}
\usepackage{tikz}
\allowdisplaybreaks
\normalsize

\newtheorem{theorem}{Theorem}
\newtheorem{lemma}{Lemma}

\newtheorem{Remark}{Remark}

\newtheorem{definition}{Definition}

\newtheorem{Corollary}{Corollary}

\newcommand{\RNum}[1]{\uppercase\expandafter{\romannumeral #1\relax}}

\begin{document}
\title{On the Nonasymptotic Bounds of Joint Source-Channel Coding with Hierarchical Sources}
\author{Shuo Shao, \textit{Member, IEEE}, Chao Qi, \textit{Member, IEEE}, Jincheng Dai, \textit{Member, IEEE}, Wenrui Dai,\textit{ Member, IEEE,} and Hongkai Xiong, \textit{Fellow, IEEE}
	
	\thanks{Shuo Shao is with the School of Information and Electronic Engineering, East China Normal University, Shanghai, China. E-mail: sshao@ee.ecnu.edu.cn}
	\thanks{Chao Qi is with the School of Computer and Software Engineering, Xihua University, Chengu, China. E-mail: chaoqi@xhu.edu.cn}
	\thanks{Jincheng Dai is with the Key Laboratory of Universal Wireless Communications, Ministry of Education, Beijing University of Posts and Telecommunications, Beijing, China. E-mail: daijincheng@bupt.edu.cn}
	\thanks{Wenrui Dai and Hongkai Xiong are with the Department of Electronic Engineering, Shanghai Jiao Tong University, Shanghai, China. E-mail: daiwenrui@sjtu.edu.cn, xionghongkai@sjtu.edu.cn}
}

\maketitle

  \begin{abstract}
  	This paper establishes tractable bounds of joint source-channel coding with hierarchical sources in the finite blocklength regime. In this setting, both the indirect source and observable source must be reconstructed under correlated distortion constraints, leading to a joint excess-distortion event.  
  	First, to build computable tight bounds, we introduce a novel $\mathsf{d}(\cdot)$-functional distortion relaxation, which enables tractable and tight bounding of the joint excess-distortion probability induced by correlated sources. 
  	By this approach, the nonasymptotic converse and achievability bounds are given.   	
  	Second, Gaussian approximations for the proposed bounds are obtained, which are optimal for the transmission of a Gaussian memoryless source over an additive white Gaussian noise channel with mean-square error distortion.
  	The optimal scheme is obtained via a structured analysis that captures the intrinsic tradeoff between semantic and observable reconstructions. 
  	Furthermore, for the transmission of Gaussian memoryless sources over AWGN channels, we obtain explicit and computable bounds, by providing a new geometric structure involving three correlated spherical regions.
  	This results extend the classical two-spherical region analysis for a single distortion constraint.
  	Numerical simulations demonstrate that the proposed achievability and converse bounds tightly sandwich the Gaussian approximation and align closely with Monte Carlo numerical results.
  \end{abstract}


\section{Introduction}

\IEEEPARstart{R}{ecently}, a new source coding model is proposed by Liu \textit{et al.}\cite{Liu-semantic-2022} called semantic rate distortion. In this model, there is a pair of hierarchical sources $(S,X)$ with joint distribution $P_{SX}$, where $X$ is the observable source and $S$ is the unobservable indirect source. Unlike the setup in the traditional indirect source coding model, now recovering both $S$ and $X$ is required, along with two distortion measures $\mathsf{d}_{s}(S,\hat{S})$ and $\mathsf{d}_{x}(X,\hat{X})$ to evaluate the quality of reconstructed data $\hat{S}$ and $\hat{X}$ respectively. This model can be regarded as the information-theoretic model for many practical scenarios, especially for recent learning-based compression and communication  methods \cite{Gunduz-23,Zhang-imagesemantic-2023}.  
{  Specifically, in learning-based compression and communication methods such as task-oriented coding or semantic communications, it is often necessary to recover both the observable data and its hidden features. For example, when compressing an image using a learning-based method, one typically needs to reconstruct both the image itself and its category label. The semantic rate distortion model studied in this paper formulates these scenarios, where the unobservable indirect source $S$ represents hidden features that cannot be directly accessed.}

This semantic rate distortion model has attracted considerable interest after being proposed, and some subsequent problems have been studied.  
Stavrou \textit{et al.} \cite{Stavrou-ba-2023} propose a Blahut-Ahrimoto based algorithm to solve the optimization problem of semantic rate distortion, which has two distortion constraints rather than only one constraint in classic rate distortion model. Wang \textit{et al.} \cite{Wang-semanticseperation-2022} { work} on an estimation-compression separated system for  semantic rate distortion model. 
{ Moreover, Shi \textit{et al.} \cite{Shi-errorexponent-2023} further extend the model in \cite{Liu-semantic-2022} to scenarios involving the joint source-channel coding (JSCC) and the separated source-channel coding (SSCC). Specifically, they derived the error exponent for the excess-distortion probability, i.e., the exponential decay rate of the event where the distortion violates the constraint.}
In these works \cite{Liu-semantic-2022,Stavrou-ba-2023,Wang-semanticseperation-2022,Shi-errorexponent-2023}, a tradeoff between the reconstruction distortion of the observable source and that of the unobservable source under fixed coding rate is revealed, which coincides with the practice in learning-based compression and communication. 

Nevertheless, these works following the semantic rate distortion model only focus on the infinite blocklength regime, which does not align with the fact that the practical scenarios often operate in the finite blocklength regime. 
{ In fact, the coding and communication problem for semantic source in the finite blocklength regime, especially under the JSCC framework, has attracted significant attention in recent years. Specifically, if only an indirect or observable source exists, the nonasymptotic bounds are studied in \cite{Kostina-fixed-12,Kostina-isc-2016, Tan21}. The channel coding rate has been studied in the literature  \cite{Polyanskiy-10,yang2014quasi,Albert14,Tan15,Collins2019coherent,lancho2019single,Tan15}. 
Besides, \cite{Ulger-23} studies the bounds on the excess-distortion probability of the semantic source coding problems\footnote{{ Note that in the paper, the achievability bound (or inner bound) and the converse (or outer bound) indicate the lower and the upper bound on excess-distortion probability, respectively.}}.	
However, when communication is also taken into account, the fundamental limits of JSCC for semantic sources in the finite blocklength regime remain an open problem.  }

{In this paper, we revisit the JSCC with a semantic source and an observable source as shown in Fig. \ref{Sys}. 
	We aim to establish computable and tight bounds on the probability of the joint excess-distortion event in the finite  regime.  Specifically,  the joint excess-distortion event represents the scenario that the reconstructed semantic information or the reconstructed observable source at the receiver end exceeds the given distortion level $\mathsf{D}_s$ and $\mathsf{D}_x$, i.e., $\{\mathsf{d}_{s}(S,\hat{S})> \mathsf{D}_s \cup \mathsf{d}_{x}(X,\hat{X})> \mathsf{D}_x\}$. 
	Although the proposed model can be viewed as a combination of the models in \cite{Kostina-isc-2016} and \cite{Kostina-lossyJSCC-13} by incorporating the data transmitting problem of two correlated sources over a noisy channel under the JSCC framework, this work is not a  straightforward extension of existing works. 
	The key distinction lies in precisely bounding the joint excess-distortion probability induced by two correlated sources, which introduces an inherent coupling between the two distortion constraints. 
	This setting differs fundamentally from prior works based on a single-distortion constraint in \cite{Kostina-lossyJSCC-13} or decoupled constraints in \cite{Zhou-19}. 

The new analysis faces significant technical challenges arising from the problem structure. First, the dependence between the two distortion events makes the joint excess-distortion probability hard to be precisely characterized by existing tools. In contrast to previous studies, where the analysis relies on a single distortion constraint  \cite{Kostina-lossyJSCC-13} or a decomposable event structure \cite{Zhou-19}, the correlation between distortion events prevents us  from treating the joint event as separate excess-distortion events. As a result, classical tools cannot be directly applied, and marginal-based bounds become loose.
Second, the intrinsic tradeoff between the reconstruction of $S$ and $X$ prevents separate optimization under a fixed coding strategy. In \cite{Ulger-23}, a similar source coding problem of hierarchical sources was studied. However, no computable optimal bounds have been established, since an effective rate allocation strategy between encoding $S$ and $X$ has not yet been established. Hence, establishing computable and tight bounds for the joint excess-distortion probability requires new analytical approaches.
The technical challenges will be explicitly explained in Sec. \ref{Sec_Converse_general} and \ref{Sec_Achievability_general}.

}

{ \subsection{Main Contributions}
Under the joint excess-distortion constraint, we derive the finite blocklength bounds for JSCC with hierarchical sources. The main contributions are summarized as follows.
\begin{itemize}
\item[1)] We establish nonasymptotic achievability and converse bounds on the excess-distortion probability for one-shot JSCC with hierarchical sources. 
To handle the joint excess-distortion event induced by two correlated sources, we introduce a novel $\mathsf{d}()$-functional distortion relaxation. 
This approach enables the derivation of computable tight bounds in the finite blocklength regime, where standard techniques based on a single distortion event or the decomposable distortion events fail to provide tractable tight bounds.

\item[2)] Gaussian approximations are given for the hierarchical sources under the joint excess-distortion constraint. 
The approximation is shown to be optimal for the transmission of a Gaussian memoryless source (GMS) with mean-square error distortion over an additive white Gaussian noise (AWGN) channel. 
Technically, the achievability analysis is based on a two-layer superposition coding scheme and the optimization of a sum of correlated $Q$-functions, explicitly capturing the dependence between semantic and observable distortions. 	
In particular, such a characterization has not been developed in prior works, and reveals how correlation between distortion measures affects second-order performance.

\item[3)] For the transmission of a GMS with mean-square error distortion over an  AWGN channel, we obtain explicit and analytically tractable tight bounds. 
In particular, the joint excess-distortion analysis provides a novel geometric structure involving three correlated spherical regions, extending classical two-region analysis for single excess-distortion in \cite{Kostina-lossyJSCC-13}. 
These results provide an explicit characterization of the joint excess-distortion event of the transmission of GMS over AWGN channels under correlated distortion constraints.

\end{itemize}

We next clarify the contributions beyond the closely related existing studies \cite{Kostina-lossyJSCC-13,Zhou-19,Ulger-23,Yang25}. 
\emph{First},  compared to \cite{Kostina-lossyJSCC-13}, our work extends the analysis of excess-distortion probability from the single excess-distortion setting to a joint excess-distortion setting with correlated sources via a $\mathsf{d}()$-functional distortion relaxation. Moreover, for the transmission of a GMS over an AWGN channel, this paper derives the tractable nonasymptotic bounds by providing a novel geometric structure involving three correlated spherical regions, extending classical two-spherical regions analysis for single excess-distortion probability in \cite{Kostina-lossyJSCC-13}. 
\emph{Second}, compared to \cite{Zhou-19}, this paper provides a general approach to analyzing the joint excess-distortion event with correlated sources. In \cite{Zhou-19}, the excess-distortion probability is analyzed via a decoupling-based approach, which separates two events into conditionally independent components. In contrast, our approach directly characterizes the joint excess-distortion event without relying on a decoupling argument,  thus extending the analysis tools for the case where the sources are correlated.
\emph{Third}, although \cite{Ulger-23} studies hierarchical source coding problems with related structures, it does not provide computable and tight finite blocklength bounds. By contrast, the proposed $\mathsf{d}()$-functional distortion relaxation proposed in this paper enables the derivation of explicit and computable bounds, thereby overcoming a key limitation of the approaches in \cite{Ulger-23}.
\emph{Fourth}, compared to \cite{Yang25}, which studies indirect lossy source coding within the SSCC framework when both excess distortion events are active, this paper provides a complete characterization of the joint excess-distortion probability. Our work removes the restriction that both excess distortion events are active, and analyzes the union event, allowing for different active events across regimes.}

\subsection{ Related Work}

{ 
	The semantic source coding model was first introduced by Liu et al.~\cite{Liu-semantic-2022}, where the semantic rate-distortion (SRD) function is formulated by combining classical rate-distortion (RD) and indirect source coding (ISC) frameworks. 	
	Several subsequent works have further investigated the performance of Semantic Communications  \cite{Wang-semanticseperation-2022,Stavrou-ba-2023,Shi-errorexponent-2023}. 
	Note that these works focus on asymptotic regimes.
	If there exists only the indirect or the observable source, the semantic rate distortion model is a generalized model of several classical models. More specifically, within the SSCC framework, the channel and source coding rates in the finite blocklength regime have been studied in the literature \cite{Polyanskiy-10,yang2014quasi,Albert14,Collins2019coherent,lancho2019single,Tan15,Kostina-isc-2016,Kostina-fixed-12,Tan21}. 	
	In particular, \cite{Kostina-isc-2016} and \cite{Kostina-fixed-12} establish nonasymptotic bounds for indirect and observable source coding problems, respectively. These results provide tight characterizations under a single distortion constraint in the finite blocklength regime.
	A closely related work is \cite{Zhou-19}, that studies a lossy source successive refinement problem. The key idea in \cite{Zhou-19} is to decompose the excess-distortion event into multiple sub-events, each characterized via a tilted-information quantity, and then apply a union bound over these events. This decoupling-based approach enables tractable analysis but relies on the separability of distortion events.
	Regarding the JSCC framework, Kostina and Verd\'{u} \cite{Kostina-lossyJSCC-13} study the nonasymptotic bounds of only the observable source, which shows that the JSCC framework performs better than the SSCC framework in terms of excess-distortion probability when the blocklength is finite. 
	However, the nonasymptotic bounds with hierarchical sources have not been considered in the aforementioned works.		
	
	For the hierarchical sources model, more recently, \cite{Ulger-23} studied a lossy observable source coding problem with hierarchical sources in the finite regime. However, the analysis therein effectively reduces the joint distortion structure by focusing on cases where both distortion constraints are simultaneously active, and no computable bound is given.  Similarly, \cite{Yang25} considers indirect source coding with the observed source reconstruction within the SSCC framework but its analysis of the union excess-distortion probability is also limited to the situation where all constraints are active, leaving the general case unresolved.  
 
In summary, the nonasymptotic bounds for lossy indirect and classical source coding in the finite blocklength regime have been well studied, and the current JSCC analyses typically focus on a single distortion measure.  
As mentioned in previous subsections, the JSCC model involving two correlated distortion constraints lacks an effective approach to handle the resulting union event. For one thing, the techniques used for source coding \cite{Kostina-fixed-12,Zhou-19,Ulger-23}  cannot be directly extended to the JSCC model, since transmitting the codeword indices associated with the semantic sources is itself a nontrivial problem. 
Meanwhile,  the study of JSCC with hierarchical sources is nontrivial, since it requires characterizing the joint excess-distortion probability for correlated sources.  Moreover, the optimal coding scheme cannot be designed by simply combining lossy semantic and observable coding schemes, since the SRD framework enforces a trade‑off between the semantic source and observable source reconstruction.
Consequently, no rigorous nonasymptotic bound has been established for the two-distortion JSCC problem.
}

\subsection{Organization of the Paper}
{ The rest of the paper is organized as follows. }In Section \ref{prelimanary}, the system model and important related work is presented. In Section \ref{Sec_Converse_general}, we establish converse bounds based on the proposed $\mathsf{d}()$-functional relaxation method. In Section \ref{Sec_Achievability_general}, achievability bounds and the corresponding coding schemes are proposed. In Section \ref{Sec_Gaussian_Apprx}, Gaussian approximation bounds are studied for finite blocklength transmission. In Section \ref{Sec_GMS_AWGN}, we specify the achievable and converse bounds to the transmission of a GMS overa an AWGN channel.

\section{Problem Formulation and Important Notations} \label{prelimanary}
\subsection{Notations}
The following definitions and notations are used in this paper. We use capital letters to denote random variables, small letters to denote the realization of random variables, and { calligraphic letters to denote the alphabets of random variables.} For example, for random variable $X$, its realization $x$ takes values in $\mathcal{X}$. We use $[a,b]$ to denote the set of real numbers $\{x:a\leq x \leq b\}$, and $[a:b]$ to denote the set of integers  $\{x:a\leq x \leq b,\, x\in \mathbb{Z}\}$. We denote the Euclidean norm by $\|\cdot\|$, i.e.  for sequence $x^n=(x_1,x_2,\dots,x_n)$ , we denote $\|x^n\|^2 = x^2_1 +\ldots+x^2_n$. { Denote by} $W^{\ell}_{m}$ a noncentral chi-square distributed random variable with $\ell$ degrees of freedom and non-centrality parameter $m$, and $f_{W_{m}^{\ell}}$ denotes its probability density function. The distribution of a Gaussian random variable with mean $\mu$ and variance $\sigma^2$ is denoted by { $\mathcal{N}(\mu,\sigma^2)$.} 
$\log(\cdot)$ denotes the natural logarithm. $|a|^+$ stands for $\max\{a,0\}$.

\subsection{System Model}
\begin{figure}[h]
		\centering
		\includestandalone[width=0.6\textwidth]{ChannelModel}%
		\caption{An $(n,k,\mathsf{D}_s,\mathsf{D}_x)$ lossy JSCC system for semantic sources}
		\label{Sys}
	\end{figure}

{  In this paper, we consider a single-shot JSCC with hierarchical sources as illustrated in Fig. \ref{Sys}. The system involves a pair of correlated sources $(S,X)$ with joint distribution $P_{SX}$. 
The semantic information is embodied in the intrinsic state $S$, which is unobservable but can  be inferred only from the observable data source $X$.
The encoder only accesses $X$ and generates a codeword $Y$, and $Y$ is transmitted over a given channel according to the transition probability $P_{Z|Y}$.
After receiving the channel output $Z$, the decoder has two tasks: reproducing the intrinsic state block  as $\hat{S}$ under a state distortion measure $\mathsf{d}_s(s,\hat{s})$, and reproducing the extrinsic observation block as $\hat{X}$ under an observation distortion measure $\mathsf{d}_x(x,\hat{x})$.
The corresponding distortion measure functions are defined as 
\begin{align}
	\mathsf{d}_s(s,\hat{s}): \mathcal{S}\times\hat{\mathcal{S}}\rightarrow [0, +\infty),\nonumber\\
	\mathsf{d}_x(x,\hat{x}): \mathcal{X}\times\hat{\mathcal{X}}\rightarrow [0, +\infty). \nonumber
\end{align}  }
For any given positive real number tuple $(\mathsf{D}_s,\mathsf{D}_x)$, we denote the event of excess distortion as 
\begin{align*}
	\mathcal{E}_{(\mathsf{D}_s,\mathsf{D}_x)}=\{(s,x,\hat{s},\hat{x})
	\in \mathcal{S} \times \mathcal{X} \times \hat{\mathcal{S}} \times \hat{\mathcal{X}}:
	\mathsf{d}_s(s,\hat{s})>\mathsf{D}_s \cup \mathsf{d}_x(x,\hat{x})>\mathsf{D}_x\}.
\end{align*}

{ A tuple $(\mathsf{D}_s,\mathsf{D}_x,\epsilon)$ is said to be achievable for the source and channel with distributions $P_{SX}$ and $P_{Z|Y}$ if, for any $\epsilon > 0$, there exist the following mappings:	
\begin{itemize}
	\item An encoding mapping (possibly random) $f_{{Y}|{X}}:\mathcal{X} \rightarrow \mathcal{Y}$: 
	the encoder directly maps the source $X$ to the channel input $Y$;
	
	\item Two decoding functions $g_{\hat{{S}}|{Z}}:\mathcal{Z} \rightarrow \hat{\mathcal{S}}$ and $g_{\hat{{X}}|{Z}}:\mathcal{Z} \rightarrow \hat{\mathcal{X}}$: the decoder reconstructs the semantic information $\hat{S}$ and the observable data $\hat{X}$ respectively, such that $P\{\mathcal{E}_{(\mathsf{D}_s,\mathsf{D}_x)}\} \leq \epsilon$. 
\end{itemize}

The main goal of this paper is to characterize the achievability and the converse bounds for the JSCC corresponding to the tuple $(\mathsf{D}_s,\mathsf{D}_x,\epsilon)$ for the given source distribution $P_{SX}$ and channel distributions $P_{Z|Y}$. }

\subsection{Important Definitions}
{ For a given channel $P_{Z|Y}$ and a cost function $\mathsf{c}: \mathcal{Y} \mapsto [0,+\infty)$, the channel capacity is defined as 
\begin{align}
	C(\alpha)\triangleq \sup_{P_Y:  \mathbb{E}[\mathsf{c}(Y)] \leq \alpha} I(Y;Z). \label{Channel_capacity}
\end{align}
where $\mathbb{E}[\mathsf{c}(Y)] \leq \alpha$ is the maximal cost constraint.}
 For a given semantic source $(S,X)$ with $P_{SX}$ over $\mathcal{S}\times \mathcal{X}$, the reconstruction alphabet  $\hat{\mathcal{S}}\times \hat{\mathcal{X}}$, and the distortion measures $\mathsf{d}_s$ and $\mathsf{d}_x$, the semantic rate distortion function is defined as follows.  
 \begin{definition}[Semantic rate distortion (SRD) function \cite{Liu-semantic-2022}]\label{Def_SRD}
 	The semantic rate distortion (SRD) function is defined as
 	\begin{align}
 		\bar{R}_{S,X}(\mathsf{D}_s,\mathsf{D}_x) \triangleq & \inf_{P_{\hat{S}\hat{X}|X} } I(X; \hat{S}, \hat{X}) \label{Equ_R}\\
 		\rm{s.t.} \quad\mathbb{E}[\mathsf{d}_s({S}, {\hat{S}})] \leq &  \mathsf{D}_s, \label{Equ_d_cst1} \\
 		 \mathbb{E}[\mathsf{d}_x({X}, {\hat{X}})] \leq &  \mathsf{D}_x \label{Equ_d_cst2a}.
 	\end{align}
\end{definition}
 	 Given a source distribution $P_X$ and the distortion measure $\mathsf{d}_x({X}, {\hat{X}})$, the well-known direct rate distortion (RD) function is defined as  	
 		\begin{align}
 			R_{RD}(\mathsf{D}_x) &\triangleq \inf_{P_{\hat{X}|X} } I(X;  \hat{X}) \\
 			\rm{s.t.} \quad \mathbb{E}[\mathsf{d}_x({X}, {\hat{X}})] &\leq \mathsf{D}_x.  \label{Def_R_RD}
 		\end{align}
 	Similarly, with the indirect source and distortion measure $\mathsf{d}_s({S}, {\hat{S}})$, the indirect source coding (ISC) function is given as 
 		\begin{align}
 			R_{ISC}(\mathsf{D}_s) &\triangleq \inf_{P_{\hat{S}|X} } I(X; \hat{S}) \label{Def_R_ISC}\\
 			\rm{s.t.} \quad \mathbb{E}[\mathsf{d}_s({S}, {\hat{S}})] &\leq \mathsf{D}_s.
 		\end{align}
 	We can see the { SRD} function is a compound of the traditional rate distortion (RD) function and indirect source coding (ISC) function.

{  Compared with the RD and ISC functions, the setup of SRD function leads to a special phenomenon which has not been observed. Indeed, there is a tradeoff between the performance of the semantic source reconstruction and that of the observable source reconstruction when the coding rate is fixed. Therefore, the best way to encode the observable source may be not the best way to preserve the semantic information, which coincides with the practical scenarios in learning-based semantic communication \cite{Zhang-imagesemantic-2023} and other conventional works in representation learning \cite{Liu-representinglearning-2021}. For example, two reconstructed images with the same average peak signal-to-noise ratio (PSNR) level may have different performances in the same classification task.}
 
Next, we introduce an essential idea called \textit{functional tilted information}. For any function 
 \begin{align}
 	\mathsf{d}(\mathsf{d}_s,\mathsf{d}_x): [0, +\infty)\times [0, +\infty) \rightarrow [0, +\infty),
 \end{align}
 where the function $\mathsf{d}(\cdot)$ is monotonically non-decreasing with $\mathsf{d}_s$ and $\mathsf{d}_x$ respectively, we first define the $\mathsf{d}()$-relaxed rate distortion function as follows. 
\begin{definition}[$\mathsf{d}()$-relaxed rate distortion function]\label{relaxed rate distortion}
	For any given $\mathsf{d}()$ function, the $\mathsf{d}()$ functional-relaxed rate distortion function $R_{S,X}$ is defined as 
	{ \begin{align}
		R_{S,X}(\mathsf{D}_s,\mathsf{D}_x) \triangleq &  \inf_{P_{\hat{S}\hat{X}|X} } I(X; \hat{S}, \hat{X}) \label{Equ_R_d}\\
		\rm{s.t.} \quad  \mathsf{d}(\mathbb{E}[\mathsf{d}_s(S, \hat{S})],\mathbb{E}[\mathsf{d}_x(X, \hat{X})]) \leq & \mathsf{d}(\mathsf{D}_s,\mathsf{D}_x) .\label{Equ_d_cst2}
	\end{align}}
\end{definition}

Note that Definition \ref{relaxed rate distortion}  not only relaxes the excess distortion event but also relaxes the distortion constraints. { Thus, all joint conditional distributions $P_{\hat{S}\hat{X}|X}$ within the feasible region of Definition \ref{relaxed rate distortion} are within the feasible region of the original semantic rate distortion function defined in Definition \ref{Def_SRD} as well.  }
More details can be found in Lemma \ref{Lem_d_function}, which proves that the event of $\mathsf{d}(\mathsf{d}_s(s,  \hat{s}),\mathsf{d}_x(x,  \hat{x}))		
> \mathsf{d}(\mathsf{D}_s,\mathsf{D}_x)$ is a sufficient but not necessary condition of  $\mathcal{E}_{(\mathsf{D}_s,\mathsf{D}_x)}$. 
Comparing \eqref{Equ_d_cst2} with \eqref{Equ_d_cst2a}, it can be seen that, different from the SRD function, the joint criterion cf. \eqref{Equ_d_cst2} is considered, which is consistent with recent works in the second-order literature \cite{Kostina-isc-2016,Zhou-19,Ulger-23,Yang25}.
In this paper, we  consider a special type of $\mathsf{d}()$ functions which are linear, defined as 
	\begin{align}
		\mathsf{d}_L(\mathsf{d}_s,\mathsf{d}_x) \triangleq  (1-\alpha)\mathsf{d}_s + \alpha \mathsf{d}_x, \label{Def_linear_d}
	\end{align}
	 where $\alpha \in [0,1]$.

Based on the $\mathsf{d}()$-relaxed rate distortion function, we further define the tilted information of $\mathsf{d} (\mathsf{d}_s,\mathsf{d}_x)$ as follows. 
\begin{definition}[$\mathsf{d}(\mathsf{d}_s,\mathsf{d}_x)$-tilted information] \label{d-tilted information}
The $\mathsf{d}(\mathsf{d}_s,\mathsf{d}_x)$-tilted information $\jmath_{S,X}(s, x)$ is defined as
\begin{align}
\jmath_{S,X}(s, x)
	\triangleq \log \frac{1}{ \mathbb{E}[ \exp\{\lambda^{\star} \mathsf{d}(\mathsf{D}_s,\mathsf{D}_x)   - \lambda^{\star} \mathsf{d}(\mathsf{d}_s({s}, \hat{S}^{\star}),\mathsf{d}_x({x}, \hat{X}^{\star}))  \}  ] }, \label{Equ_Def_j_ori} 
\end{align}
where the expected value is with respect to some distribution $P_{\hat{S}^{\star}\hat{X}^{\star}}$ which achieves the rate distortion function, i.e., the unconditional distribution of the reproduction random variable that achieves the infimum in \eqref{Equ_R_d}, and 
\begin{align}
	\lambda^{\star} \triangleq & - \frac{\partial R_{S,X}(\mathsf{D}_s,\mathsf{D}_x) }{ \partial\mathsf{d}(\mathsf{D}_s,\mathsf{D}_x)}. \label{Def_lamdas}
\end{align}		
\end{definition}

Intuitively, the $\mathsf{d}(\mathsf{d}_s,\mathsf{d}_x)$-tilted information denotes how many bits are needed to encode a source output symbol $x$ when the expectation of $\mathsf{d}(\mathsf{d}_s,\mathsf{d}_x)$ does not exceed the distortion constraint $\mathsf{d}(\mathsf{D}_s,\mathsf{D}_x)$. The parameter $\lambda^{\star}$ can be regarded as the Lagrange multiplier of constraint $\mathsf{d}(\mathsf{d}_s({s}, {\hat{s}}),\mathsf{d}_x({x}, {\hat{x}}))-\mathsf{d}(\mathsf{D}_s,\mathsf{D}_x)>0$. 
The properties of $\mathsf{d}$($\mathsf{d}_s$,$\mathsf{d}_x$)-tilted information are given as follows \cite[Lemma 1.4]{Csisza1974} and \cite[Definition 6]{Kostina-lossyJSCC-13}: 
	\begin{align}
		&{ \jmath_{S,X}(s, x)}
			\triangleq\imath_{X; \hat{X}, \hat{S}} (x; \hat{x}, \hat{s})
				+\lambda^{\star} \big(\mathsf{d}(\mathsf{d}_s({s}, {\hat{s}}),\mathsf{d}_x({x}, {\hat{x}}))-\mathsf{d}(\mathsf{D}_s,\mathsf{D}_x)\big), \label{Equ_Def_j} \\
		&\mathbb{E}[{ \jmath_{S,X}(s, x)}]=R_{S,X}(\mathsf{D}_s,\mathsf{D}_x), \label{Equ_Def_Ej} \\
		&\mathbb{E} \left[ \exp\left({ \jmath_{S,X}(s, x)}+\lambda^{\star} \big(\mathsf{d}(\mathsf{D}_s,\mathsf{D}_x)-\mathsf{d}(\mathsf{d}_s({s}, {\hat{s}}),\mathsf{d}_x({x}, {\hat{x}}))\big) \right)   \right]\leq 1, \label{Equ_Def_Ejd}
	\end{align}
	where { the information density 
	\begin{align}
		\imath_{X; \hat{X}, \hat{S}} (x; \hat{x}, \hat{s})
		\triangleq & \log \frac{ {\textnormal{d}} P_{ \hat{X}, \hat{S}|X=x} }{ {\textnormal{d}} P_{\hat{X}, \hat{S}}} (\hat{x}, \hat{s})= \log \frac{ {\textnormal{d}} P_{ X| \hat{X}=\hat{x}, \hat{S}=\hat{s}} }{ {\textnormal{d}} P_{X}} (x),
	\end{align}
	for any given joint distribution $P_{\hat{X}\hat{S}X}$. }

\begin{definition}[Information density function for channel capacity]
	The information density $\imath_{Y;Z} (Y; Z)$ is defined as 
	\begin{align}
		\imath_{Y;Z} (y; z)
		\triangleq & \log \frac{ {\textnormal{d}} P_{Z|Y=y} }{ {\textnormal{d}} P_{Z}} (z),
	\end{align}
	for a given channel transition probability $P_{Z|Y}$.
\end{definition}

\begin{definition}[Distortion ball]
The distortion $\mathsf{D}_x$-ball around $x$ is defined as
\begin{align}
	B_{\mathsf{D}_x}(x) \triangleq \{\hat{x}\in \hat{\mathcal{X}}: \mathsf{d}_x(x, \hat{x})\leq \mathsf{D}_x \}.
\end{align}
Similarly, the distortion $\mathsf{D}_s$-ball around $s$ is defined as
\begin{align}
	B_{\mathsf{D}_s}(s) \triangleq \{\hat{s}\in \hat{\mathcal{S}}: \mathsf{d}_s(s, \hat{s})\leq \mathsf{D}_s \}.
\end{align}
\end{definition}

\section{General Converse Bounds} \label{Sec_Converse_general}
{ In this section, we establish general converse bounds on the joint excess-distortion probability.} First, in subsection \ref{converse_preliminary}, we will give an intuitive introduction to the excess-distortion probability for hierarchical sources in the JSCC framework, and explain why it is not a straightforward extension of existing results. Then in subsection \ref{converse bounds}, we will derive converse bounds with a general form, which are derived by the proposed relaxed functional distortion constraint. 

{ \subsection{$\mathsf{d}()$-functional Distortion Relaxation Approach}\label{converse_preliminary}}

Characterizing the converse bounds of $\epsilon$ is essentially to find the lower bounds for the minimum excess-distortion probability $P\{\mathcal{E}_{(\mathsf{D}_s,\mathsf{D}_x)}\}$. 
However, the unique features of hierarchical sources make it substantially different from conventional works in this context, especially because the low-complexity numerical evaluation of the nonasymptotic upper and lower bounds becomes { considerably complex}.

In general, the difficulties of establishing converse bounds for hierarchical sources are two-fold.
First, the excess distortion event in the hierarchical sources model corresponds to the union of two events: either the distortion in reconstructing the semantic information or the distortion in reconstructing the observable source exceeds its respective threshold. However, the tight bound for the excess-distortion probability of the union event is still unknown.  
{ Second, two random variables $W_x\triangleq \mathsf{d}_x(X,\hat{X})$ and $W_s\triangleq \mathsf{d}_s(S,\hat{S})$ are correlated, since $S$ and $X$ are correlated. 
Hence, $P\{\mathcal{E}_{(\mathsf{D}_s,\mathsf{D}_x)}\}$ equals the probability density measure of the excess distortion region, namely the shadow area in Fig. \ref{Fig_region}.
This probability is even hard to be numerically computed, since the probability density function of $(W_s,W_x)$ varies with the choice of codebooks.  }
For example, in the transmission of a GMS over an AWGN channel which will be discussed later in Section \ref{Sec_Gaussian_Apprx}, the accurate excess-distortion probability $P\{\mathcal{E}_{(\mathsf{D}_s,\mathsf{D}_x)}\}$ is the cumulative probability of two joint chi-square distributed random variables, whose computation is notably complex.  

	\begin{figure}[!htpb]
	\centering  
	\includegraphics[width=0.45\textwidth]{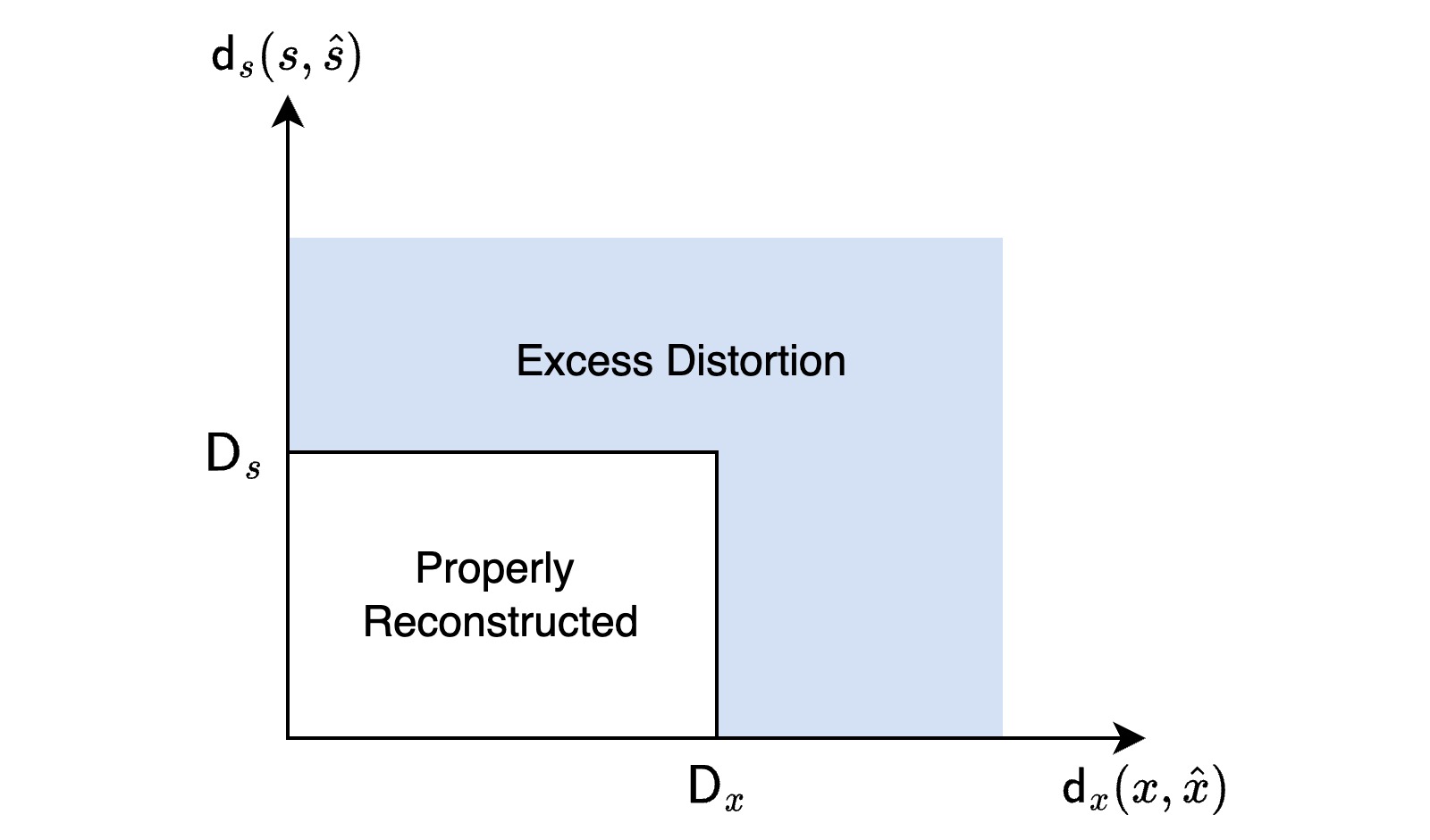} 
	\caption{The geometry of the excess-distortion probability calculation. }
	\label{Fig_region} 
	\end{figure}

A  { straightforward} bound can be obtained trivially by ignoring one of the distortion constraints, i.e. $\mathsf{D}_x=\infty$ or $\mathsf{D}_x=\infty$:
\begin{align}
	P\{\mathcal{E}_{(\mathsf{D}_s,\mathsf{D}_x)}\} \geq \max \{P\{\mathcal{E}_{(\infty,\mathsf{D}_x)}\},P\{\mathcal{E}_{(\mathsf{D}_s,\infty)}\}\}, \label{Equ_converse_loose}
\end{align}
where the lower bound of $P\{\mathcal{E}_{(\infty,\mathsf{D}_x)}\}$ and $P\{\mathcal{E}_{(\mathsf{D}_s,\infty)}\}$ can be obtained following similar proofs in \cite[Theorem 1]{Kostina-lossyJSCC-13}. 
Though the bound \eqref{Equ_converse_loose} can provide a computable converse bound for the excess-distortion probability, a natural question arises whether better bounds can be obtained. 
In the following, we propose a $\mathsf{d}()$-functional distortion relaxation approach to investigate tighter converse bounds for the union excess-distortion probability.
\begin{lemma}\label{Lem_d_function}
	Consider a $\mathsf{d}()$-functional relaxed excess distortion event defined as follows,
	\begin{align}
		\tilde{\mathcal{E}}_{(\mathsf{D}_s,\mathsf{D}_x)}=&\{(s,x,\hat{s},\hat{x}) 
		\in \mathcal{S} \times \mathcal{X} \times \hat{\mathcal{S}} \times \hat{\mathcal{X}}:
		{\mathsf{d}(\mathsf{d}_s(s,  \hat{s}),\mathsf{d}_x(x,  \hat{x}))		
			> \mathsf{d}(\mathsf{D}_s,\mathsf{D}_x) } \}.
	\end{align}
	where $\mathsf{d}(\mathsf{d}_s,\mathsf{d}_x)$ is monotonically non-decreasing with $\mathsf{d}_s$ and $\mathsf{d}_x$, and is called the $\mathsf{d}()$-functional distortion relaxation.
	Then, the probabilities of excess distortion event $\tilde{\mathcal{E}}_{(\mathsf{D}_s,\mathsf{D}_x)}$ and the original event $\mathcal{E}_{(\mathsf{D}_s,\mathsf{D}_x)}$ satisfy 
	\begin{align}
		P\{\tilde{\mathcal{E}}_{(\mathsf{D}_s,\mathsf{D}_x)}\} \leq  P\{\mathcal{E}_{(\mathsf{D}_s,\mathsf{D}_x)}\}, \label{sufficient}
	\end{align}
\end{lemma}
\begin{IEEEproof}
	Obviously, it can be seen that $\tilde{\mathcal{E}}_{(\mathsf{D}_s,\mathsf{D}_x)}$ is a sufficient but not necessary condition of $\mathcal{E}_{(\mathsf{D}_s,\mathsf{D}_x)}$. Thus, \eqref{sufficient} is obtained. 
\end{IEEEproof}

Therefore, according to this lemma, the $\mathsf{d}()$-functional distortion relaxation provides a method to derive a converse bound of the original excess-distortion probability.	
 Moreover, if the probability of the relaxed region is smaller, the relaxed excess-distortion probability will be tighter. However, a more complicated $\mathsf{d}(\mathsf{d}_s,\mathsf{d}_x)$ will bring more complicated derivations of the bounds.

\subsection{Converse Bounds with General and Linear Relaxation Functions}\label{converse bounds}
According to the relaxed excess distortion event, we derive converse bounds on $\epsilon$ with a general form of function $\mathsf{d}(\mathsf{d}_s,\mathsf{d}_x)$. A nonasymptotic converse bound can be stated as follows. 

		
		
		\begin{theorem}[Converse bound with general $\mathsf{d}(\mathsf{d}_s,\mathsf{d}_x)$ function] \label{Thm_converse}
			For any achievable tuple  $(\mathsf{D}_s, \mathsf{D}_x, \epsilon)$, the following inequality must hold:
			\begin{align}
				\epsilon \geq &\inf_{P_{Y|X} } 
				\sup_{\gamma \geq 0}  
				\left\{  \sup_{P_{{Z}} } \mathbb{P} \big[ 
				\jmath_{S,X}(S, X) -\imath_{Y;{Z}} (Y; Z)\geq \gamma \big]
				-\exp(-\gamma) \right\}, \label{Equ_Converse}
			\end{align}
			for any $\mathsf{d}(\mathsf{d}_s,\mathsf{d}_x)$ which is monotonically non-decreasing with $\mathsf{d}_s$ and $\mathsf{d}_x$. Here the tilted information $\jmath_{S,X}(S,X)$ is defined as \eqref{Equ_Def_j}, and random variable tuple $(S,X,Y,Z,\hat{X},\hat{S})$ follows a distribution 
  \begin{align}
      P_{SXYZ\hat{X}\hat{S}}=P_{\hat{X}\hat{S}|Z}P_{Z|Y}P_{Y|X}{P_{X|S}P_{S}}.
  \end{align} 
		\end{theorem}
The proof of Theorem \ref{Thm_converse} is provided in Appendix \ref{appendix-proof of general converse}.  

The lower bound of excess-distortion probability in Theorem \ref{Thm_converse} indicates that an excess distortion event is likely to occur when the number of bits for encoding a source symbol is larger than the number of bits that a channel symbol can convey. Specifically, the probability of satisfying the distortion constraint is smaller than $\exp(-\gamma)$ when the gap between the tilted information of source symbol and the information density of the channel is larger than $\gamma$. But different from the result in  \cite[Theorem 1]{Kostina-lossyJSCC-13}, the tilted information  used in \eqref{Equ_Converse} is the $\mathsf{d}(\mathsf{d}_s,\mathsf{d}_x)$-tilted information defined in Definition \ref{d-tilted information}, which is a balanced tilted information between the observable source and unobservable source rather than choosing either one merely. 

{ 
\begin{Remark}
 	There are two main technical difficulties in building the nonasymptotic bounds.	
	First, classical approaches of analyzing a single distortion \cite[Theorem 1]{Kostina-lossyJSCC-13} or a decomposable event structure \cite{Zhou-19} can not be extended to analyze the union of correlated excess-distortion events. 
	Furthermore, the dependence between the two distortion events makes the joint excess-distortion probability hard to be precisely characterized by existing analysis tools. Actually, even for the Gaussian sources, the joint excess-distortion probability requires integrating correlated chi-square random variables over a high-dimensional spherical region, making it intractable. 
	Second, the tilted-information approach used in previous works \cite{Kostina-lossyJSCC-13, Zhou-19} becomes inadequate under multiple distortion constraints, since an inactive constraint leads to vanishing Lagrange multiplier and hence ignores the corresponding distortion event. This results in a mismatch in the finite blocklength regime, where the excess distortion event may still occur even under nominally inactive constraints. 	
	To address these difficulties, we introduce a $\mathsf{d}()$-functional distortion relaxation that transforms the correlated distortion constraints into a tractable form. 
	This approach further enables the derivation of computable bounds and the optimal Gaussian approximations.

\end{Remark}
}

 If we further specify the function $\mathsf{d}()$ as the linear function $\mathsf{d}_L(\mathsf{d}_s,\mathsf{d}_x)$ in \eqref{Def_linear_d},  the lower bound of $\epsilon$ can be tightened by adjusting the parameters of the function in the following corollary.
\begin{Corollary}\label{L-converse bound}
    For the linear function $\mathsf{d}_L(\mathsf{d}_s,\mathsf{d}_x)$ defined in \eqref{Def_linear_d}, the converse bound in Theorem \ref{Thm_converse} can be further tightened as 
			\begin{align}
				\epsilon \geq &\inf_{P_{Y|X} } \sup_{\alpha \in [0,1]}
				\sup_{\gamma \geq 0}  
				\left\{  \sup_{P_{{Z}} } \mathbb{P} \big[ 
				\jmath_{S,X}(S, X) -\imath_{Y;{Z} } (Y; Z)\geq \gamma \big]
				-\exp(-\gamma) \right\},\label{L-converse bound-general}
			\end{align}
\end{Corollary}
\begin{IEEEproof}
	The proof is similar to that of Theorem \ref{Thm_converse}, but replacing $\mathsf{d}(\mathsf{d}_s({s}, {\hat{s}}),\mathsf{d}_x({x}, {\hat{x}}))>\mathsf{d}(\mathsf{D}_s,\mathsf{D}_x)$ by $\mathsf{d}_L(\mathsf{d}_s(s,  \hat{s}),\mathsf{d}_x(x,  \hat{x}))>\mathsf{d}_L(\mathsf{D}_s,\mathsf{D}_x)$,
	and $\mathsf{d}(\mathsf{d}_s({s}, {\hat{s}}),\mathsf{d}_x({x}, {\hat{x}}))\leq \mathsf{d}(\mathsf{D}_s,\mathsf{D}_x)$ by $\mathsf{d}_L(\mathsf{d}_s(s,  \hat{s}),\mathsf{d}_x(x,  \hat{x}))		
	\leq \mathsf{d}_L(\mathsf{D}_s,\mathsf{D}_x)$.
	Correspondingly, $\hat{s}\notin B_{\mathsf{D}_s}(s)$ and $\hat{x}\notin B_{\mathsf{D}_x}(x)$ can be replaced by the $(\hat{s}, \hat{x}) \notin {B}_{\mathsf{D}_L}(s,x)$, where ${B}_{\mathsf{D}_L}(s,x)=\{\hat{s}\in \hat{\mathcal{M}}_s, \hat{x}\in \hat{\mathcal{M}}_x: \mathsf{d}_L(\mathsf{d}_s(s,  \hat{s}),\mathsf{d}_x(x,  \hat{x}))		
	\leq \mathsf{d}_L(\mathsf{D}_s,\mathsf{D}_x)\}$ denotes the set of $(\hat{s},\hat{x})$ which do not satisfy the linear distortion constraint. 
\end{IEEEproof}

{ 
\begin{Remark}
	This corollary can be specialized to JSCC with a single observable source $X$ \cite[Theorem 1]{Kostina-lossyJSCC-13} and JSCC with the indirect source $S$ by setting the slope factor $\alpha=1$ and $\alpha=0$ respectively. 	
	More precisely, when $\alpha=1$, the relaxed distortion constraint degenerates to $(\mathsf{D}_s=\infty, \mathbb{E}[\mathsf{d}_x({X}, {\hat{X}})] \leq \mathsf{D}_x)$, and the $\mathsf{d}()$-relaxed rate distortion function degenerates to the traditional rate distortion function. Hence Corollary \ref{L-converse bound} degenerates to \cite[Theorem 1]{Kostina-lossyJSCC-13} when $\alpha=1$. Meanwhile, when $\alpha=0$, the distortion constraint on reconstructing $X$ is relaxed, and the relaxed distortion constraint degenerates to $(\mathsf{D}_x=\infty, \mathbb{E}[\mathsf{d}_s({S}, {\hat{S}})] \leq \mathsf{D}_s)$. By setting $\hat{X}$ to be a constant, 
	Corollary \ref{L-converse bound} degenerates to the following corollary.
		\begin{Corollary}\label{converse-indirect jscc} 
			The existence of a $(\mathsf{D}_s, \epsilon)$ code requires that 
			\begin{align}
				\epsilon \geq &\inf_{P_{Y|X} } 
				\sup_{\gamma \geq 0}  
				\left\{  \sup_{P_{{Z}} } \mathbb{P} \big[ 
				\jmath_X(X,\hat{S}) -\imath_{Y;{Z}} (Y; Z)\geq \gamma \big]
				-\exp(-\gamma) \right\},
			\end{align}
			where the function $\jmath_X(X, \hat{S})$ is defined as 
			\begin{align}
				\jmath_X(x, \hat{s})
				=\imath_{X;\hat{S}} (x; \hat{s})
				+\lambda_s \big(\mathsf{d}_s({s}, {\hat{s}})-\mathsf{D}_s\big),  \label{Def_jxs}
			\end{align}
			and $\lambda_s$ is defined as 
			\begin{align}
				\lambda_s \triangleq & \frac{\partial R_{ISC}(\mathsf{D}_s) }{ \partial \mathsf{D}_s},
			\end{align}
			where $R_{ISC}(\mathsf{D}_s)$ is the indirect source coding function defined in \eqref{Def_R_ISC}. 
		\end{Corollary}
		
		The proof of Corollary \ref{converse-indirect jscc} follows the same steps of the proof of \cite[Theorem 1]{Kostina-lossyJSCC-13} , by replacing the direct source $\hat{X}$ by  $\hat{S}$.
\end{Remark}
}

 \section{General Achievability Bounds and Coding Scheme} \label{Sec_Achievability_general}

 In Section \ref{achi-scheme}, a two-layered superposition coding scheme is proposed for JSCC with two distortion constraints. Then, we establish {the achievable bounds of the excess-distortion probability} in Section \ref{achievability bounds}.  
Before proposing the achievability bounds, we first give the complementary definitions about the achievable rate tuple for the coding scheme under SSCC framework. For convenience, the following encoding and decoding functions are defined for a semantic source-channel code under the SSCC framework. 
\begin{itemize}
	\item Encoding function of the source encoder $f_S(\cdot):\mathcal{X}\rightarrow[1:M]$ maps the observable source $X$ into a message $W$.
	\item Encoding function of the channel encoder $f_C(\cdot):[1:M]\rightarrow \mathcal{Y}$ maps the source encoder output $W$ to a channel input $Y$.
	\item Decoding function of the channel decoder $g_C(\cdot):\mathcal{Z}\rightarrow [1:M]$ decodes the channel output $Z$ as a reconstruction of source encoder output $\hat{W}_1$.
	\item Decoding function for the observable source $g_O(\cdot):[1:M]\rightarrow\mathcal{X}$ further decodes $\hat{W}_1$ as a reconstruction of the observable source $\hat{X}$.
	\item Decoding function for the semantic source $g_S(\cdot):[1:M]\rightarrow\mathcal{S}$ further decodes $\hat{W}_1$ as a reconstruction of the semantic source $\hat{S}$.
\end{itemize} 
From a practical perspective, here we require $M\leq |\mathcal{Y}|$. Hence, a tuple $(M,\mathsf{D}_s,\mathsf{D}_x,\epsilon)$ is achievable if there exists encoding and decoding functions such that $P\{\mathcal{E}_{(\mathsf{D}_s,\mathsf{D}_x)}\} \leq \epsilon$. 

\subsection{Achievable Coding Scheme}\label{achi-scheme}

In this subsection, we propose a a two-layer superposition code for encoding the semantic source and observable source respectively.  To show the scheme explicitly, we will first introduce how to generate the codebook, and then we will introduce  the encoding and decoding functions.

	\subsubsection{Codebook Generation}
	For the source codebook, we first generate a core codebook with $M_1$ codewords $\{\hat{x}_1,\dots,\hat{x}_{M_1}\}$, according to the distribution $P_{\hat{X}}(\hat{x}_i)$. Then for every $\hat{x}_i$ where $i=1,\dots,M_1$, we generate a satellite codebook with $M_2$ codewords $
    \{\hat{s}_{i,1},\dots,\hat{s}_{i,M_2}\}$ according to the distribution $P_{\hat{S}|\hat{X}}(\hat{s}_{i,j}| \hat{x}_{m_i})$, where $j=1,\ldots, M_2$. Here $M_1\times M_2\leq M$.     
	For the channel codebook, without loss of generality we can write $\mathcal{Y}=\mathcal{Y}_1 \times \mathcal{Y}_2$.
	First, we can generate $M_1$ codewords $\{y_1(1), \dots, y_1(M_1)\}$ independently according to some $P_{Y_1}$. 
	Then, for each $m_1$, generate $M_2$ satellite codewords $\{ y_2(m_1, 1), \dots, y_2(m_1, m_2)\}$ according to $ P_{Y_2}$.  Notice that  $ P_{Y}$ can be uniquely identified by  $ P_{Y_1}P_{Y_2}$, and $y$ can also be uniquely identified by $(y_1,y_2)$. Hence, we have a codebook with $M_1\times M_2$ codewords as $\{y(1,1), \dots, y(M_1,M_2)\}$. 	
	Considering the randomness of $\hat{X}$ into account, it is equivalent to directly generating $\hat{s}$ according to $P_{\hat{S}}$ or first generating an $\hat{x}\in \hat{x}^{M_1}$ then generating $\hat{s}$ according to $P_{\hat{S}|\hat{X}}$ as described above.

	\subsubsection{Source Encoder}
	The source encoder first selects the lowest index $m_1$ from $\{1,\dots,M_1\}$ such that the distance between $x$ and $\hat{x}_{m_1}$ is within the  distance $\mathsf{D}_x$. If there is no such feasible index, the encoder outputs the largest index. Then it selects an index $m_2$ from $\{1,\dots,M_2\}$ so that the codeword from the satellite codebook $\hat{s}_{m_1,m_2}$ has the minimum excess-distortion probability over $\hat{s}^{M_2}(m_1)$.  Finally, the index tuple $(m_1, m_2)$ is mapped to the channel input as follows
	\begin{equation}
		m_1=\left\{
		\begin{array}{cc}
			m,   &  d(x,\hat{x}_{m})\leq \mathsf{D}_x < \min_{i=1,\dots,m-1} d(x,\hat{x}_{i}), \\
			M_1,    & \mathsf{D}_x< \min_{1,\dots,M_1} d(x,\hat{x}_{i}),
		\end{array}
		\right.
	\end{equation}
	and 
	\begin{align}
		m_2=\arg \underset{i}{\min}~\pi(x,\hat{s}_{m_1,i}|\hat{x}_{m_1}).
	\end{align}

	\subsubsection{Channel Encoder}
	The channel encoder maps the source tuple $(m_1, m_2)$ to a channel input sequence using superposition coding, i.e.,  the transmitted signal $y$ is a functional superposition:
	\begin{align}
		y = (y_1(m_1), y_2(m_1, m_2)).
	\end{align}
	where $y_1(m_1)$ carries the information of the observable source $x$, and $y_2(m_1, m_2)$ denotes the refinement layer codeword which carries the semantic information $s$ superimposed on $y_1(m_1)$.
	
	\subsubsection{Channel Decoder}
	Since the distribution of channel input $(Y_1, Y_2)$ can hardly be uniformly distributed, the best decoder here should be the maximum a posteriori (MAP) decoder. However, the performance of a MAP decoder is hard to analyze. Hence, here we use the maximum-likelihood (ML) decoder instead, where we have 
	\begin{align}
		\hat{m}_1 \triangleq g_{C_1}(z)=& \arg \underset{i}{\max}~P(z|y_1^n(i)), \nonumber\\
		\tilde{m}_2 \triangleq g_{C_2}(z)=& \arg \underset{j}{\max}~P(z|y_1^n(\hat{m}_1), y_2^n(\hat{m}_1, j)), \nonumber
	\end{align}
	where $\hat{m}_1$ is a recovery of the channel input $m$.

	\subsubsection{Source Decoder}
	We first obtain $\hat{m}_2$ by  
	\begin{align*}
		& \hat{m}_2= \tilde{m}_2-\lfloor \tilde{m}_2 /M_2\rfloor\times M_2,
	\end{align*}
	which is a recovery of the source encoder output $(m_1,m_2)$. Accordingly, the reconstructions of the observable source and the semantic source are given by
	\begin{align}
		& g_O(\hat{m}_1)=\hat{x}_{\hat{m}_1},\\
		& g_S(\hat{m}_2)=\hat{s}_{\hat{m}_1,\hat{m}_2}.
	\end{align}
 	
	\begin{Remark}
		The existing schemes in the literature can be regarded as special cases of the proposed codebook. For example, we consider two extreme cases where the distortion constraint of the observable source or that of the semantic source is relaxed. When the distortion constraint of the observable source is relaxed, by setting $M_1=1$ and $M_2=M$, the scheme is equivalent to that in \cite[Theorem 7]{Kostina-lossyJSCC-13}. Meanwhile, when the distortion constraint of the semantic source is relaxed, by setting $M_1=M$ and $M_2=1$, the scheme degenerates to that in \cite[Theorem 3]{Kostina-isc-2016}. 
	\end{Remark}
	
\subsection{Achievability Bounds for Semantic Communication in JSCC Framework}\label{achievability bounds}

Based on the coding scheme proposed in Sec. \ref{achi-scheme}, for any achievable $(M,\mathsf{D}_s,\mathsf{D}_x,\epsilon)$ tuple, the nonasymptotic achievability bound is given in the following theorem. 
{ 
	\begin{theorem} \label{achievability bound-JSCC2}
		There exists an $(M,\mathsf{D}_s,\mathsf{D}_x,\epsilon)$ code satisfying that 
		\begin{align}
			\epsilon' \leq &\underset{M_1\times M_2=M}{\inf}\underset{P_{Y_1Y_2},P_{\hat{S}\hat{X}}}{\inf} \Bigg\{ \mathbb{E}[\exp(-|\imath_{Y_1;Z}(Y_1;Z)-\log(M_1)|^{+})]
			+  \mathbb{E} \bigg[ \left(1-P_{\hat{X}}\left(B_{\mathsf{D}_x}\left(X\right)\right)\right)^{M_1}\bigg]\nonumber\\
			&+\mathbb{E}[\exp(-|\imath_{Y_2;Z}(Y;Z|Y_1)-\log(M_2)|^{+})]
			+ \mathbb{E} \bigg[T(X)\sum_{\hat{x}_i\in B_{\mathsf{D}_x}\left(X\right)} P(\hat{x})\int_{0}^{1} P^{M_2} \left\{  {\pi(X,\hat{S}|\hat{x})}>t\right\}dt\bigg]\Bigg\}, \label{Equ_Inner_GA1}
		\end{align}
		where $\pi(x,\hat{s}|\hat{x})=P\{\mathsf{d}_s(S,\hat{s})>\mathsf{D}_s|X=x,\hat{X}=\hat{x}\}$ with respect to the distribution $P_{SX\hat{S}\hat{X}}=P_{SX}P_{\hat{S}\hat{X}}$, and $T(X)=\sum_{i=0}^{M_1-1}\left(1-P_{\hat{X}}\left(B_{\mathsf{D}_x}\left(X\right)\right)\right)^{i}$.
	\end{theorem}
	
}
The proof of Theorem \ref{achievability bound-JSCC2} can be found in Appendix \ref{analyze ex-dis prob}. Notably, the distribution $P_{SX\hat{S}\hat{X}}=P_{SX}P_{\hat{S}\hat{X}}$ characterizes the probability of generating some particular codewords according to the marginal distribution of $(\hat{X},\hat{S})$, and is used for evaluating the performance of a random codebook. In fact, the distribution of $(\hat{X},\hat{S})$ is highly correlated with $X$, and $P_{\hat{S}\hat{X}}$ is obtained by calculating the marginal distribution of $P_{\hat{S}\hat{X}|X}P_{X}$. 

{ 
The proof characterizes the error probability by decomposing the overall system into two SSCC frameworks corresponding to the observable source $X$ and the semantic source $S$. 
The resulting bound \eqref{Equ_Inner_GA1} consists of four terms, which can be grouped into channel decoding errors and source reconstruction errors.
The first and third terms represent channel decoding errors associated with the two layers $(Y_1,Y_2)$.
The second term characterizes the excess-distortion probability of reconstructing the observable source $X$, where $\left(1-P_{\hat{X}}(B_{\mathsf{D}_x}(X))\right)^{M_1}$ represents the probability that none of the $M_1$ codewords fall within the distortion ball of $X$. Similarly, the fourth term captures the conditional excess-distortion probability of reconstructing the semantic source $S$ given $X$.
}
	
\begin{Remark}
	The achievability bound in Theorem \ref{achievability bound-JSCC2} differs from classical results in that the transmitter must simultaneously convey information about both the observable source and the semantic source, which necessitates an appropriate rate allocation between the two components.
	Specifically, $\log M_1$ bits are spent to transmit the information of the observable source, and $\log M_2$ bits are used to transmit the information of the semantic source.  
The achievability bound is indeed characterized as the infimum of the tradeoff between $M_1$ and $M_2$. 
\end{Remark}
	
{ \begin{Remark}
The achievability bound in Theorem \ref{achievability bound-JSCC2} can be specialized to related previous results.  
First, when $\mathsf{D}_s=\infty$, the JSCC with hierarchical sources degenerates to the JSCC model with observable data source. Intuitively,  the semantic source becomes a dummy component of this system, and the excess-distortion probability will be zero, even when no rate is spent on transmitting the semantic information.
For the $\pi(\cdot)$ function, we have
\begin{align}
	\pi(x,\hat{s}|\hat{x})=P\{\mathsf{d}_s(S,\hat{s})>\infty|X=x,\hat{X}=\hat{x}\}=0,
\end{align}
when $\mathsf{D}_s=\infty$ for all codebooks and all possible values of $S$.
Thus the third term in \eqref{Equ_Inner_GA1} equals zero. Thus, we have
\begin{align}
	T(X)\sum_{\hat{x}\in B_{\mathsf{D}_x}\left(X\right)} P(\hat{x})\int_{0}^{1} P^{M_2} \left\{ \pi(X,\hat{S}|\hat{x})>t\right\}dt=0,
\end{align}
for any value of $M_1$. 
Meanwhile, since $1-P_{\hat{X}}\left(B_{\mathsf{D}_x}\left(x\right)\right)\leq 1$,  
$\left(1-P_{\hat{X}}\left(B_{\mathsf{D}_x}\left(x\right)\right)\right)^{M_1}$ is monotonically nonincreasing as $M_1$ increases. Given the constraint that $M_1\times M_2=M$, the infimum value is taken when $M_1=M$ and $M_2=1$.  
Therefore, we can obtain
\begin{align}
	\epsilon \leq &\underset{P_Y,P_{\hat{X}}}{\inf} \Bigg\{ \mathbb{E}[\exp(-|{ \imath_{Y;Z} (Y; Z)}-\log(M)|^{+})] 
	+ \mathbb{E} \bigg[ \left(1-P_{\hat{X}}\left(B_{\mathsf{D}_x}\left({ X}\right)\right)\right)^{M}\bigg]\Bigg\}\label{JSCC-achievability bound-general},
\end{align} 
which exactly is the result of  \cite[Theorem 7]{Kostina-lossyJSCC-13}.

Second, when $\mathsf{D}_x=\infty$, the model degenerates to the indirect rate distortion model. 
Since there is no requirement on the distortion of reconstructing observable source $X$, we can set its reconstruction $\hat{X}$ as a constant, and apparently $P_{\hat{X}}\left(B_{\mathsf{D}_x}\left(x\right)\right)=1$ for all $x$ trivially. 
Moreover, since the channel is error-free, the information density $\imath(y;z)=\infty$ according to the definition of channel information density in \cite[Eq.(13)]{Kostina-isc-2016}. 
Given these two points, the term $\left(1-P_{\hat{X}}\left(B_{\mathsf{D}_x}\left(x\right)\right)\right)^{M_1}$ vanishes, and $T(X)$ and $\sum_{\hat{x}\in B_{\mathsf{D}_x}\left(X\right)} P(\hat{x})$ equal 1 trivially. Finally, since $\hat{X}$ is a constant, $\pi(x,\hat{s}|\hat{x})$ can be written as $\pi(x,\hat{s})$, where $\pi(x,\hat{s})=P\{\mathsf{d}_s(S,\hat{s})>\mathsf{D}_s|{X=s}\}$. 
Thus, the achievability bound in Theorem \ref{achievability bound-JSCC2} can be derived as 
        \begin{align}
\epsilon \leq &\underset{P_{\hat{S}}}{\inf} \, \mathbb{E}_{\sim P_X} \Bigg\{ \int_{0}^{1} P^{M} \left\{ \pi(X,\hat{S})>t\right\}dt\Bigg\},
        \end{align}
which is exactly the result of \cite[Theorem 3]{Kostina-isc-2016}. 

\end{Remark}}

Next, we consider a special class of codebooks where $M_1= M$ and $M_2=1$, and with a deterministic reconstruction function $\hat{s}=h(\hat{x})$, where $h(\cdot)$ is a deterministic function. 
 Hence, in this type of codebooks, we focus on the quality of reconstructing the observable source $X$, and construct the semantic information as a deterministic function $h(\hat{x})$, which is the best guess of $\hat{s}$ given $\hat{x}$.

{ This codebook follows a straightforward idea that since the encoder is not capable of observing the semantic source $S$, then we focus on encoding the information of $X$ and decode $S$ according to the reconstruction $\hat{X}$ at the receiver. }This strategy is reasonable when the semantic source and the observable source are joint GMSs. As pointed out in previous works of semantic rate distortion problem in infinite regime for GMSs \cite{Shi-errorexponent-2023}, the semantic source $S$, the observable source $X$, and the reconstruction of observable source $\hat{X}$, the reconstruction of semantic information $\hat{S}$ form a Markov chain $S-X-\hat{X}-\hat{S}$. It shows that $\hat{X}$ is the sufficient statistic to estimate $\hat{S}$, thus there is no need to spend extra bits on encoding $S$.  

For such coding scheme, an achievability bound can be established in the follow theorem. 
{ \begin{Corollary} \label{achievability bound-simplified}
     There exists an $(M,\mathsf{D}_s,\mathsf{D}_x,\epsilon)$ code satisfying that 
 \begin{align}
	\epsilon \leq &\underset{P_Y,P_{\hat{S}\hat{X}}}{\inf} \, \underset{M\geq \left\lfloor \frac{\gamma}{P_{\hat{X}}\left(B_{\mathsf{D}_x}\left(X\right)\right)}\right\rfloor +1}{\inf} \, \underset{h(\hat{x})}{\inf}\Bigg\{ \exp(1-\gamma) +\mathbb{E} \bigg[ \exp\left(-\left|\imath_{Y_1,Y_2;Z} (Y_1,Y_2; Z)-\log \frac{\gamma}{P_{\hat{X}}\left(B_{\mathsf{D}_x}\left(X\right)\right)}\right|^{+}\right)\nonumber\\
	&+\sum_{\hat{x}\in B_{\mathsf{D}_x}\left(X\right)}\frac{P(\hat{x})}{P_{\hat{X}}\left(B_{\mathsf{D}_x}\left(X\right)\right)}P\left\{\mathsf{d}_s(S,h(\hat{x}))>\mathsf{D}_s\right\}\bigg]\Bigg\}\label{leq-achievability bound-general-simplified},
\end{align}
where $\gamma$ is an arbitrary positive number, $\pi(x,h(\hat{x}))=P\{\mathsf{d}_s(S,h(\hat{x})>\mathsf{D}_s|X=x)\}$ and $h(\hat{x})$ is a function to generate  $\hat{s}$ according to $\hat{x}$.  
\end{Corollary}}

The proof of Corollary \ref{achievability bound-simplified} can be found in Appendix \ref{appendix-proof of achi}. 
Note that the codebook with $M_2=1$ can be regarded as a codebook for traditional JSCC, but with additional information about the semantic in each codeword. Hence, the excess-distortion probability of such codebook equals the excess-distortion probability of the JSCC codebook, which are the first two terms on the RHS of  \eqref{leq-achievability bound-general-simplified}, plus the probability of making a ``bad" guess of $\hat{S}$ according to $\hat{X}$, which is the third term on the RHS of  \eqref{leq-achievability bound-general-simplified}. 
{ Similar to Theorem \ref{achievability bound-JSCC2}, if $\mathsf{D}_s=\infty$, Corollary \ref{achievability bound-simplified} can also be specialized to \cite[Theorem 8]{Kostina-lossyJSCC-13}. }

As a remark, optimizing over all possible estimation functions $h(\hat{x})$ sets an obstacle for numerically calculating this bound, but this problem can be addressed by fixing $h(\hat{x})$ as a particular one. For example, it is a reasonable choice to set $h(\hat{x})$ as the minimum mean squared error (MMSE) estimator of $S$ given $X$ if the distortion measure function is linear.

\section{Gaussian Approximation  for Block-wise Transmission} \label{Sec_Gaussian_Apprx}

In this section,  we focus on the case where the source is an i.i.d. sequence of length $k$, which is encoded into a sequence of length $n$ and then transmitted over a memoryless channel. Note that the distributions of source and channel are not specified in this section. 

\begin{figure}[h]
		\centering
		\includestandalone[width=0.6\textwidth]{ChannelModel2}%
		\caption{A semantic-aware communication system model equipped with an $(n,k,\mathsf{D}_s,\mathsf{D}_x)$ lossy joint source-channel code}
		\label{Sys2}
\end{figure}

Now we consider a special case of the single-shot transmission as shown in Fig. \ref{Sys2}, where each single shot of the source outputs are a $k$-length i.i.d sequence $S^k$ and $X^k$ respectively. Each single shot of channel input  is an $n$-length sequence $Y^n$, and the channel output is an $n$-length sequence $Z^n$. A more rigorous description of the source and channel model is stated as follows. 

	\begin{itemize}
		\item[(1)] The channel is stationary and memoryless. For channel input ${Y}^n=(Y_1,\dots,Y_n)$ and channel output  ${Z}^n=(Z_1,\dots,Z_n)$, the transfer probability $P_{Z^n|Y^n}=\prod_{i=1}^nP_{Z_i|Y_i}$, where $P_{Z_1|Y_1}= \ldots = P_{Z_n|Y_n}$ are identically distributed. There is a cost function $c(\cdot)$ to measure the cost of channel input, where the average cost of some sequence $y^n=(y_1,\dots,y_n)$ is defined as 
  \begin{align}
      c_n(y^n)=\frac{1}{n}\sum_{i=1}^{n} c(y_i).
  \end{align}
		\item[(2)] The source is stationary and memoryless with distribution $P_{S^kX^k}=\prod_{i=1}^kP_{S_iX_i}$, where $P_{S_1X_1}=\ldots = P_{S_kX_k}$ are identically distributed. The average distortions are defined as 
{   \begin{align}
     \mathsf{d}_s(s^k,\hat{s}^k)=\frac{1}{k}\sum_{i=1}^{k} \mathsf{d}_s(s_i,\hat{s}_i), \\
      \mathsf{d}_x(x^k,\hat{x}^k)=\frac{1}{k}\sum_{i=1}^{k} \mathsf{d}_x(x_i,\hat{x}_i), 
  \end{align}}
where $\hat{s}^k=(\hat{s}_1,\dots,\hat{s}_k)$ and $\hat{x}^k=(\hat{x}_1,\dots,\hat{x}_k)$. 
  
  Consequently, for some given distortion constraint $(\mathsf{D}_s,\mathsf{D}_x)$, the excess distortion event $\mathcal{E}_{(\mathsf{D}_s,\mathsf{D}_x)}$ in this case is defined as 
  \begin{align}
    \mathcal{E}_{(\mathsf{D}_s,\mathsf{D}_x)}=\{(s^k,x^k,\hat{s}^k,\hat{x}^k)
    \in \mathcal{S}^k \times \mathcal{X}^k \times \hat{\mathcal{S}}^k \times \hat{\mathcal{X}}^k:
    \mathsf{d}_s(s^k,\hat{s}^k)>\mathsf{D}_s \cup \mathsf{d}_x(x^k,\hat{x}^k)>\mathsf{D}_x\}.
\end{align}
Meanwhile, for the given $\mathsf{d}()$ function, the $\mathsf{d}()$-relaxed excess distortion event $\tilde{\mathcal{E}}_{(\mathsf{D}_s,\mathsf{D}_x)}$ is defined by 
  \begin{align}
    \tilde{\mathcal{E}}_{(\mathsf{D}_s,\mathsf{D}_x)}=\{(s^k,x^k,\hat{s}^k,\hat{x}^k)
    \in \mathcal{S}^k \times \mathcal{X}^k \times \hat{\mathcal{S}}^k \times \hat{\mathcal{X}}^k:
   \mathsf{d}(\mathsf{d}_s(s^k,\hat{s}^k),\mathsf{d}_x(x^k,\hat{x}^k)) >  \mathsf{d}(\mathsf{D}_s,\mathsf{D}_x) \}.
\end{align}
		\item[(3)] Additionally, we impose the following restrictions: the expected value, the second, and the third moment of $\imath_{Z;Y} (Z; Y)$ and $\jmath_{S,X}(S, X) $ are finite.
	\end{itemize}

For this case, we define the achievability of a 
$(k,n,\rho,\mathsf{D}_s,\mathsf{D}_x,\epsilon)$ code, and the achievability of a $(k,n,\rho,\mathsf{D}_s,\mathsf{D}_x,\mathsf{d}(),\epsilon)$ code below.
\begin{definition}[An achievable $(k,n,\rho,\mathsf{D}_s,\mathsf{D}_x,\epsilon)$ code]
			For the given source $P_{S^kX^k}$ with alphabet $\mathcal{S}^k\times \mathcal{X}^k$, the channel $P_{Z^n|Y^n}$ with channel input-output alphabet $\mathcal{Y}^n\times \mathcal{Z}^n$, the cost function $c(\cdot)$, the cost constraint $\rho$,  and the distortion measures $\mathsf{d}_s(s^k, \hat{s}^k)$ and $\mathsf{d}_x(x^k, \hat{x}^k)$, 		
			a $(k,n,\rho,\mathsf{D}_s,\mathsf{D}_x,\epsilon)$ source-channel code { consists of} encoding function $f_k: \mathcal{X}^k\rightarrow { \mathcal{Y}^n}$, decoding function $g_{n,S}:\mathcal{Z}^n\rightarrow \hat{\mathcal{S}}^k$ and decoding function $g_{n,O}:\mathcal{Z}^n\rightarrow \hat{\mathcal{X}}^k$, such that 
   \begin{itemize}
       \item[(1)] { the expected average cost of channel input satisfies the cost constraint, namely $\mathbb{E}[c_n(Y^n)] \leq \rho\sigma_1^2$};
       \item[(2)] the excess-distortion probability satisfies  $P\{\mathcal{E}_{(\mathsf{D}_s,\mathsf{D}_x)}\} \leq \epsilon$.
   \end{itemize}
		\end{definition}	

\begin{definition}[An achievable $(k,n,\rho,\mathsf{D}_s,\mathsf{D}_x,\mathsf{d}(),\epsilon)$ code]
			For the given source $P_{S^kX^k}$ with alphabet $\mathcal{S}^k\times \mathcal{X}^k$, the channel $P_{Z^n|Y^n}$ with channel input-output alphabet $\mathcal{Y}^n\times \mathcal{Z}^n$, the cost function $c(\cdot)$, the cost constraint $\rho$, and the distortion measures $\mathsf{d}_s(s^k, \hat{s}^k)$ and $\mathsf{d}_x(x^k, \hat{x}^k)$, 		
			a $(k,n,\rho,\mathsf{D}_s,\mathsf{D}_x,\mathsf{d}(),\epsilon)$ source-channel code consists of encoding function $f_k: \mathcal{X}^k\rightarrow \mathcal{Y}^n$, decoding function $g_{n,S}:\mathcal{Z}^n\rightarrow \hat{\mathcal{S}}^k$ and decoding function $g_{n,O}:\mathcal{Z}^n\rightarrow \hat{\mathcal{X}}^k$, such that 
   \begin{itemize}
       \item[(1)] the expected average cost of channel input satisfies the cost constraint, namely $\mathbb{E}[c_n(Y^n)] \leq \rho \sigma_1^2$, where $\sigma_1^2$ is the variance of channel noise;
       \item[(2)] the excess-distortion probability satisfies  $P\{\tilde{\mathcal{E}}_{(\mathsf{D}_s,\mathsf{D}_x)}\} \leq \epsilon$.
   \end{itemize}
\end{definition}

Apparently, the achievability condition of a $(k,n,\rho,\mathsf{D}_s,\mathsf{D}_x,\mathsf{d}(),\epsilon)$ code is weaker than that of a $(k,n,\rho,\mathsf{D}_s,\mathsf{D}_x,\epsilon)$ code. Therefore, if a code is $(k,n,\rho,\mathsf{D}_s,\mathsf{D}_x,\epsilon)$-achievable,  it is definitely $(k,n,\rho,\mathsf{D}_s,\mathsf{D}_x,\mathsf{d}(),\epsilon)$ achievable. Thus, we have the following lemma. 

\begin{lemma}\label{relaxed-epsilon-bound}
    For a given tuple $(k,n,\rho,\mathsf{D}_s,\mathsf{D}_x)$ and $\mathsf{d}()$ function, the minimum achievable excess-distortion probability of a $(k,n,\rho,\mathsf{D}_s,\mathsf{D}_x,\epsilon)$ code is no less than that of a $(k,n,\rho,\mathsf{D}_s,\mathsf{D}_x,\mathsf{d}(),\epsilon)$ code. Therefore, 
    \begin{align}
        \inf \{\epsilon: (k,n,\rho,\mathsf{D}_s,\mathsf{D}_x,\epsilon)\textit{ is achievable}\} 
        \geq &\inf \{\epsilon: (k,n,\rho,\mathsf{D}_s,\mathsf{D}_x,\mathsf{d}(),\epsilon)\textit{ is achievable}\} .
    \end{align}
\end{lemma}

Lemma \ref{relaxed-epsilon-bound} shows that the excess-distortion probability can be lower-bounded by studying the relaxed excess-distortion probability in the block-wise transmission case. 
 With these definitions and Lemma \ref{relaxed-epsilon-bound}, we can establish the Gaussian approximation of converse bound on excess-distortion probability $\epsilon$ in the following theorem.
\begin{theorem}[Gaussian approximation of converse bound with $\mathsf{d}()$-functional distortion relaxation] \label{Thm_Gau_Converse1}
For the optimal  $(k, n, \rho, \mathsf{D}_s, \mathsf{D}_x, \mathsf{d}(), \epsilon)$ code, the following inequality holds,
	\begin{align}
		n{  C(\rho)}-kR_{S,X}(\mathsf{D}_s,\mathsf{D}_x)
		\leq & \sqrt{n{  V(\rho)}+k{ \tilde{V}(\rho)}}Q^{-1}(\epsilon)
		 +\theta(n),\label{Equ_NA_Cov_M1}
	\end{align}
where $\theta(n)\geq -\log n/2+\mathcal{O}(1)$, ${  C(\rho)}$ is the channel capacity  defined in \eqref{Channel_capacity}, and
\begin{itemize}
	\item[(1)] 
	{ $ V(\rho)$} is the channel dispersion given by
	\begin{align}
		 V(\rho)=& {\rm{Var}}\big[\imath_{Y;Z}(Y;Z) \big]. \label{channel-dispersion}
	\end{align}
	\item[(2)] { $\tilde{V}(\rho)$} is the source dispersion given by
	\begin{align}
		{ \tilde{V}(\rho)}=&  {\rm{Var}}\big[\jmath_{S,X}(S, X)\big], \label{GN_source-dispersion}
	\end{align}
	and $R_{S,X}(\mathsf{D}_s,\mathsf{D}_x)$ is the rate distortion function defined in \eqref{Equ_R_d}.	
\end{itemize}
\end{theorem}

The proof of Theorem \ref{Thm_Gau_Converse1} can be found in Appendix \ref{appendix-proof of gaussian approximation linear}. Similar to the result in \cite{Kostina-lossyJSCC-13}, Theorem \ref{Thm_Gau_Converse1} shows that the excess-distortion probability can be approximated as the tail probability of a Gaussian distribution. Meanwhile, this probability can be lower-bounded by the gap between channel capacity and rate distortion, with a coefficient determined by the channel dispersion and source dispersion.

{ \begin{Remark}
Theorem \ref{Thm_Gau_Converse1} is a reasonable result since the essence of nonasymptotic bound depends on the central limit theorem and Chebyshev inequality, which hold regardless of whether it is a classic JSCC model with one source or JSCC model with hierarchical sources. However, the nonasymptotic bound for JSCC model with hierarchical sources is not a straightforward extension of that for classic JSCC model for two reasons. 
First, the distortion now becomes a multivariate random variable with the distortion of semantic source and the distortion of observable source respectively. And the excess distortion event now is a union of two correlated error events which is difficult to analyze its probability. 
Second, it is a challenge to establish Gaussian approximation even by using the proposed $\mathsf{d}()$-tilted information, since the parameter of the $\mathsf{d}()$ function should be optimized.
An in-depth discussion of these two points can be found in the next subsection, where we derive the bounds for  GMS transmitting over AWGN channels. 
\end{Remark}}

According to Theorem \ref{Thm_Gau_Converse1}, we can derive the lower bound of excess-distortion probability $\epsilon$ as follows. 
\begin{Corollary}\label{Converse_GA_err}
For any achievable $(k,n,\rho,\mathsf{D}_s,\mathsf{D}_x,\epsilon)$ code, the error probability is lower-bounded by
\begin{align}
	\epsilon  \geq & Q\left( \frac{n C(\rho)-kR_{S,X}(\mathsf{D}_s,\mathsf{D}_x)-\theta(n)}{\sqrt{n V(\rho)+k\tilde{V}(\rho)} } 
 \right).\label{Equ_Corr1}
\end{align}
\end{Corollary}
{\color{black}Next, by Theorem \ref{achievability bound-JSCC2}, we can derive the Gaussian approximation of achievability bound as follows. }

{ 
\begin{theorem}[Gaussian approximation of achievability bound with superposition coding] \label{Thm_Gau_Ach1}
	Denote $\upsilon_1 \triangleq \sqrt{n V_1(\rho) + k \tilde{V}_1(\rho)}$, $\upsilon_2\triangleq \sqrt{n V_2(\rho) + k \tilde{V}_2(\rho)}$, and $\theta(k)\leq \log k+\log \log k+\mathcal{O}(1)$.
	For the optimal  $(k,n,\rho,\mathsf{D}_s,\mathsf{D}_x,\mathsf{d}(),\epsilon)$ code, the following bounds hold.

	\begin{itemize}
		\item[(1)] If $\frac{\upsilon_1}{\upsilon_2}< \exp\big\{-\frac{(n C(\rho)-k R_{S,X}(\mathsf{D}_s,\mathsf{D}_x))^2}{2}\big\} < \frac{\upsilon_2}{\upsilon_1}$, 
		\begin{align}
			\epsilon 
			\leq Q\left(\frac{n C(\rho)-k R_{S,X}(\mathsf{D}_s,\mathsf{D}_x)+\theta(k)}{\upsilon_2}  \right).
		\end{align}  		
		
		\item[(2)] If $\frac{\upsilon_2}{\upsilon_1}< \exp\big\{-\frac{(n C(\rho)-k R_{S,X}(\mathsf{D}_s,\mathsf{D}_x))^2}{2}\big\}< \frac{\upsilon_1}{\upsilon_2}$, 
		\begin{align}
			\epsilon 
			\leq Q\left(\frac{n C(\rho)-k R_{S,X}(\mathsf{D}_s,\mathsf{D}_x)+\theta(k)}{\upsilon_1}  \right).
		\end{align}  
		
		\item[(3)] If $\exp\big\{-\frac{(n C(\rho)-k R_{S,X}(\mathsf{D}_s,\mathsf{D}_x))^2}{2}\big\} <\min \Big\{ \frac{\upsilon_2}{\upsilon_1},  \frac{\upsilon_1}{\upsilon_2} \Big\}$, or $ \exp\big\{-\frac{(n C(\rho)-k R_{S,X}(\mathsf{D}_s,\mathsf{D}_x))^2}{2}\big\}>\max \Big\{ \frac{\upsilon_2}{\upsilon_1},  \frac{\upsilon_1}{\upsilon_2} \Big\}$, 
		\begin{align}
			\epsilon 
			\leq  Q\left(\frac{T_1^\star+\theta(k)}{\upsilon_1}  \right)
			+Q\left(\frac{n C(\rho)-k R_{S,X}(\mathsf{D}_s,\mathsf{D}_x)-T_1^\star+\theta(k)}{\upsilon_2}  \right),		
		\end{align}   
		where $T_1^\star$ is the solution to
		\begin{align}
			\frac{1}{\upsilon_1} f\left(\frac{T_1^\star}{\upsilon_1}  \right)
			-\frac{1}{\upsilon_2} f\left(\frac{n C(\rho)-k R_{S,X}(\mathsf{D}_s,\mathsf{D}_x)-T_1^\star}{\upsilon_2}  \right)=0. \label{Equ_Inner_C31}
		\end{align} 
		and $f(x)=\frac{1}{ \sqrt{2\pi} } e^{-\frac{x^2}{2}}$.  
	\end{itemize} 
	The corresponding definitions are given as follows.
	\begin{itemize}
		\item[(1)] 	$\upsilon_1 \triangleq \sqrt{n V_1(\rho) + k \tilde{V}_1(\rho)}$ and $\upsilon_2\triangleq \sqrt{n V_2(\rho) + k \tilde{V}_2(\rho)}$;
		\item[(2)]  $V_1(\rho)$ and $V_2(\rho)$ are the channel dispersions given by
		\begin{align}
			V_1(\rho)=& {\rm{Var}}\big[\imath_{Y_1;Z}(Y_1;Z) \big],  \label{channel-dispersion1} \\
			V_2(\rho)=&{\rm{Var}}\big[\imath_{Y_2;Z}(Y_2;Z|Y_1)  \big],  \label{channel-dispersion2} 
		\end{align}
		and $ C(\rho)$ is the channel capacity;
		\item[(3)] $R_{S,X}(\mathsf{D}_s,\mathsf{D}_x)$ is the rate distortion function defined in \eqref{Equ_R_d}.
		\item[(4)] $\tilde{V}(\rho)$ is the source dispersion given by
		\begin{align}
			\tilde{V}_1(\rho)=& {\rm{Var}}\big[\jmath_X(X, \hat{X})\big], \\
			\tilde{V}_2(\rho)=& {\rm{Var}}\big[\jmath_X(X, \hat{S})|\hat{X}=\hat{x}\big],
		\end{align}
		where $\jmath_X(x, \hat{s})$ is defined in \eqref{Def_jxs}, and $\jmath_X(x, \hat{x})$ is defined as follows
		\begin{align}
			&\jmath_{X}(x, \hat{x})
			\triangleq\imath_{X; \hat{X}} (x; \hat{x})
			+\lambda_x\big( \mathsf{d}_x({x}, {\hat{x}})-\mathsf{d}(\mathsf{D}_x)\big), \label{Equ_Def_j1} \\
			& \imath_{X; \hat{X}} (x; \hat{x})
			\triangleq \log \frac{ {\textnormal{d}} P_{ X| \hat{X}=\hat{x}} }{ {\textnormal{d}} P_{X}} (x),\\
			& \lambda_x \triangleq \frac{\partial R_{RD}(\mathsf{D}_x) }{ \partial \mathsf{D}_x}. 
		\end{align}
		and $R_{RD}(\mathsf{D}_x)$  is defined in \eqref{Def_R_RD}.		
	\end{itemize} 
\end{theorem}
 \begin{IEEEproof}
 	See Appendix \ref{Appx_GN_Superposition}.  
 	Applying the nonasymptotic achievability bound of Theorem  \ref{achievability bound-JSCC2}, the excess-distortion probability is expressed as the sum of channel information density and source covering terms for the two layers. Each term is analyzed using the Berry-Esseen theorem, leading to the constraints $nC_1(P)-k R_1(\mathsf{D}_x) \geq  \sqrt{n V_1(\rho) + k \tilde{V}_1(\rho)} Q^{-1}(\epsilon_{1})+\theta_1(k)$ and $nC_2(P)-k R_2(\mathsf{D}_s)	\geq  \sqrt{n V_2(\rho) + k \tilde{V}_2(\rho)} Q^{-1}(\epsilon_{2})$. After optimizing the sum of two inverse $Q$-functions over $\epsilon_1+\epsilon_2=\epsilon$, the Gaussian approximation is derived.
 \end{IEEEproof}
}

{ Characterizing the optimal JSCC performance is inherently challenging due to the intrinsic tradeoff between the reconstructions of observable source $X$ and semantic source $S$, which means optimizing the distortion of $X$ is generally misaligned with that of $S$.		
Indeed, even in the first-order asymptotic regime, this joint optimization problem is computationally complex  \cite{Liu-semantic-2022}.
This difficulty is further exacerbated in the finite blocklength regime, where the inclusion of the dispersion term and the Gaussian $Q$-function leads to transcendental equations during optimization. 
In \cite{Ulger-23}, a similar source coding problem of the hierarchical source was studied. However, no computable optimal bounds are established. Theorem \ref{Thm_Gau_Ach1} solves the optimization problem  by decomposing the overall JSCC system into two coupled SSCC processes and deriving a constrained minimization of a sum of Gaussian $Q$-functions.
It is proved in Theorem \ref{Thm_GApprx_GMS} that the Gaussian approximation is optimal for the transmission of Gaussian memoryless sources over AWGN channels when $\hat{S}$ is a deterministic function of $\hat{X}$. }
 
{ \begin{Remark}
	For the JSCC with a single observable source, the Gaussian approximation is given by \cite[Theorem 19]{Kostina-lossyJSCC-13}
	\begin{align}
		n C(\rho)-kR_{X}(\mathsf{D}_x)
		\leq & \sqrt{n V(\rho)+k\tilde{V}(\rho)}Q^{-1}(\epsilon)
		+\theta(n),\label{Equ_NA_Kositna}
	\end{align}
	where $ C(\rho), R_X(\mathsf{D}_x),  V(\rho)$, and $\tilde{V}(\rho)$ are defined in \eqref{Channel_capacity}, \eqref{Equ_R_d}, \eqref{channel-dispersion} and \eqref{GN_source-dispersion} respectively with $\mathsf{D}_s=\infty$, and $\mathcal{O}(1)\leq \theta(n)\leq \log n+\log \log n+\mathcal{O}(1)$.	
	Compared with Gaussian approximations of the converse and the achievability bounds in Corollary \ref{Converse_GA_err} and Theorem \ref{Thm_Gau_Ach1}, the Gaussian approximation of single source given by \cite[Theorem 19]{Kostina-lossyJSCC-13} is optimal for the general case. The proposed Gaussian approximation of hierarchical sources is optimal for the transmission of a GMS over AWGN channels, which is given in Theorem \ref{Thm_GApprx_GMS}. In that case, the result in Theorem \ref{Thm_GApprx_GMS} can be specialized to \cite[Theorem 19]{Kostina-lossyJSCC-13} with $\mathsf{D}_s=\infty$.
	
\end{Remark}}

\section{Transmission of a GMS over an AWGN Channel}   \label{Sec_GMS_AWGN}

In this section,  we focus on the transmission of an $k$-length i.i.d. Gaussian memoryless source, which is encoded into a sequence of length $n$ and then transmitted over a memoryless channel. We establish the nonasymptotic achievability and converse bounds for the case of GMS and AWGN channels. Moreover, the Gaussian approximations is given and proved to be optimal. 
First, the source and channel models are introduced below. 
\begin{itemize}
    \item[(1)] The semantic source $S^k$ and the observable source $X^k$ are jointly Gaussian distributed. $X^k$ is a $k$-length i.i.d. sequence, and $X_i\sim \mathcal{N}(0,\sigma_X^2)$, $i=1, \ldots, k$. 
    Meanwhile, the semantic source $S_i$ can be written as $S_i=X_i+N_{s_i}$, where $N_{s_i}\sim \mathcal{N}(0,\sigma_{N_s}^2)$, $i=1, \ldots, k$. 
    \item[(2)]  The channel output of an AWGN channel is $Z^n=Y^n+N_1^n$, where ${Y}^n=(Y_1,\dots,Y_n)$ is the $n$-length i.i.d. channel input sequence satisfying the cost constraint $\mathbb{E}[c_n(Y^n)] \leq \rho\sigma_1^2$.    
     $N_1^n$ is the $n$-length i.i.d. noise sequence, where $N_{1_i}\sim \mathcal{N}(0,\sigma_1^2)$. The channel is stationary and memoryless, and the transfer probability $P_{Z^n|Y^n}=\prod_{i=1}^nP_{Z_i|Y_i}$, where $P_{Z_1|Y_1}= \ldots = P_{Z_n|Y_n}$ are identically distributed.     
    
    \item[(3)]   The mean squared error is used as the distortion measure function, i.e. 
    \begin{align}
        \mathsf{d}_x(x^k,\hat{x}^k)=\frac{1}{k}\sum_{i=1}^{k}\mathsf{d}_x(x_i-\hat{x}_i)^2,\\
        \mathsf{d}_s(s^k,\hat{s}^k)=\frac{1}{k}\sum_{i=1}^{k}\mathsf{d}_s(s_i-\hat{s}_i)^2.
    \end{align}
\end{itemize}

For such a transmission of a GMS over an AWGN Channel, the nonasymptotic bounds are established as follows. 
\begin{theorem} \label{converse-GMS}
For the case of transmitting GMS over AWGN channels, if there exists a $(k,n,\rho,\mathsf{D}_s,\mathsf{D}_x,\mathsf{d}(),\epsilon)$ code, then 
    \begin{align}
        	\epsilon \geq \sup_{\alpha \in [0,1]}
				\sup_{\gamma \geq 0}  
				\left\{ \mathbb{P} \big[ U \geq 
				nC-kR_{S,X}(\mathsf{D}_s,\mathsf{D}_x)+\gamma \big]
				-\exp(-\gamma) \right\},\label{L-converse bound-gaussian}
    \end{align}
   where 
   \begin{align}
       U=(\frac{\rho}{1+\rho}W_{\frac{n}{\rho}}^n-n)+W_{0}^{k}-\frac{ k\mathsf{D}_L}{\mathsf{D}_L -(1-\alpha) \sigma_{N_s}^2}
       +\frac{ 2(1-\alpha)\sum_{i=1}^{k}\Delta_{X_i}\sigma_{N_s}+ \sum_{i=1}^{k}(1-\alpha)\sigma^2_{N_s}}{\mathsf{D}_L -(1-\alpha) \sigma_{N_s}^2 }.\label{def of U}
   \end{align}

\end{theorem}

	The proof of Theorem \ref{converse-GMS} can be found in Appendix \ref{proof of converse GMS}. It can be proved by taking the distribution of GMS and AWGN channels into Theorem \ref{Thm_converse}. To derive the achievability bound, we first introduce the a useful lemma.
 
{ 	 
 \begin{figure*} 
 	\centering
 	\captionsetup[subfigure]{margin=120pt} 
 	\begin{minipage}[t]{0.3\linewidth}
 		\subfigure[]{
 			\centering
 			\includegraphics[width=0.8\textwidth]{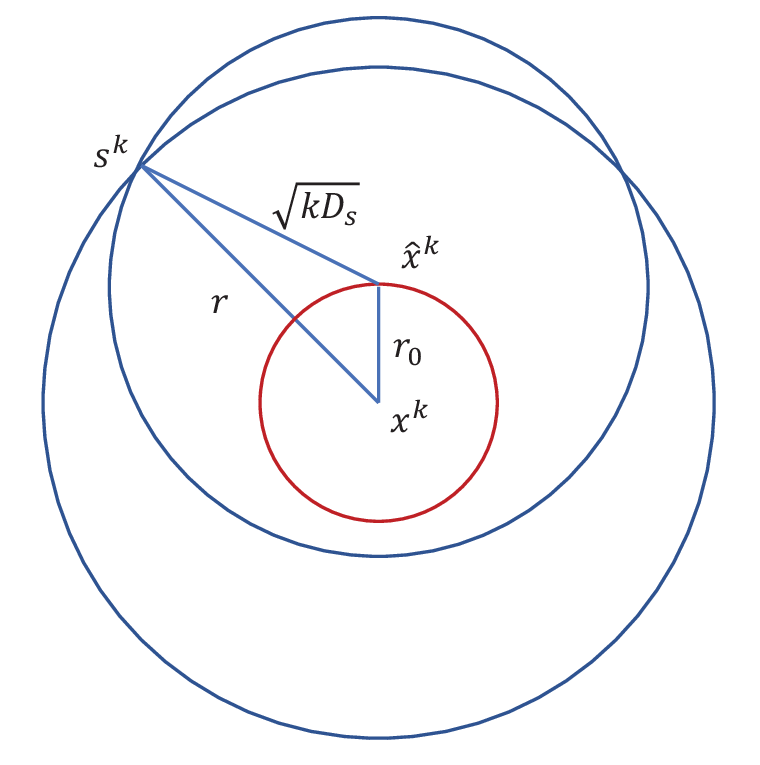}}
 	\end{minipage}
 	\hspace{-.45in} 
 	\begin{minipage}[t]{0.3\linewidth}
 		\centering
 		\subfigure[]{
 			\includegraphics[width=0.8\textwidth]{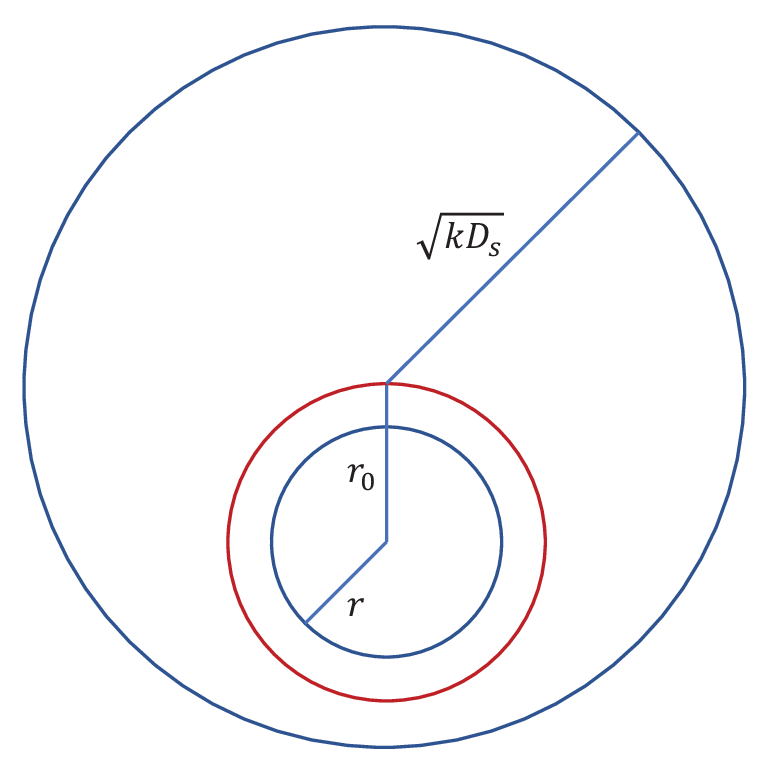}}
 	\end{minipage}
 	\begin{minipage}[t]{0.3\linewidth}
 		\subfigure[]{
 			\centering
 			\includegraphics[width=0.8\textwidth]{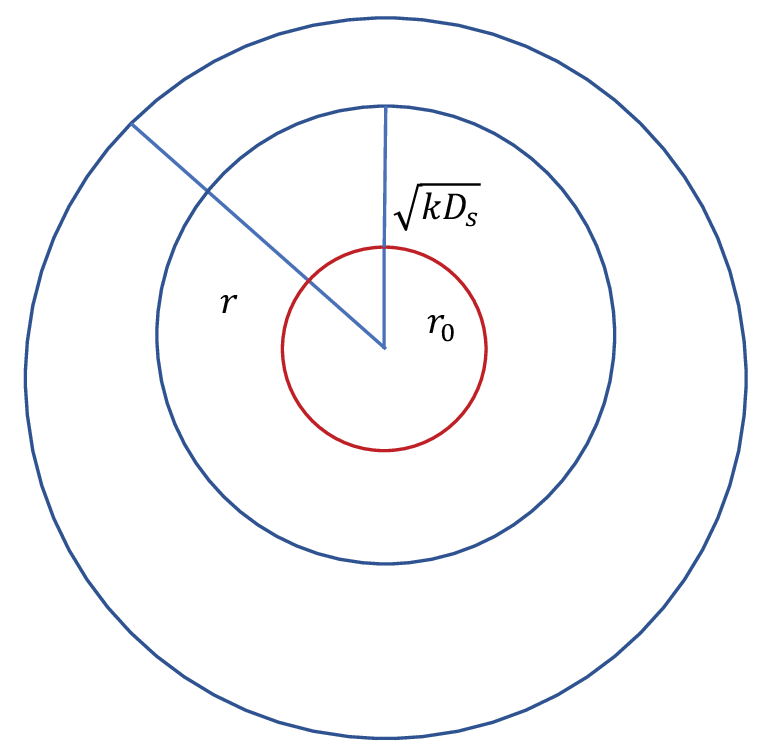}}
 	\end{minipage}
 	\caption{The geometry of the excess-distortion probability calculation for three different cases. (a) $r$ is within the radius bounds; (b) $r$ is too small; (c) $r$ is too large.}\label{Fig_Bd_Inner}
 \end{figure*}
 
Denote $\Delta(x_i)$ by the reconstruction error for given the observable source $x_i$, namely  $\Delta(x_i)=x_i-\hat{X}_i$. 
 $\Delta(x_i)$ is a Gaussian random variable with distribution $\Delta(x_i)\sim \mathcal{N}(x_i, \sigma^2_{\hat{X}})$, which varies with the source output $x_i$. 
We have the following lemma.
 \begin{lemma} \label{Lemma_Bd}
 	Let $\sum_{i=1}^k(s_i-x_i)^2=r^2$, $\sum_{i=1}^k(x_i-\hat{x}_i)^2=r_0^2$. The following statements hold, as shown in Fig. \ref{Fig_Bd_Inner}.
 	\begin{itemize}
 		\item[(1)] If $r\leq \sqrt{k\mathsf{D}_s}-r_0$, the excess-distortion probability $P_{S}\left(B_{\mathsf{D}_s}\left(\hat{x}\right)\right)=1$.
 		
 		\item[(2)] If $r>\sqrt{k\mathsf{D}_s}+r_0$, the excess-distortion probability $P_{S}\left(B_{\mathsf{D}_s}\left(\hat{x}\right)\right) =0$. 
 		
 		\item[(3)] If $\sqrt{k\mathsf{D}_s}-r_0\leq |n_s|\leq \sqrt{k\mathsf{D}_s}+r_0$, the excess-distortion probability 
 		\begin{align}
 			P_{S^k}\left(B_{\mathsf{D}_s}\left(\hat{x}^k\right)\right)
 			\geq &
 			\frac{\Gamma\left(\frac{k}{2}+1\right)}{\sqrt{\pi}k\Gamma\left(\frac{k-1}{2}+1\right)} 
 			\left(1-  \left(  \frac{ \sigma^2_{N_s}w_1 +\sigma^2_{\hat{X}}w_2- \mathsf{D}_s}{2 \sqrt{ \sigma^2_{N_s}w_1  } \sqrt{  \sigma^2_{\Delta_X}w_2}} \right)^2  \right)^{\frac{k-1}{2}}.
 			\label{Equ_GSM_Inner_rhoB0}
 		\end{align}
 	\end{itemize}
 	where $w_1=r^2/\sigma_{N_s}^2$ and $w_2=r_0^2/\sigma^2_{\hat{X}}$.
 \end{lemma}
 \begin{IEEEproof}
 	See Appendix \ref{Proof_Lem_Bd}. 	
 	Note that the results in \cite[Therem 37]{Kostina-fixed-12} can not be directly applied here, since the distribution of $S^k$ depends on $\hat{x}^k$.  There are two major differences between the geometry of this case and that of the case in \cite{Kostina-fixed-12}. First, the sphere center here is $x^k$ instead of the origin point.
 	In this sense, we evaluate the probability that an i.i.d. Gaussian sequence $N_s^k = (N_{s_1}, \dots, N_{s_k})$ lies within the $ \mathsf{D}_s$-ball centered at  $\Delta(x^k) = (\Delta(x_1), \dots, \Delta(x_k))$. 	 
 	Second, in this case the distortion constraint for semantic information $\mathsf{D}_s$ is larger than $\sigma_{N_s}^2$ from a physical interpretation. That is to say, it is trivial to analyze the case where $\mathsf{D}_s < \sigma^2_{N_s}$, since the MMSE of estimating $S_i$ with a perfect knowledge of $X_i$ is $\sigma^2_{N_s}$. Hence, it will have a non-negligible probability that when $r$ is small enough, all $s^k$ such that $\sum_{i=1}^k(s_i-x_i)^2=r^2$ can be decoded within the distortion constraint. This situation is shown in Fig. \ref{Fig_Bd_Inner}-(b). Thus, when $r\leq \sqrt{k\mathsf{D}_s}-r_0$, the excess-distortion probability for semantic source, namely $1-P_{S^k}\left(B_{\mathsf{D}_s}\left(\hat{x}^k\right)\right)$, equals zero. This is a case that does not arise in \cite{Kostina-fixed-12}. Meanwhile, it will also have a small probability that $r$ is too large so that no such $s^k$ with $\sum_{i=1}^k(s_i-x_i)^2=r^2$ can be correctly encoded. This case is shown in Fig. \ref{Fig_Bd_Inner}-(c). In this case, namely when $r>\sqrt{k\mathsf{D}_s}+r_0$, the excess-distortion probability $1-P_{S^k}\left(B_{\mathsf{D}_s}\left(\hat{x}^k\right)\right) =1$. 
 	The detailed proof of the case Fig. \ref{Fig_Bd_Inner}-(a) is given in Appendix \ref{Proof_Lem_Bd}. 	
 \end{IEEEproof}
}

\begin{theorem} \label{achievability bound-GMS}
	For the transmission of GMS over AWGN channels, if there exists a $(k,n,\rho,\mathsf{D}_s,\mathsf{D}_x,\mathsf{d}(),\epsilon)$ code, then 
	{\small \begin{align}
		\epsilon \leq &  \underset{\gamma}{\inf}
		\left\{ \exp(1-\gamma) 
		+\mathbb{E}_{W_{\frac{n}{\rho}}^n, W_{0}^n} \left[ \exp\left\{ - |U-\log \gamma |^+ \right\} \right]    \right. \nonumber \\
	    & \left.+\mathbb{E}_{W_0^k}  
		  \left[ \frac{1}{\rho_1 \left( W_0^{k}, \mathsf{D}_x, \sigma_x \right)} \int_{w_2 \leq {\mathsf{D}_x}/{\sigma^2_{\hat{X}}}} f_{W_{\sigma^2_{X}W_0^k/\sigma^2_{\hat{X}}}^{k}}(w_2)  
  \mathbb{E}_{W_0^k} 	\left[1-  \rho_2(\sigma^2_{N_s}w_{0}^k,\sigma^2_{\hat{X}}w_2,\mathsf{D}_s)  \right] \text{d}w_2 \right]  \right\}, \label{Equ_GSM_Inner}
	\end{align}} 
	where
	\begin{align}
		U=&\frac{n}{2}\log (1+\rho)-\frac{1}{2}\left(\frac{\rho}{1+\rho} W_{\frac{n}{\rho}}^n-n \right) -\log \frac{F}{\rho_1 \left( W_0^{k}, \mathsf{D}_x, \sigma_x \right)  }, \\
		F=& \max_{n\in \mathbb{N}, t\in \mathbb{R}^+ } \frac{f_{W_{n\rho}^n}(t)}{f_{W_{0}^n}(\frac{t}{1+\rho})} <\infty, \\
  \sigma^2_{\hat{X}}=&\sigma^2_{X}-\min\{\mathsf{D}_x, \mathsf{D}_s-\sigma^2_{N_s}\},
	\end{align}
	and  $\gamma$ is an arbitrary positive number. Meanwhile, function 
 $\rho_1: \mathbb{R}^+\mapsto [0,1]$ is defined by
	\begin{align}
		\rho_1(t,\mathsf{D},\sigma )
		=& \frac{\Gamma\left(\frac{k}{2}+1\right)}{\sqrt{\pi}k\Gamma\left(\frac{k-1}{2}+1\right)} 
		\left(1-L\left(\sqrt{\frac{t}{k}},\mathsf{D},\sigma\right)\right)^{\frac{k-1}{2}} ,\label{Equ_GSM_Inner_rho}
	\end{align}
	where the $L(\cdot)$ function is defined as 
	\begin{equation}  \label{Equ_GMS_Inner_Lr}
		L(r,\mathsf{D},\sigma)=\left\{
		\begin{aligned}
			&0, & \quad  r<\sqrt{\frac{\mathsf{D}}{\sigma^2}}-\sqrt{1-\frac{\mathsf{D}}{\sigma^2}} \\
			&1, &\quad \left|  r-\sqrt{1-\frac{\mathsf{D}}{\sigma^2}}  \right| > \sqrt{\frac{\mathsf{D}}{\sigma^2}} \\
			& \frac{\left(1+r^2-2\frac{\mathsf{D}}{\sigma^2} \right)^2}{4\left( 1-\frac{\mathsf{D}}{\sigma^2}\right)r^2}, &  \rm{otherwise,}
		\end{aligned}
		\right.
	\end{equation} 
and function $\rho_2: \mathbb{R}^+\mapsto [0,1]$ is defined by
	\begin{equation}  \label{Equ_GMS_Inner_rho2}
		\rho_2(r^2,r_0^2,\mathsf{D})=\left\{
		\begin{aligned}
			&0, & \quad  r>r_0+\sqrt{k\mathsf{D}} \\
			&1, &\quad r\leq \sqrt{k\mathsf{D}}-r_0 \\
			& \frac{\Gamma\left(\frac{k}{2}+1\right)}{\sqrt{\pi}k\Gamma\left(\frac{k-1}{2}+1\right)} 
		\left(1-\frac{\left(r^2+r_0^2-2k\mathsf{D} \right)^2}{4r_0^2r^2}\right)^{\frac{k-1}{2}} , &  \rm{otherwise.}
		\end{aligned}
		\right.
	\end{equation} 
\end{theorem}
\begin{IEEEproof}
	 See Appendix \ref{Proof of achievability bound GMS}. The proof is by taking the distribution of GMS and AWGN channel into Corollary \ref{achievability bound-simplified}. Specifically, for Gaussian sources and codebook, the probability of correctly decoding within some $\mathsf{D}$ distortion constraint $P(B_{\mathsf{D}})$ can be lower bounded by the ratio of two Gamma functions. 	 
\end{IEEEproof}

{ \begin{theorem}[Gaussian approximation of GMS-AWGN] \label{Thm_GApprx_GMS}
	For the transmission of a GMS over an AWGN channel, the error probability should satisfy
	\begin{align}
		\epsilon 
		= Q\left(\frac{n C(\rho)-k R_{S,X}(\mathsf{D}_s,\mathsf{D}_x)+\theta(n,k)}{\sqrt{n  V(\rho) + k \tilde{V}(\rho)}}  \right).
	\end{align} 
	where $-\log n/2+\mathcal{O}(1)\leq \theta(n,k)\leq \log k+\log \log k+\mathcal{O}(1)$, and the channel capacity $ C(\rho)$ and channel dispersion $ V(\rho)$ are given by
	\begin{align}
		&  C(\rho)= \frac{1}{2} \log (1+\rho), \label{Equ_C}\\
		&  V(\rho)= \frac{1}{2} \left(1-\frac{1}{(1+\rho)^2}\right) . \label{Equ_V}
	\end{align}
	Moreover, for the linear $\mathsf{d}(\cdot)$ function defined in  \eqref{Def_linear_d}, the linear $\mathsf{d}()$-relaxed rate distortion function and the source dispersion $\tilde{V}$ are given by
	\begin{align}
		R_{S,X}(\mathsf{D}_s,\mathsf{D}_x) 
		=&\frac{1}{2} \log \frac{\sigma_x^2}{ \mathsf{D}_L- (1-\alpha) \sigma_{N_s}^2 } \leq \bar{R}_{S,X}(\mathsf{D}_s,\mathsf{D}_x),\label{Equ_linear_R} \\
		\tilde{V}(\rho)=& \frac{1}{2}+\frac{4(1-\alpha)^2(\mathsf{D}_L -(1-\alpha) \sigma_{N_s}^2)\sigma_{N_s}^2+2(1-\alpha)^2\sigma_{N_s}^4}{4(\mathsf{D}_L -(1-\alpha) \sigma_{N_s}^2)^2 }\label{Equ_linear_V}, 
	\end{align}
	where $\mathsf{D}_L=\alpha \mathsf{D}_x+(1-\alpha)\mathsf{D}_s$. 
\end{theorem}
\begin{IEEEproof}
	See Appendix \ref{source dispersion linear proof}. Specifically, we evaluate the  source dispersion of linear $\mathsf{d}(\mathsf{d}_s,\mathsf{d}_x)$-tilted information for GMS, where the variance of $\mathsf{d}(\mathsf{d}_s,\mathsf{d}_x)$-tilted information can be computed explicitly. 
\end{IEEEproof}

Note that the slope rate $\alpha$ should be optimized to obtain the tightest Gaussian approximation. However, the derivative of the Gaussian approximation will be a high-order polynomial whose root cannot be written in analytical forms in general, and we can only obtain closed-form solutions for some special cases. }

\begin{Corollary}
	In the transmission of a GMS over an AWGN channel, the optimal Gaussian approximation for excess-distortion probability with a linear $\mathsf{d}(\cdot)$ function is obtained by ignoring the distortion constraint of semantic information when { $\mathsf{D}_x+\sigma^2_{N_s}-\mathsf{D}_s<0$}. That is 
	\begin{align}
		\arg \min_{\alpha} \frac{nC(\rho)-kR_{S,X}(\mathsf{D}_s,\mathsf{D}_x)}{\sqrt{nV(\rho)+k\tilde{V}(\rho)}}=1,
	\end{align}
\end{Corollary}

\begin{figure}[!htpb]
	\centering  
	\includegraphics[width=0.3\textwidth]{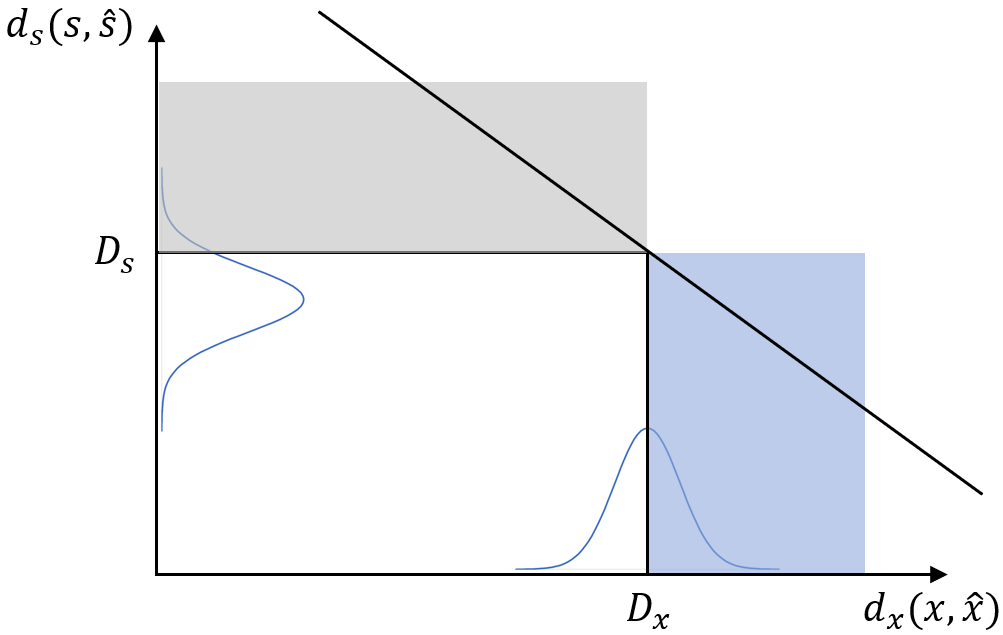}
	\caption{The geometry of the excess-distortion probability calculation. }
	\label{Fig_approximation} 
\end{figure}  
\begin{IEEEproof}
	This corollary can be proved by analyzing the monotonicity of the Gaussian approximation bound.
	The intuitive explanation of this corollary is as illustrated in Fig. \ref{Fig_approximation}. When $\mathsf{D}_x+\sigma^2_{N_s}-\mathsf{D}_s<0$, the distortion constraint on the semantic information becomes an inactive constraint, and the expectation of $\mathsf{d}_s(S,\hat{S})$ lies below the constraint $\mathsf{D}_s$.  Meanwhile, the variance of $\mathsf{d}_s(S,\hat{S})$ is strictly larger than the variance of $\mathsf{d}_x(X,\hat{X})$ in this case. Hence, the probability mass of the grey area will be much larger than that of the blue area. Thus, the optimal linear relaxation function $\mathsf{D}_L$ should be the one which includes the blue area and { excludes the grey area}, namely the vertical line. 
\end{IEEEproof}

\begin{Remark}
	Since there exists a Markov chain $S-X-\hat{X}-\hat{S}$ in the case of transmitting GMS over AWGN channels, which is proved in \cite[Appendix III]{Liu-semantic-2022}, Corollary \ref{achievability bound-simplified} rather than Theorem \ref{achievability bound-JSCC2} is used to establish the achievability bound.
	This fact leads to the fact that we can decode the semantic information $S$ functionally according to $\hat{X}$ as the reconstruction of the observable source. More precisely, we choose a simple decoding function for the semantic information as $\hat{s}=h(\hat{x})=\hat{x}$. 
	Since for this particular codebook, we only need to optimizing with the codebook size $M_1$ with $M_2=1$, hence we only have one parameter $\gamma$ to optimize with. 
\end{Remark}

\begin{figure}[!htpb]
	\centering  
	\includegraphics[width=0.5\textwidth]{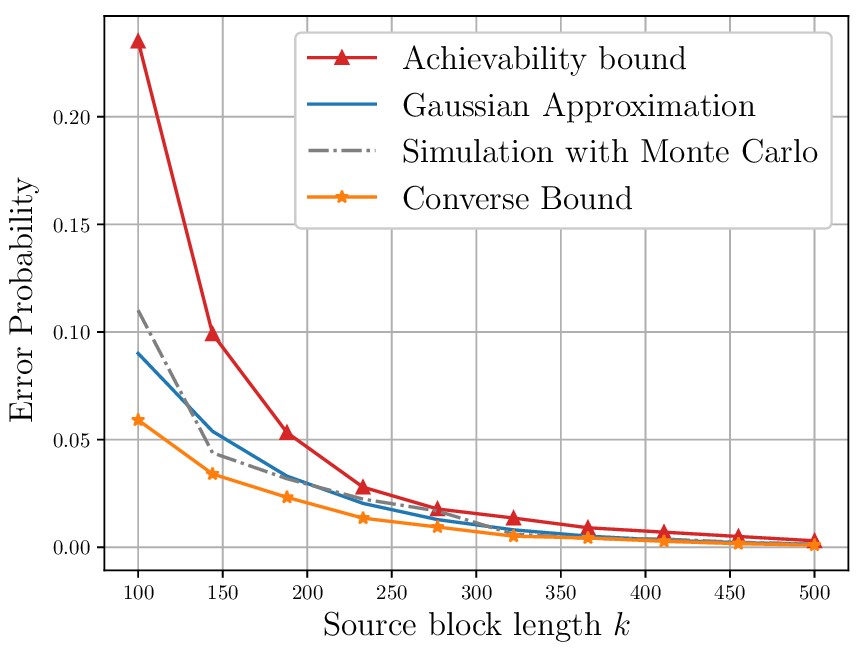}
	\caption{The geometry of the excess-distortion probability calculation. }
	\label{error vs blocklength} 
	\end{figure}  
 
 In order to show the proposed results more directly, we plot the numerical evaluation of the proposed bounds on excess-distortion probability over different blocklengths in Fig. \ref{error vs blocklength}. Specifically, we set the channel SNR as 0 dB, the variance of the observable source as $\sigma^2_X=1$, the variance of $N_s$ as $\sigma^2_{N_s}=0.5$, the distortion constraint for reconstructing the observable source as $\mathsf{D}_x=0.6$, and the distortion constraint for reconstructing the semantic source as $\mathsf{D}_s=1$. Moreover, we set the ratio $n/k=1.2$, and vary $k$ from 100 to 500 with a step size of 50.
In general, we can see that our proposed inner and converse bounds sandwich the Gaussian approximation bound, and converge to zero when the blocklength increases. Specifically, when the blocklength $k=100$, the gap between the achievability bound and the Gaussian approximation is about 0.119, and that between the Gaussian approximation and the converse bound is 0.028. However, the values of these two gaps decrease to 0.002 and 0.004 respectively when $k=500$.

\section{Conclusion}
{ 
This paper investigates computable bounds for joint source-channel coding with hierarchical sources in the finite blocklength regime. To address the challenge of precisely characterizing the joint excess-distortion probability, we introduce a novel $\mathsf{d}(\cdot)$-based distortion relaxation framework. 
Based on this formulation, the nonasymptotic  achievability and converse bounds are derived for  one-shot transmission. 
Furthermore, for the block-wise transmission, where the source is an i.i.d. sequence of length $k$, which is encoded into a sequence of length $n$ and then transmitted over a memoryless channel, Gaussian approximations of both achievability and converse bounds are derived via a two-layer superposition coding scheme, together with the optimization of a sum of correlated $Q$-functions, thereby explicitly capturing the statistical dependence between semantic and observable distortions. 
In particular, the proposed Gaussian approximation is shown to be optimal for the transmission of GMS over AWGN channels. 
Furthermore, for the transmission of a GMS over AWGN channels, we establish explicit bounds by leveraging a novel geometric characterization involving three correlated spherical regions, which generalizes the classical two-region framework used in single-source excess-distortion analysis.
}

\begin{appendices}
\section{Proof of Theorem \ref{Thm_converse}}\label{appendix-proof of general converse}
	
     For any $\gamma \geq 0$, we have
\begin{align}
& \mathbb{P} \big[\jmath_{S,X}(S, X) -\imath_{Y;Z} (Y; Z) \geq \gamma \big] \nonumber \\
=& \mathbb{P} \big[ \jmath_{S,X}(S, X)  -\imath_{Y;Z} (Y; Z) \geq \gamma, \mathsf{d}(\mathsf{d}_s({s}, {\hat{s}}),\mathsf{d}_x({x}, {\hat{x}}))>\mathsf{d}(\mathsf{D}_s,\mathsf{D}_x) \big] \nonumber \\
	& +  \mathbb{P} \big[ \jmath_{S,X}(S, X)  -\imath_{Y;Z} (Y; Z)\geq \gamma, \mathsf{d}(\mathsf{d}_s({s}, {\hat{s}}),\mathsf{d}_x({x}, {\hat{x}})) \leq \mathsf{d}(\mathsf{D}_s,\mathsf{D}_x) \big] \nonumber \\
	\leq&  \mathbb{P} \big[ \mathsf{d}(\mathsf{d}_s({s}, {\hat{s}}),\mathsf{d}_x({x}, {\hat{x}}))>\mathsf{d}(\mathsf{D}_s,\mathsf{D}_x) \big] \nonumber \\
	&+ \underset{s,x,y,\hat{x},\hat{s}}{\sum} \mathbb{P} \big[ \mathsf{d}(\mathsf{d}_s({s}, {\hat{s}}),\mathsf{d}_x({x}, {\hat{x}}))\leq \mathsf{d}(\mathsf{D}_s,\mathsf{D}_x) \big] 
	\mathds{1} \big[ \jmath_{S,X}(s,x) -\imath_{Y;Z} (y; z)\geq \gamma \big] \nonumber \\ 
	\overset{(a)}{\leq} & \epsilon+\underset{s,x,y,\hat{x},\hat{s}}{\sum} \mathbb{P} \big[ \mathsf{d}_s({s}, {\hat{s}})\leq \mathsf{D}_s \cup  \mathsf{d}_x({x}, {\hat{x}})\leq \mathsf{D}_x \big] 
	\mathds{1} \big[ \jmath_{S,X}(s,x) -\imath_{Y;Z} (y; z)\geq \gamma \big] \nonumber \\ 
\overset{(b)}{\leq} & \epsilon
	+ \sum_{(s,x)} P_{SX}(sx) \sum_{y} P_{Y|X}(y|x) \sum_{z} \sum_{\hat{s}\in B_{\mathsf{D}_s}(s)} \sum_{\hat{x}\in B_{\mathsf{D}_x}(x)} P_{\hat{X}|Z}(\hat{x}|z) P_{\hat{S}|Z}(\hat{s}|z) P_{Z|Y}(z|y) \nonumber \\
	& \cdot \mathds{1} \big[ P_{Z|Y}(z|y)\leq P_{\bar{Z}}(z) \exp(\jmath_{S,X}(s,x) -\gamma) \big] \nonumber \\ 
\leq & \epsilon 
	+ \exp(-\gamma) \sum_{(s,x)} P_{SX}(sx) \exp(\jmath_{S,X}(s,x))
	\sum_{z} P_{\bar{Z}}(z)  \sum_{\hat{s}\in B_{\mathsf{D}_s}(s)} P_{\hat{S}|Z}(\hat{s}|z) \sum_{\hat{x}\in B_{\mathsf{D}_x}(x)}  P_{\hat{X}|Z}(\hat{x}|z)  \nonumber \\
	& \cdot \sum_{y} P_{Y|X}(y|x) \nonumber  \\
	\leq & \epsilon 
	+ \exp(-\gamma) \sum_{(s,x)} P_{SX}(sx) \exp(\jmath_{S,X}(s,x))
	\sum_{z} P_{\bar{Z}}(z)  \sum_{\hat{s}\in B_{\mathsf{D}_s}(s)} P_{\hat{S}|Z}(\hat{s}|z) \sum_{\hat{x}\in B_{\mathsf{D}_x}(x)}  P_{\hat{X}|Z}(\hat{x}|z) \nonumber \\
= & \epsilon+ \exp(-\gamma) \sum_{(s,x)} P_{SX}(sx) \exp(\jmath_{S,X}(s,x))
	P_{\hat{S}}(B_{\mathsf{D}_s}(s)) P_{\hat{X}}(B_{\mathsf{D}_x}(x)) \nonumber \\
	\overset{(c)}{\leq} & \epsilon+ \exp(-\gamma) \sum_{(s,x)} P_{SX}(sx)  \sum_{\hat{s}\in \hat{S}} P_{\hat{S}}(\hat{s})  \sum_{\hat{x}\in \hat{X}} P_{\hat{X}}(\hat{x})   
				\exp\left(\jmath_{S,X}(s,x)+\lambda^{\star} \big(\mathsf{d}(\mathsf{D}_s,\mathsf{D}_x)-\mathsf{d}(\mathsf{d}_s({s}, {\hat{s}}),\mathsf{d}_x({x}, {\hat{x}}))\big) \right) \nonumber \\
\overset{(d)}{\leq} & \epsilon_1 + \exp(-\gamma). \label{eq-main thm-converse}  
				\end{align}
where step $(a)$ follows from the sufficiency of relaxed excess distortion event shown in  \eqref{sufficient}; step $(b)$ follows from the definition of information density function; step $(c)$ follows from $\mathsf{d}(\mathsf{d}_s({s}, {\hat{s}}),\mathsf{d}_x({x}, {\hat{x}}))-\mathsf{d}(\mathsf{D}_s,\mathsf{D}_x)>0$ if $\hat{x}\notin B_{\mathsf{D}_x}(x)$ and $\hat{s}\notin B_{\mathsf{D}_s}(s)$, and the exponential term is larger than zero;
step $(d )$ follows from the property of $\mathsf{d}$-tilted information that \begin{align*}
    \mathbb{E} \left[ \exp\left(\jmath_{S,X}(s,x)+\lambda^{\star} \big(\mathsf{d}(\mathsf{D}_s,\mathsf{D}_x)-\mathsf{d}(\mathsf{d}_s({s}, {\hat{s}}),\mathsf{d}_x({x}, {\hat{x}}))\big) \right)   \right]\leq 1.
\end{align*}
By optimizing over all possible distributions, we complete the proof Theorem \ref{Thm_converse} by \eqref{eq-main thm-converse}.

\section{Proof of Theorem \ref{achievability bound-JSCC2}} \label{analyze ex-dis prob}
In this section, we analyze the excess-distortion probability of the proposed coding scheme in Section \ref{achi-scheme}. For the given source tuple $(s,x) \in \mathcal{X}\times\mathcal{S}$ which is encoded as channel input $y_m$, an excess distortion event may happen when the following three error events occur:
\begin{itemize}
    \item \textit{Error type I} (Error with channel transmission): Namely, the output of the channel decoder is not the correct channel input. We denote this event by $\mathcal{E}_1 \triangleq \{m_1 \neq \arg \underset{i}{\max} P(z|y_{1i}) \cup m_2 \neq \arg \underset{i}{\max} P(z|y_{1i}, y_{2i}) \}$. 
    \item \textit{Error type II} (Odd observable source): The output of observable source is an ``odd'' one so that the source encoder failed to find any reconstruction $\hat{x}$ in the codebook  {$\hat{x}^{M_1}$ which satisfies the distortion constraint. We denote this event by $\mathcal{E}_2 \triangleq \{ m_1=M_1\}$}. 
    \item \textit{Error type III} (Odd semantic source): The semantic source is an ``odd'' one so that the reconstruction of semantic source exceeds the distortion requirement. We denote this event as $\mathcal{E}_3 \triangleq\{ \mathsf{d}(s,\hat{s}_{m_1,m_2})>\mathsf{D}_s\}$. 
\end{itemize}
Note that an excess distortion event may not occur even if these three error events occur, hence the probability of these error events provides an upper bound of the excess-distortion probability. 
That is to say, we will derive the upper bound of $P\{\mathcal{E}_1 \cup \mathcal{E}_2 \cup \mathcal{E}_3|X\}$ to characterize the excess distortion rate for observable source output $x$. Obviously, we have \begin{align}
  \epsilon &\leq \mathbb{E} \left[  P\{\mathcal{E}_1 \cup \mathcal{E}_2 \cup \mathcal{E}_3|X\}\right] \nonumber\\
  &\leq \mathbb{E} \left[ P\{\mathcal{E}_1|X\}\right]+ \mathbb{E} \left[ P\{\mathcal{E}_2 |X\}\right]  + \mathbb{E}\left[ P\{\mathcal{E}_3|\Bar{\mathcal{E}_2},X\}\right],\label{achievability bound-general-mid}
\end{align} 
Hence, in the following we will derive an upper bound for the three terms in the RHS of \eqref{achievability bound-general-mid} one by one. 
\begin{itemize}
	\item[(1)] Since a SSCC framework is used here, and the type I error is only about channel transmission error event, which is independent of the source, the error probability is the same as \cite[Theorem 16]{Yury-10}. 
	Thus, by applying the result in \cite{Yury-10}, the error probability of type I error can be upper bounded as  
	\begin{align}
		\mathbb{E} \left[ P\{\mathcal{E}_1|X\}\right] \leq \underset{P_{Y_1Y_2}}{\inf} \, \left\{  \mathbb{E}[\exp(-|\imath(Y_1;Z)-\log(M_1)|^{+})]
		+ \mathbb{E}[\exp(-|\imath(Y_2;Z|Y_1)-\log(M_2)|^{+})]\right\}. \label{Equ_errorI}
	\end{align}
	
	\item[(2)] For some output source symbol $x$, the type II error happens when all codewords in the codebook $\hat{x}^{M_1}$ exceed the $\mathsf{D}_x$ distortion constraint with $x$. That is to say,  $\hat{x}_i \notin B_{\mathsf{D}_x}\left(x\right) $ for all $i\in [1:M_1]$. 
	Since the codebook is independently generated with probability $P_{\hat{X}}$, hence the probability of one generated codeword exceeding the distortion constraint is $1-P_{\hat{X}}\left(B_{\mathsf{D}_x}\left(x\right)\right)$, and the probability of all codewords exceeding the distortion constraint is
	\begin{align}
		\left(1-P_{\hat{X}}\left(B_{\mathsf{D}_x}\left(x\right)\right)\right)^{M_1}. 
	\end{align}
	By averaging over all $x \in \mathcal{X}$, the error probability of type II error can be upper bounded as 
	\begin{align}
		\mathbb{E} \left[ P\{\mathcal{E}_2 |X\}\right] \leq  \underset{P_{\hat{X}}}{\inf} \, \mathbb{E} \left[ \left(1-P_{\hat{X}}\left(B_{\mathsf{D}_x}\left(x\right)\right)\right)^{M_1}\right]. \label{Equ_errorII}
	\end{align}
	
	\item[(3)]     For some source output $x$, it is first encoded as some $\hat{x}$ then encoded as an $\hat{s}$ in the conditional codebook of $\hat{x}$. Hence, the probability of excess distortion depends on $\hat{x}$. Thus, we have 
	\begin{align}
		P\{\mathcal{E}_3|\Bar{\mathcal{E}_2},x\}= \sum_{\hat{x}\in B_{\mathsf{D}_x}(x)} P\{\mathcal{E}_3|g_O(f_S(x))=\hat{x},x\} P\{g_O(f_S(x))=\hat{x}\}\label{lemma-achi-mid1},
	\end{align}
	where $P\{g_O(f_S(x)))=\hat{x}\}$ denotes the probability of $x$ being encoded as $\hat{x}$.
	
	We first focus on $P\{g_O(f_S(x))=\hat{x}\}$. If $x$ is encoded as some $\hat{x}$, it means that $\hat{x}$ is the codeword satisfying the $\mathsf{D}_x$ distortion requirement and with the smallest index. It indicates that all codewords generated before this particular $\hat{x}$ are not satisfying the distortion constraint. Hence, we have 
	\begin{align}
		P\{g_O(f_S(x))=\hat{x}\}=P(\hat{x}) \sum_{i=0}^{M_1-1}\left(1-P_{\hat{X}}\left(B_{\mathsf{D}_x}\left(x\right)\right)\right)^{i}=P(\hat{x})T(x)\label{lemma-achi-mid2}.
	\end{align}
	Note that $T(x)=\sum_{i=0}^{M_1-1}\left(1-P_{\hat{X}}\left(B_{\mathsf{D}_x}\left(x\right)\right)\right)^{i}$ is a function of $x$ and independent of the value of $\hat{x}$.
	Then we focus on $P\{\mathcal{E}_3|g_O(f_S(x))=\hat{x},x\}$, thus the excess-distortion probability for reconstructing the semantic information if the source output $x$ is firstly encoded as some $\hat{x}$. We have
	\begin{align}
		&P\{\mathcal{E}_3|g_O(f_S(x))=\hat{x},x\}\nonumber\\
		=&P\{d(S,g_S(f_S(X))>\mathsf{D}_s  |g_O(f_S(x))=\hat{x},x\}\nonumber\\
		\overset{(a)}{=}& \mathbb{E} \left[\pi(x,g_S(f_S(x))|\hat{x}\right]\nonumber\\
		\overset{(b)}{=}& \int_{0}^{1} P \left\{ \pi(x,g_S(f_S(x))|\hat{x})>t\right\}dt. \label{lemma-achi-mid3}
	\end{align}
	The expectation in step $(a)$ is to average with respect to the randomness of codebook. Step $(b)$ is because the function $\pi(\cdot)$ takes value from $[0,1]$. 
		
	By recalling codebook generation in Section \ref{achi-scheme}, the excess-distortion probability being larger than $t$ means the all codewords in the conditional satellite codebook have an excess-distortion probability larger than $t$. Equivalently, it can be obtained that 
	\begin{align}
		&\mathds{1} \{\pi(x,g_S(f_S(x))|\hat{x})>t\}\nonumber\\
		= &\mathds{1} \left\{ \underset{i\in [1:M_2]}{\min} \pi(x,\hat{s}_{m_1,i}|\hat{x})>t\right\}\nonumber\\
		=& \Pi_{i=1}^{M_2}\mathds{1} \left\{  \pi(x,\hat{s}_{m_1,i}|\hat{x})>t\right\}\label{lemma-achi-mid4}.
	\end{align}
	
	Therefore, 
	\begin{align}
		&\mathbb{E} \left[ P\{\mathcal{E}_3|\Bar{\mathcal{E}_2},X\}\right] \nonumber\\
		\overset{(a)}{=} & \underset{P_{\hat{S}|\hat{X}}}{\inf}\,\mathbb{E} \left[ T(X)\sum_{\hat{x}\in B_{\mathsf{D}_x}\left(X\right)} P(\hat{x})P\{d(S,g_S(f_S(X)))>\mathsf{D}_s   |g_O(f_S(X))=\hat{x}\}\right] \nonumber\\
		\overset{(b)}{=}& \underset{P_{\hat{S}|\hat{X}}}{\inf}\,\mathbb{E} \left[ T(X)\sum_{\hat{x}\in B_{\mathsf{D}_x}\left(X\right)} P(\hat{x})\mathbb{E}\left[\pi(X,g_S(f_S(x))|\hat{x})\right]\right] \nonumber\\
		\overset{(c)}{=}&\underset{P_{\hat{S}|\hat{X}}}{\inf}\,\mathbb{E} \left[ T(X)\sum_{\hat{x}\in B_{\mathsf{D}_x}\left(X\right)} P(\hat{x})\int_{0}^{1} P^{M_2} \left\{ \pi(X,\hat{S}|\hat{x})>t\right\}dt\right], \label{Equ_errorIII}
	\end{align}
	where $(a)$ follows from  \eqref{lemma-achi-mid2}, $(b)$ is from  \eqref{lemma-achi-mid3}, and $(c)$ is obtained by substituting  \eqref{lemma-achi-mid4} into  \eqref{lemma-achi-mid3}.       
	
	Note  that the proof of type III error is different from the conventional work on indirect source coding \cite{Kostina-lossyJSCC-13}. This is mainly because we use a conditional satellite codebook to encode the semantic information $S$ instead of encoding it directly as \cite{Kostina-lossyJSCC-13}.  
\end{itemize}

By adding the probabilities of three error types (i.e. \eqref{Equ_errorI}, \eqref{Equ_errorII}, and \eqref{Equ_errorIII}) together, Theorem \ref{achievability bound-JSCC2} is obtained. 

\section{Proof of Corollary \ref{achievability bound-simplified}}\label{appendix-proof of achi}
The corollary can be derived from Theorem \ref{achievability bound-JSCC2} with the setup of $M_1=M$ and $M_2=1$.
We first focus on the term $T(X)\sum_{\hat{x}\in B_{\mathsf{D}_x}\left(X\right)} P(\hat{x})\int_{0}^{1} P^{M_2} \{ \pi(X,\hat{S}|\hat{x})>t\}dt$ on the RHS of  \eqref{Equ_Inner_GA1}, which gives the probability of Type III error defined in Section \ref{analyze ex-dis prob}.

First, from the codebook generation, we can find that for any given $\hat{x}$ the satellite codebook is deterministic. It means that 
\begin{equation}
     P\big\{ \pi(X,\hat{S}|\hat{x})>t\big\}\\
     =\mathds{1}\big\{ P\{\mathsf{d}_s(S,h(\hat{x}))>\mathsf{D}_s\}>t\big\},
\end{equation}
and hence 
\begin{align}
    &\int_{0}^{1} P^{M_2} \big\{ \pi(X,\hat{S}|\hat{x})>t\big\}dt\nonumber\\
    =&\int_{0}^{1} P \big\{ \pi(X,\hat{S}|\hat{x})>t\big\}dt\nonumber\\
    =& \int_{0}^{1} \mathds{1}\big\{ P\{\mathsf{d}_s(S,h(\hat{x}))>\mathsf{D}_s\}>t\big\}dt\nonumber\\
    =&  P\big\{\mathsf{d}_s(S,h(\hat{x}))>\mathsf{D}_s\big\}.\label{appendixa-1}
\end{align}
For any given $\hat{x}$, the probability of $\hat{s}=h(\hat{x})$ exceeding the distortion constraint $\mathsf{D}_s$ only related to the randomness of $S$.
Meanwhile, we have 
\begin{align}
    T(X)=&\sum_{i=0}^{M-1}\left(1-P_{\hat{X}}\left(B_{\mathsf{D}_x}\left(X\right)\right)\right)^{i}\nonumber\\
    =&\frac{1-\left(1-P_{\hat{X}}\left(B_{\mathsf{D}_x}\left(X\right)\right)\right)^{M}}{1-\left(1-P_{\hat{X}}\left(B_{\mathsf{D}_x}\left(X\right)\right)\right)}\nonumber\\
    \leq & \frac{1}{P_{\hat{X}}\left(B_{\mathsf{D}_x}\left(X\right)\right)}.\label{appendixa-2}
\end{align}
The probability of some $x$ being encoded as $\hat{x}$ within the distortion constraint $\mathsf{D}_x$ is given by $T(x)P(\hat{x})$, which is upper bounded by $P(\hat{x}) / P_{\hat{X}}\!\left(B_{\mathsf{D}_x}(x)\right)$, i.e., the probability of selecting $\hat{x}$ from the distortion ball $B_{\mathsf{D}_x}(x)$.
By combining the results in  \eqref{appendixa-1} and  \eqref{appendixa-2} together, we have that
\begin{align}
    &T(X)\sum_{\hat{x}\in B_{\mathsf{D}_x}\left(X\right)} P(\hat{x})\int_{0}^{1} P^{M_2} \big\{ \pi(X,\hat{S}|\hat{x})>t\big\}dt\nonumber\\
    \leq &\sum_{\hat{x}\in B_{\mathsf{D}_x}\left(X\right)}\frac{P(\hat{x})}{P_{\hat{X}}\left(B_{\mathsf{D}_x}\left(X\right)\right)}P\big\{\mathsf{d}_s(S,h(\hat{x}))>\mathsf{D}_s\big\}, 
\end{align}
where the RHS is the probability that an $\hat{x}$ within $B_{\mathsf{D}_x}\left(x\right)$ is picked and its functional reconstruction $\hat{s}$ exceeds the $\mathsf{D}_s$ distortion constraint of the semantic information $S$.

By letting $M=\lfloor \frac{\gamma}{P_{\hat{X}}\left(B_{\mathsf{D}_x}\left(x\right)\right) } \rfloor$ for some $\gamma > 0$, it follows that
\begin{align}
	\left(1-P_{\hat{X}}\left(B_{\mathsf{D}_x}\left(X\right)\right)\right)^{M_1}
	\leq & \left(1-P_{\hat{X}}\left(B_{\mathsf{D}_x}\left(X\right)\right)\right)^{\frac{\gamma}{P_{\hat{X}}\left(B_{\mathsf{D}_x}\left(x\right)\right) }-1} \nonumber\\
	\leq & \exp \left\{- P_{\hat{X}}\left(B_{\mathsf{D}_x}\left(X\right)\right) \left(\frac{\gamma}{P_{\hat{X}}\left(B_{\mathsf{D}_x}\left(x\right)\right) }-1\right)   \right\} \nonumber\\
	\leq & \exp\{1-\gamma\}
\end{align}
    where the last inequality utilizes $(1-a)^{1/a-1} \leq e^{a-1} \leq e^{1}$ for small $a$. Combining these bounds leads to the final expression in Corollary \ref{achievability bound-simplified}.

\section{Proof of Theorem  \ref{Thm_Gau_Converse1} } \label{appendix-proof of gaussian approximation linear}
	Following Appendix C in \cite{Kostina-lossyJSCC-13}, we ignore all rates exceeding the bound
	\begin{align}
		\frac{k}{n} \geq \frac{C(\rho)}{R_{S,X}(\mathsf{D}_s,\mathsf{D}_x)-3\tau}, \label{Def_Converse_kn}
	\end{align}
	for any $0<\tau< R_{S,X}(\mathsf{D}_s,\mathsf{D}_x)/3$.
	Specifically, the excess-distortion probability of any code with such rate will converge to 1 as $n\to \infty$. Hence, for any $\epsilon<1$, there exists a blocklength threshold $n_0$ such that no $(M, k,n, \mathsf{d}_s, \mathsf{d}_x, \epsilon)$  code can exist for $k, n$ satisfying \eqref{Def_Converse_kn} for all $n\geq n_0$.
 
	Denote by $C(\rho)=\max_{P_Y} \mathbb{E}[\imath_{Y;Z}(Y;Z)]$ the channel capacity, and let $\mathcal{P}^{\star}$ be the set of capacity-achieving distributions, i.e.
	\begin{align}
		\mathcal{P}^{\star}
		= \left\{  P_Y:  \mathbb{E}[\imath_{Y;Z}(Y;Z)] =C(\rho) \right\}.
	\end{align}
	
	If either $V(\rho)>0$ or $\tilde{V}(\rho)>0$, let
	\begin{align}
		\gamma= & \frac{1}{2} \log (n+1), \label{Equ_Converse_gamma}\\
		\epsilon_{k,n}=&\epsilon+\frac{B}{\sqrt{n+k}}+\frac{1}{\sqrt{n}}+\frac{4C(\rho)^2 }{R^2(\mathsf{D}_s,\mathsf{D}_x)\bar{\Delta}^2 } \frac{\tilde{V}}{k}, \label{Equ_Converse_err}
	\end{align}
	where $B>0$ can be chosen in the sequel, and $k, n$ are chosen so that there is a sequence of $(k,n,\mathsf{d}_s,\mathsf{d}_x,\epsilon')$ codes such that
	\begin{align}
		-3k\tau 
		\leq nC(\rho)- kR_{S,X}(\mathsf{D}_s,\mathsf{D}_x) 
		\leq \sqrt{nV(\rho)+k\tilde{V}(\rho)}Q^{-1}(\epsilon_{k,n})
		-\gamma, \label{Equ_Converse_Cons}
	\end{align}
	then $\epsilon'\geq \epsilon_{k,n}$.
	Theorem \ref{Thm_converse} implies that the error probability of every $(k,n,\mathsf{d}_s,\mathsf{d}_x,\epsilon')$ code must satisfy for an arbitrary sequence $x^n\in \mathcal{X}^n$
	\begin{align}
		\epsilon' \geq & 
	    \sup_{\lambda\geq 0} \left\{ \mathbb{P} \left[ \sum_{i=1}^{k} \jmath_{S,X}(s_i, x_i)  - \sum_{i=1}^{n}\imath_{Y;Z} (z_i; y_i)\geq \gamma \right]
		-\exp(-\gamma)\right\}. \label{Equ_Def_err1}
	\end{align}

In order to apply Lemma 1, we isolate the typical set of source sequences 
	\begin{align}
		\mathcal{T}_{k,n}
		=\left\{ s^k\in \mathcal{S}^k:  \left|  \sum_{i=1 }^{k} \jmath_{S,X}(s_i, x_i)-nC(\rho)   \right| \leq n \bar{\Delta}  -\gamma \right\},
	\end{align}
	for some constant $\bar{\Delta}>0$. 
	Observe that
	\begin{align}
		&\mathbb{P} \left[ S^k \notin  \mathcal{T}_{k,n} \right] \nonumber\\
		=&   \mathbb{P} \left[ \left|  \sum_{i=1 }^{k} \jmath_{S,X}(s_i, x_i))-nC(\rho)   \right| > n \bar{\Delta}  -\gamma\right] \nonumber\\
		\leq &  \mathbb{P} \left[ \left|  \sum_{i=1 }^{k} \jmath_{S,X}(s_i, x_i))-kR_{S,X}(\mathsf{D}_s,\mathsf{D}_x) \right|  
			+  \left|  nC(\rho)-kR_{S,X}(\mathsf{D}_s,\mathsf{D}_x) \right| + \gamma> n \bar{\Delta}  \right] \nonumber\\
		\overset{(a)}{\leq} &  \mathbb{P} \left[ \left|  \sum_{i=1 }^{k} \jmath_{S,X}(s_i, x_i))-kR_{S,X}(\mathsf{D}_s,\mathsf{D}_x) \right| > k \frac{\bar{\Delta} R_{S,X}(\mathsf{D}_s,\mathsf{D}_x)}{2C(\rho)} \right] \nonumber\\
		 \overset{(b)}{ \leq} & \frac{4C(\rho)^2 }{R^2(\mathsf{D}_s,\mathsf{D}_x)\bar{\Delta}^2 } \frac{\tilde{V}}{k} \label{Equ_Converse_NoTy2}
	\end{align}
	where step $(b)$ holds by Chebyshev's inequality, 
	 and step $(a)$ comes from the fact that
  \begin{align}
      n\bar{\Delta}-\gamma- \left|  nC(\rho)-kR_{S,X}(\mathsf{D}_s,\mathsf{D}_x) \right|\geq k \frac{\bar{\Delta} R_{S,X}(\mathsf{D}_s,\mathsf{D}_x)}{2C(\rho)}.\label{Equ_Converse_NoTy14}
  \end{align}
  The inequality of \eqref{Equ_Converse_NoTy14} can be proved as follows
	\begin{align}
		& n\bar{\Delta}-\gamma- \left|  nC(\rho)-kR_{S,X}(\mathsf{D}_s,\mathsf{D}_x) \right| \nonumber\\
		\overset{(a)}{\geq} & n\bar{\Delta}-\gamma-3k\tau  \nonumber\\
		\overset{(b)}{\geq} & n\frac{3\bar{\Delta}}{4}-3k\tau \nonumber\\
		\overset{(c)}{\geq} & k\frac{3\bar{\Delta}}{4C(\rho)}(R_{S,X}(\mathsf{D}_s,\mathsf{D}_x) -3\tau) -3k\tau \nonumber\\
		\overset{(d)}{\geq} & k \frac{\bar{\Delta} R_{S,X}(\mathsf{D}_s,\mathsf{D}_x)}{2C(\rho)} \nonumber.
	\end{align}
	where step $(a)$ holds for large enough $n$ due to \eqref{Equ_Converse_Cons};
	step $(b)$ holds for large enough $n$ by the choice of $\gamma$ in \eqref{Equ_Converse_gamma};
	step $(c)$ is by \eqref{Equ_Converse_Cons};
	and step $(d)$ holds for a small enough $\tau>0$.
	
By \cite[Lemma 1]{Kostina-lossyJSCC-13}, $\epsilon'$ can be lower bounded by
	 \begin{align}
	 	\epsilon' \geq & 
	 	 \mathbb{E} \left[ \min_{x^n\in \mathcal{X}^n} \mathbb{P} \left[ \sum_{i=1}^{k} \jmath_{S,X}(s_i, x_i))  - \sum_{i=1}^{n}\imath_{Y;Z}(y_i; z_i)\geq \gamma | S^k\right] 
	 	 \cdot 1\left\{S^k\in \mathcal{T}_{k,n} \right\}\right]-\frac{1}{\sqrt{n+1}} \\
	 	 \geq & \mathbb{E} \left[ \mathbb{P} \left[ \sum_{i=1}^{k} \jmath_{S,X}(s_i, x_i))  - \sum_{i=1}^{n}\imath_{Y;Z}(y_i; z_i)\geq \gamma | S^k\right] 
	 	 \cdot 1\left\{S^k\in \mathcal{T}_{k,n} \right\}\right]-\frac{K}{\sqrt{n}}-\frac{1}{\sqrt{n+1}}   \label{Equ_Converse_4}\\
	 	 = & \mathbb{P} \left[ \sum_{i=1}^{k} \jmath_{S,X}(s_i, x_i))  - \sum_{i=1}^{n}\imath_{Y;Z}(y_i; z_i)\geq \gamma, S^k\in \mathcal{T}_{k,n}  | S^k\right] -\frac{K}{\sqrt{n}}-\frac{1}{\sqrt{n+1}} \\
	 	 \geq & \mathbb{P} \left[ \sum_{i=1}^{k} \jmath_{S,X}(s_i, x_i))-\sum_{i=1}^{n}\imath_{Y;Z}(y_i; z_i)\geq \gamma \right] 
	 	 	 - \mathbb{P} \left[ S^k\notin \mathcal{T}_{k,n}  | S^k\right] -\frac{K}{\sqrt{n}}-\frac{1}{\sqrt{n+1}}  \label{Equ_Converse_5}\\
	 	 \geq &  \mathbb{P} \left[ \sum_{i=1}^{k} \jmath_{S,X}(s_i, x_i))-\sum_{i=1}^{n}\imath_{Y;Z}(y_i; z_i)\geq \gamma \right] 
	 	 -\frac{4C(\rho)^2 }{R^2(\mathsf{D}_s,\mathsf{D}_x)\bar{\Delta}^2 } \frac{\tilde{V}(\rho)}{k} -\frac{K}{\sqrt{n}}-\frac{1}{\sqrt{n+1}}  \label{Equ_Converse_6} \\
	 	 \geq  & \epsilon  \label{Equ_Converse_7}
	 \end{align}
	 where \eqref{Equ_Converse_4} is by \cite[Lemma 1]{Kostina-lossyJSCC-13};
	 \eqref{Equ_Converse_5} is by the union bound.
	  \eqref{Equ_Converse_7} is obtained by Berry-Esseen Theorem.
	  The expected value, the variance and the third moment of the random variables are given by
	  \begin{align}
	  		D_{n+k} = & \frac{n}{n+k} \mathbb{E}\left[ \imath_{Y;Z}(Y;Z)\right] -\frac{k}{n+k} R_{S,X}(\mathsf{D}_s,\mathsf{D}_x)\nonumber\\
	  			\leq & \frac{n}{n+k} C(\rho)-\frac{k}{n+k} R_{S,X}(\mathsf{D}_s,\mathsf{D}_x),  \label{Equ_Converse_D}\\
	  		V_{n+k} = &\frac{n}{n+k} \text{Var}\left( \imath_{Y;Z}(Y;Z)\right) -\frac{k}{n+k}  \text{Var}\left( \jmath_{S,X}(S,X)\right) \nonumber\\
	  			= & \frac{n}{n+k} V(\rho)-\frac{k}{n+k} \tilde{V}(\rho), \label{Equ_Converse_V} \\
	  		T_{n+k} = &\frac{n}{n+k} \left(\mathbb{E} \big[ \big| \imath_{Y;Z}(Y;Z)-\mathbb{E} [\imath_{Y;Z}(Y;Z)] \big|^3\big]  +\mathbb{E} \big[ \big| \jmath_{S,X}(S,X)-R_{S,X}(\mathsf{D}_s,\mathsf{D}_x)\big|^3\big]   \right). \label{Equ_Converse_T}
	  \end{align}
	 Applying \eqref{Equ_Converse_D} and \eqref{Equ_Converse_V} to \eqref{Equ_Converse_Cons}, we conclude that 
	 \begin{align}
	 	-\gamma \geq & (n+k)D_{n+k}-\sqrt{(n+k)V_{n+k}} Q^{-1}(\epsilon_{k,n}).
	 \end{align}
Thus, Theorem \ref{Thm_Gau_Converse1} can be proved. 

\section{Proof of Theorem  \ref{Thm_Gau_Ach1} }  \label{Appx_GN_Superposition}
By the coding scheme proposed in Sec. \ref{achi-scheme}, the overall transmission scheme decomposes into two JSCC processes of $\hat{X}$ and $\hat{S}$, where the decoding error probability is upper bounded by the first term and the last two terms on the RHS of  \eqref{Equ_Inner_GA1}.  
Next, we analyze the terms on the RHS of \eqref{Equ_Inner_GA1} to show that $\epsilon'\leq \epsilon$.  

\begin{itemize}
	\item[(1)] Following similar steps as in \cite[Appendix D]{Kostina-lossyJSCC-13}, by letting 
	\begin{align}
		M_1=&\left\lfloor \frac{\gamma_1}{P_{\hat{X}}\left(B_{\mathsf{D}_x}\left(X\right)\right)}\right\rfloor, \\
		\gamma_1=&\frac{1}{2}\log_e k+1,
	\end{align}
	and letting $k$ and $n$ be such that
	\begin{align} 
		nC_1(\rho)-k R_1(\mathsf{D}_x) 
		\geq  \sqrt{n V_1(\rho) + k \tilde{V}_1(\rho)} Q^{-1}(\epsilon_{k1})
		+ \bar{c}_1 \log k + \log\log\sqrt{k}  +c,  \label{Equ_GNA_TF12}
	\end{align}
	it can be obtained that 
	\begin{align}
		\mathbb{E}\left[ \exp\left( -|\imath_{Y_1;Z}(Y_1;Z)-\log(M_1)|^{+} \right) \right] + \mathbb{E} \bigg[ \left(1-P_{\hat{X}}\left(B_{\mathsf{D}_x}\left(X\right)\right)\right)^{M_1}\bigg]
		\leq & \epsilon_{k1}+ \frac{2}{\sqrt{k}} +\frac{B_1}{\sqrt{n+k}}. \label{Equ_Inner_GA_EF1}
	\end{align}

	\item[(2)] First, let 
	\begin{align}
		M_2\leq & \frac{\log_e \sqrt{k} }{P_{\hat{S}}\left(B_{\mathsf{D}_s}\left(X\right)\right)}, \label{Eq_M2_1} 
	\end{align}
	and let $k$ and $n$ be such that
	\begin{align} 
		nC_2(\rho)-k R_2(\mathsf{D}_s) 
		\geq  \sqrt{n V_2(\rho) + k \tilde{V}_2(\rho)} Q^{-1}(\epsilon_{k2})
		+ \bar{c}_2 \log k  +c.  \label{Equ_GNA_TF22} 
	\end{align}
	Following similar steps as in \cite[Lem. 5]{Kostina-lossyJSCC-13} and  \cite[Th. 8]{Kostina-lossyJSCC-13}, it follows that
	\begin{align}
		\mathbb{E}\left[ \exp\left( -|\imath(Y_2;Z|Y_1)-\log(M_2)|^{+} \right) \right] 
		\leq & \epsilon_{k2}+ \frac{1}{\sqrt{k}} +\frac{B_2}{\sqrt{n+k}}. \label{Equ_GNA_3}
	\end{align}
	
	Next, we upper bound the integral in the third term on the RHS of \eqref{Equ_Inner_GA1}  following similar steps as in \cite[Thm. 4 (59)]{Kostina-isc-2016}. 
	For any $\gamma_2\geq 0$, it yields that 
	\begin{align}
		\int_{0}^{1} P^{M_2} \left\{  {\pi(X,\hat{S}|\hat{x})}>t\right\} dt
		\leq& \int_{0}^{1} \left( e^{-\frac{M_2}{\gamma_2}} 
		+ \left|1- \gamma_2  P \left[  \pi(X,\hat{S}|\hat{x})>t \right]  \right|^+ \right)  dt	 \\
		=& e^{-\frac{M_2}{\gamma_2}}  
		+  \int_{0}^{1}    1\{ D(P_{\hat{S}} \| P_{\tilde{S}}) > \log \gamma_2 \}	dt  \\
		=&e^{-\frac{M_2}{\gamma_2}}  
		+ P\left[ \underset{ P_{\hat{S}}: \pi(x,\hat{S}|\hat{x})\leq t }{\inf}  D(P_{\hat{S}} \| P_{\tilde{S}})>  \log \gamma_2  \right], \label{Equ_GNA_T31}
	\end{align}
	where $t$ is uniformly distributed over $[0,1]$, and  the  second step is by  \cite[(63)]{Kostina-isc-2016}, and $D(\cdot\| \cdot)$ is the relative entropy.
	Next, we upper bound the second term on the RHS of \eqref{Equ_GNA_T31}.
	Define the  typical set  as
	\begin{align}
		\mathcal{T}_k = \left\{ x^k \in \mathcal{A}^k : \left| \frac{1}{k} \text{type}(x^k) - P_X \right|^2 \leq |\mathcal{A}|  \frac{\log k}{k} \right\},  
	\end{align}
	where $|\cdot|$ denotes the Euclidean norm. By \cite[Lem. 1]{Kostina-isc-2016}, we obtain that
	\begin{align}
		\mathbb{P}\left( X^k \notin \mathcal{T}_k \right) \leq \frac{2 |\mathcal{A}|}{\sqrt{k}}. 
	\end{align}
	Thus, \eqref{Equ_GNA_T31} can be upper bounded by
	\begin{align}
		&P\left[ \underset{ P_{\hat{S}}: \pi(x,\hat{S}|\hat{x})\leq t }{\inf}  D(P_{\hat{S}} \| P_{\tilde{S}})>  \log \gamma_2  \right]  \\
		=&	P\left[ \underset{ P_{\hat{S}}: \pi(x,\hat{S}|\hat{x})\leq t }{\inf}  D(P_{\hat{S}} \| P_{\tilde{S}})>  \log \gamma_2,  X^k \in \mathcal{T}_k  \right] 
		+	P\left[ \underset{ P_{\hat{S}}: \pi(x,\hat{S}|\hat{x})\leq t }{\inf}  D(P_{\hat{S}} \| P_{\tilde{S}})>  \log \gamma_2,  X^k \notin \mathcal{T}_k \right] \\
		\leq & 	P\left[ \underset{ P_{\hat{S}}: \pi(x,\hat{S}|\hat{x})\leq t }{\inf}  D(P_{\hat{S}} \| P_{\tilde{S}})>  \log \gamma_2,  X^k \in \mathcal{T}_k  \right] 
		+  \frac{2+2 |\mathcal{A}|}{\sqrt{k}}. \label{Equ_GNA_T35}
	\end{align}
	Then,  following similar steps of \cite[(180)-(191)]{Kostina-isc-2016}, we obtain that
	\begin{align}
		&P\left[ \underset{ P_{\hat{S}}: \pi(x,\hat{S}|\hat{x})\leq t }{\inf}  D(P_{\hat{S}} \| P_{\tilde{S}})>  \log \gamma_2,  X^k \in \mathcal{T}_k  \right]  \nonumber\\
		\leq & P\left[  \sum_{i=1}^k \jmath_{X}(X_i,\hat{S}_i|\hat{x} ) + \lambda_2 \sqrt{k  \tilde{V}_2(\rho)} G + \mathcal{O}(\log k) \geq \log \gamma_2 \right] \label{Equ_GNA_T32}
	\end{align}
	where $G \sim \mathcal{N}(0, 1)$. Denote $\sum_{i=1}^k \jmath_X(X_i,\hat{S}_i|\hat{x} ) + \lambda_2 \sqrt{k  \tilde{V}_2(\rho)} G + \mathcal{O}(\log k)\triangleq g(X^k)$. Let  
	\begin{align}
		\log \gamma_2 = &k R_2(\mathsf{D}_s) + \sqrt{k \tilde{V}_2(\rho)} Q^{-1}\left(\frac{1}{\sqrt{k}} \right). \label{Eq_gamma2}
	\end{align}
	By Berry-Esseen Theorem, we have
	\begin{align}
		\mathbb{P}\left[  \frac{g(X^k)-R_2(\mathsf{D}_s)}{\sqrt{k \tilde{V}_2(\rho)}} \geq \frac{\log \gamma_2-R_2(\mathsf{D}_s)}{\sqrt{k \tilde{V}_2(\rho)}} \right] 
		\leq & \frac{B_2}{\sqrt{k}} + Q\left( \frac{\log \gamma_2-R_2(\mathsf{D}_s)}{\sqrt{k \tilde{V}_2(\rho)}} \right) \\
		=&\frac{B_2}{\sqrt{k}} +\frac{1}{\sqrt{k}}
	\end{align}
	where $B_2$ is the Berry-Esseen ratio of $\sum_{i=1}^k\jmath_X(X_i,\hat{S}_i|\hat{x} )$.
	Thus, \eqref{Equ_GNA_T32} can be upper bounded by
	\begin{align}
		P\left[ \underset{ P_{\hat{S}}: \pi(x,\hat{S}|\hat{x})\leq t }{\inf}  D(P_{\hat{S}} \| P_{\tilde{S}})>  \log \gamma_2,  X^k \in \mathcal{T}_k  \right]   
		\leq &\frac{B_2}{\sqrt{k}} +\frac{1}{\sqrt{k}}. \label{Equ_GNA_T34}
	\end{align}
	By \eqref{Equ_GNA_T31}, \eqref{Equ_GNA_T35} and \eqref{Equ_GNA_T34}, we have
	\begin{align}
		\int_{0}^{1} P^{M_2} \left\{  {\pi(X,\hat{S}|\hat{x})}>t\right\} dt
		\leq& \frac{B_2+3+2 |\mathcal{A}|}{\sqrt{k}}.  \label{Equ_GNA_1}
	\end{align}
	
	Moreover, by
	\begin{align}
		\sum_{\hat{x} \in B_{\mathsf{D}_x}(X)} P(\hat{x}) &= P_{\hat{X}}({B_{\mathsf{D}_x}}(X)), \nonumber 
	\end{align}
	we obtain that
	\begin{align}
		T(x) \sum_{\hat{x} \in B_{\mathsf{D}_x}(X)} P(\hat{x}) 
		&=\sum_{i=0}^{M_1 - 1}   \left(1-P_{\hat{X}}({B_{\mathsf{D}_x}}(X))\right)^i P_{\hat{X}}({B_{\mathsf{D}_x}}(X)) \nonumber\\
		&= \frac{1-  \left(1-P_{\hat{X}}({B_{\mathsf{D}_x}}(X))\right)^{M_1}}{ 1- \left(1-P_{\hat{X}}({B_{\mathsf{D}_x}}(X))\right)}  P_{\hat{X}}({B_{\mathsf{D}_x}}(X)) \nonumber\\
		&= 1-  \left(1-P_{\hat{X}}(B_{{\mathsf{D}_x}}(X))\right)^{M_1} \nonumber \\
		&\leq 1. \label{Equ_GNA_T36}
	\end{align}
	
	Finally, let $\frac{M_2}{\gamma_2}\geq \log \sqrt{k}$. Then, by \eqref{Equ_GNA_1} and \eqref{Equ_GNA_T36}, the third term on the RHS of \eqref{Equ_Inner_GA1} can be upper bounded by 
	\begin{align}
		&	\mathbb{E}[\exp(-|\imath(Y_2;Z|Y_1)-\log(M_2)|^{+})]
		+ \mathbb{E}_{\sim P_X} \bigg[T(X)\sum_{\hat{x}\in B_{\mathsf{D}_x}\left(X\right)} P(\hat{x})\int_{0}^{1} P^{M_2} \left\{  {\pi(X,\hat{S}|\hat{x})}>t\right\}dt\bigg]  \nonumber\\
		\leq &  
		\epsilon_{k2}  +\frac{B_2}{\sqrt{n+k}} 
		+ \frac{B_2+4+2 |\mathcal{A}|}{\sqrt{k}}. \label{Equ_Inner_GA_EF2}
	\end{align}
\end{itemize}
Let $\epsilon_{k1}=\epsilon_{1}- \frac{2}{\sqrt{k}} -\frac{B_1}{\sqrt{n+k}}$, $\epsilon_{k2} =\epsilon_{2}  -\frac{B_2}{\sqrt{n+k}} - \frac{B_2+3+2 |\mathcal{A}|}{\sqrt{k}}$, and $\epsilon_{1}+\epsilon_{2}\leq \epsilon$.
Then combining \eqref{Equ_Inner_GA_EF1} and \eqref{Equ_Inner_GA_EF2}, we conclude that 
$\epsilon'\leq \epsilon$.  
	Next, we analyze the error probability $\epsilon_1+\epsilon_2$.
Referring to \eqref{Equ_Inner_GA_EF1} and \eqref{Equ_Inner_GA_EF2}, two bounds on $\epsilon_1$ and $\epsilon_2$ are established:
\begin{align} 
	\epsilon_1\geq & Q\left(\frac{nC_1(\rho)-k R_1(\mathsf{D}_x)-\theta_1(k)}{\sqrt{n V_1(\rho) + k \tilde{V}_1(\rho)}}  \right) 
	=Q\left(\frac{nC_1(\rho)-k R_1(\mathsf{D}_x)}{\sqrt{n V_1(\rho) + k \tilde{V}_1(\rho)}}  \right)+\tilde{ \theta}_1(k) \label{Equ_Inner_e1}\\ 
	\epsilon_2\geq & Q\left(\frac{nC_2(\rho)-k R_2(\mathsf{D}_s) -\theta_2(k)}{\sqrt{n V_2(\rho) + k \tilde{V}_2(\rho)}}  \right)   
	= Q\left(\frac{nC_2(\rho)-k R_2(\mathsf{D}_s)}{\sqrt{n V_2(\rho) + k \tilde{V}_2(\rho)}}  \right)   +\tilde{ \theta}_2(k). \label{Equ_Inner_e2}
\end{align}   
By the assumption that $V_1(\rho)$ and $\tilde{V}_1(\rho)$ are bounded, the last step is by Taylor-series expansion of $Q(\cdot)$ around $\frac{nC_1(\rho)-k R_1(\mathsf{D}_x)}{\sqrt{n V_1(\rho) + k \tilde{V}_1(\rho)}}$ which yields
\begin{align}
	\left|Q\left(\frac{nC_1(\rho)-k R_1(\mathsf{D}_x)-\theta_1(k)}{\sqrt{n V_1(\rho) + k \tilde{V}_1(\rho)}}  \right)-  Q\left(\frac{nC_1(\rho)-k R_1(\mathsf{D}_x)}{\sqrt{n V_1(\rho) + k \tilde{V}_1(\rho)}}  \right)\right| 
	\leq A \log k
\end{align}
for some constants $A$. 
For simplicity, denote $T_1\triangleq nC_1(\rho)-k R_1(\mathsf{D}_x)$, $T_2\triangleq nC_2(\rho)-k R_2(\mathsf{D}_x)$, $\upsilon_1 \triangleq \sqrt{n V_1(\rho) + k \tilde{V}_1(\rho)}$, and $\upsilon_2\triangleq \sqrt{n V_2(\rho) + k \tilde{V}_2(\rho)}$.  	
By \eqref{Equ_Inner_e1} and \eqref{Equ_Inner_e2}, $\epsilon_1+\epsilon_2$ can be lower bounded by 
\begin{align}
	\epsilon_1+\epsilon_2\geq Q\left(\frac{T_1}{\upsilon_1}  \right)
	+Q\left(\frac{T_2}{\upsilon_2}  \right) +\tilde{ \theta}_1(k)+\tilde{ \theta}_2(k).
\end{align} 
Then, we consider the infimum of the lower bound of $ \epsilon_1 + \epsilon_2$, i.e.   
\begin{align}
	\inf_{\substack{(T_1, T_2): \\ T_1+ T_2= T }  }  g(T_1)
	\triangleq &
	\inf_{\substack{(T_1, T_2): \\ T_1+ T_2= T }  } \left\{  Q\left(\frac{T_1}{\upsilon_1}  \right)
	+Q\left(\frac{T_2}{\upsilon_2}  \right) \right\}, \label{Equ_Inner_Inf}
\end{align} 
where $T=n C(\rho)-k R_{S,X}(\mathsf{D}_s,\mathsf{D}_x)$, $ C(\rho)=C_1(\rho)+C_2(\rho)$ and $R_{S,X}(\mathsf{D}_s,\mathsf{D}_x)=R_1(\mathsf{D}_x)+R_2(\mathsf{D}_x)$. 
The optimal $T_1$ of \eqref{Equ_Inner_Inf} is the solution of $g'(T_1)=0$, which is given by
\begin{align}
	g'(T_1)=\frac{1}{\upsilon_1} f\left(\frac{T_1}{\upsilon_1}  \right)
	-\frac{1}{\upsilon_2} f\left(\frac{T-T_1}{\upsilon_2}  \right)=0 \label{Equ_Inner_inf_2}
\end{align}
where $f(x)=\frac{1}{ \sqrt{2\pi} } e^{-\frac{x^2}{2}}$.   
Next, we consider different cases of the solutions of \eqref{Equ_Inner_inf_2}.
\begin{itemize}
	\item[(1)] If $\frac{\upsilon_1}{\upsilon_2}< e^{-\frac{T^2}{2}} < \frac{\upsilon_2}{\upsilon_1}$, then we have
	\begin{align}
		g'(0)=&\frac{1}{\upsilon_1} 
		-\frac{1}{\upsilon_2} f\left(\frac{T}{\upsilon_2}  \right) \nonumber\\
		=&\frac{1}{\upsilon_1}
		-\frac{1}{\upsilon_2} e^{-\frac{T^2}{2}} \nonumber\\
		>& 0 \nonumber \\
		g'(T)=&\frac{1}{\upsilon_1} f\left(\frac{T}{\upsilon_1}  \right)
		-\frac{1}{\upsilon_2}  \\
		=& \frac{1}{\upsilon_1} e^{-\frac{T^2}{2}}
		-\frac{1}{\upsilon_2}   \nonumber\\
		>& 0. \nonumber 
	\end{align}
	Further, since $f(x)$ is decreasing with $x$, $g'(T_1)=0$ has no solution in $(0,T)$, and the infimum is achieved at $T_1=0$. Thus, 
	\begin{align}
		\inf_{\substack{(T_1, T_2): \\ T_1+ T_2= T }  }  g(T_1)
		= g(0)
		=Q\left(\frac{T}{\upsilon_2}  \right).
	\end{align} 
	Therefore, 
	\begin{align}
		\epsilon_1 + \epsilon_2
		\geq g(0)
		=Q\left(\frac{T}{\upsilon_2}  \right)
		+\tilde{ \theta}_1(k)+\tilde{ \theta}_2(k).
	\end{align}  		
	
	\item[(2)] Similarly, if $\frac{\upsilon_2}{\upsilon_1}< e^{-\frac{T^2}{2}} < \frac{\upsilon_1}{\upsilon_2}$, then we have
	\begin{align}
		g'(0)<& 0 \nonumber \\
		g'(T)<& 0. \nonumber 
	\end{align}
	Thus, the infimum is achieved at $T_1=T$, then 
	\begin{align}
		\epsilon_1 + \epsilon_2
		\geq g(T) 
		=Q\left(\frac{T}{\upsilon_1}  \right)
		+\tilde{ \theta}_1(k)+\tilde{ \theta}_2(k).
	\end{align}  
	
	\item[(3)] If $ e^{-\frac{T^2}{2}} <\min \Big\{ \frac{\upsilon_2}{\upsilon_1},  \frac{\upsilon_1}{\upsilon_2} \Big\}$, then we have
	\begin{align}
		g'(0)
		=& \frac{1}{\upsilon_1}
		-\frac{1}{\upsilon_2} e^{-\frac{T^2}{2}} 			
		> 0 \nonumber \\
		g'(T)=& \frac{1}{\upsilon_1} e^{-\frac{T^2}{2}}
		-\frac{1}{\upsilon_2} 
		<  0. \nonumber 
	\end{align}
	
	If $ e^{-\frac{T^2}{2}} >\max \Big\{ \frac{\upsilon_2}{\upsilon_1},  \frac{\upsilon_1}{\upsilon_2} \Big\}$, then we have
	\begin{align}
		g'(0)
		=& \frac{1}{\upsilon_1}
		-\frac{1}{\upsilon_2} e^{-\frac{T^2}{2}}
		< 0 \nonumber \\
		g'(T)=& \frac{1}{\upsilon_1} e^{-\frac{T^2}{2}}
		-\frac{1}{\upsilon_2}  
		>  0. \nonumber 
	\end{align}
	
	Then,  in these two cases, $g'(T_1)=0$ has a solution in $(0,T)$, denoted by $T_1^\star$. 
	Thus,  
	\begin{align}
		\epsilon_1 + \epsilon_2
		\geq g(T_1^\star)+\tilde{ \theta}(k),
	\end{align}   
	where $\tilde{ \theta}(k)= \log k+\log \log k+\mathcal{O}(1)$.
\end{itemize}
This completes the proof. 
 
\section{Proof of Theorem \ref{converse-GMS}} \label{proof of converse GMS}

This theorem is proved by substituting $\imath_{Y^n;Z^n}(Y^n;Z^n)$ and $\jmath_{S^k,X^k}(S^k,X^k)$ into  \eqref{L-converse bound-general}. We have 
\begin{align}
 	\imath_{Y^n;Z^n}(Y^n;Z^n)=\frac{n}{2}\log (1+P)-\frac{1}{2}\left(\frac{P}{1+P}      W^{n}_{\frac{n}{P}} -n \right),
 \end{align}
since $\sum_{j=1}^{n}(Z_j-Y_j)^2$ is a chi-square distributed random variable. Similarly,  
\begin{align}
    &\jmath_{S^k,X^k}(S^k,X^k)\nonumber\\
    =& \frac{k}{2} \log \frac{  \sigma_x^2}{ \mathsf{D}_L -(1-\alpha) \sigma_{N_s}^2}  
	+\frac{1}{2} \sum_{i=1}^{k}\left(\frac{X_i^2}{\sigma_x^2}
	+ \frac{ -(1-\alpha)\Delta_{X_i}^2 +(1-\alpha) (\Delta_{X_i}+N_{S_i})^2 }{\mathsf{D}_L -(1-\alpha) \sigma_{N_s}^2 } 
	-\frac{ \mathsf{D}_L}{\mathsf{D}_L -(1-\alpha) \sigma_{N_s}^2} \right).
\end{align}
Note that $\sum_{i=1}^{k}X_i^2/\sigma_x^2$ follows a chi-square distribution with $k$ degrees of freedom and zero noncentrality.

\section{Proof of Theorem \ref{achievability bound-GMS}}\label{Proof of achievability bound GMS}
{ 
	Theorem \ref{achievability bound-GMS} is proved by specializing Corollary \ref{achievability bound-simplified} to the transmission of GMS over AWGN channels. 
	By the proposed channel coding scheme, the information density $\imath_{Y_1,Y_2;Z}(Y_1,Y_2;Z)$ has the same distribution as 
	$\imath_{Y;Z}(Y;Z)$, as shown in Appendix \ref{proof of converse GMS}. 
	Furthermore, following \cite{Kostina-fixed-12}, $P_{\hat{X}}\!\left(B_{\mathsf{D}_x}(x)\right)$, i.e., the probability that $x$ is reconstructed within distortion level $\mathsf{D}_x$, can be tightly bounded.}
The evaluation of Type I and Type II errors follows standard techniques for finite blocklength channel coding and lossy source coding.
However, it remains a problem of evaluating the excess-distortion probability of reconstructing the semantic information (type III error), namely the third term on the RHS of \eqref{leq-achievability bound-general-simplified},
\begin{align}
    \sum_{\hat{x}\in B_{\mathsf{D}_x}\left(X\right)}\frac{P(\hat{x})}{P_{\hat{X}}\left(B_{\mathsf{D}_x}\left(X\right)\right)}P\{\mathsf{d}_s(S,h(\hat{x}))>\mathsf{D}_s\}. \label{type 3 gaussian}
\end{align}

To facilitate the proof, we first characterize the key statistical properties of the proposed coding scheme. Given the optimality of Gaussian signaling for GMS, the observable reconstruction $\hat{X}$ is generated using an i.i.d. Gaussian codebook. Specifically, the reconstruction error $\Delta_X=X-\hat{X}$ follows a Gaussian distribution that 
\begin{align}
    \Delta_X \sim \mathcal{N}(x_i, \min\{\mathsf{D}_x, \mathsf{D}_s-\sigma^2_{N_s}\}).
\end{align}
According to \cite[Thm. 2]{Liu-semantic-2022}, and since$\hat{X}$ and $\Delta_X$ are mutually independent, it follows that the variance of $\hat{X}$ as $\sigma^2_{\hat{X}}=\sigma^2_{X}-\min\{\mathsf{D}_x, \mathsf{D}_s-\sigma^2_{N_s}\}$. 
Meanwhile, we use a linear MMSE estimator to reconstruct the semantic information. That is to say, we let 
 \begin{align}
     \hat{s}=h(\hat{x})=c(\hat{x}-\mu_x)+\mu_s,
 \end{align} where $c$ is a constant scalar, and $\mu_x$ and $\mu_s$ are the expectation of $X$ and $S$ respectively. In the example, we assume $S=X+N_s$ with zero mean Gaussian random variables $S$, $X$ and $N_s$. Hence, now we use a simple decoding function for the semantic information as $\hat{s}=h(\hat{x})=\hat{x}$. Consequently, we have 
 \begin{align}
     P\{\mathsf{d}_s(S,h(\hat{x}))> \mathsf{D}_s\}=1-P\left(B_{\mathsf{D}_s}\left(\hat{x}\right)\right). 
 \end{align}
Thus, the excess-distortion probability of reconstructing semantic information for particular $\hat{x}$  is equal to the probability that $S$ falls outside the ball $B_{\mathsf{D}_s}(\hat{x})$. Hence, { for the block-wise transmission, the key problem in this proof is to evaluate  $P_{S^k}\left(B_{\mathsf{D}_s}\left(\hat{x}^k\right)\right)$. }
 

\subsubsection{Evaluating the probability of type I and II error}  
  Similar to \cite[Thm. 18]{Kostina-lossyJSCC-13}, the infimum over $\hat{X}^k$ in RHS of \eqref{leq-achievability bound-general-simplified} is weakened by choosing to be the uniform distribution on the surface of the $k$-dimensional sphere with center at 0 and radius $r'_{0}=\sqrt{k}\sigma_x \sqrt{1-\frac{\mathsf{D}_x}{\sigma^2_x}}$. 
Then, by \cite[Thm. 37]{Kostina-fixed-12} (also see in \cite{Wyner-analog-68,Sakrison-geometric-68,Kostina-lossyJSCC-13})      
\begin{align}
	P_{\hat{X}^k}\left(B_{\mathsf{D}_x}\left(X^k \right)\right)
	\geq & \rho_1 \left( W_0^{k}, \mathsf{D}_x, \sigma_x \right), 
\end{align}
and $\rho_1: \mathbb{R}^+\mapsto [0,1]$ is defined in \eqref{Equ_GSM_Inner_rho}
and $L(\cdot)$ function is defined in \eqref{Equ_GMS_Inner_Lr}.
Meanwhile, by \cite[Thm. 37]{Kostina-fixed-12}, $\imath(Y^n; Z^n)$ has the same distribution as
 \begin{align}
 		\frac{n}{2}\log (1+\rho)-\frac{\log e}{2}\left(\frac{\rho}{1+\rho} W_{\frac{n}{\rho}}^n-n \right) -\log \frac{f_{W_{n\rho}^n}(t)}{f_{W_{0}^n}\left(\frac{t}{1+\rho}\right)}  \label{Equ_GSM_Inner_i}.
 \end{align}
Hence, by taking \eqref{Equ_GSM_Inner_rho} and \eqref{Equ_GSM_Inner_i} into \eqref{leq-achievability bound-general-simplified}, the second term on the RHS of \eqref{leq-achievability bound-general-simplified} can be upper bounded by 
 \begin{align}
 	\mathbb{E} \left[ 
 	\exp\left\{ - |U-\log \gamma |^+ \right\}\right],\label{Equ_GSM_Inner_T2}
 \end{align}
 where
 \begin{align}
 	U=&\frac{n}{2}\log (1+\rho)-\frac{\log e}{2}\left(\frac{\rho}{1+\rho} W_{\frac{n}{\rho}}^n-n \right) -\log \frac{F}{\rho \left( W_0^{k}, \mathsf{D}_x, \sigma_x \right)  }, \\
 	F=& \max_{n\in \mathbb{N}, t\in \mathbb{R}^+ } \frac{f_{W_{n\rho}^n}(t)}{f_{W_{0}^n}\left(\frac{t}{1+\rho}\right)} <\infty.
 \end{align}

\subsubsection{Evaluating the probability of type III error} 

Next, we now return to \eqref{type 3 gaussian} and and express each term under the Gaussian model. We introduce a random variable $\Delta(x_i)$ as the reconstruction error for given the observable source $x_i$, namely $\Delta(x_i)=x_i-\hat{X}_i$. As an important remark, this $\Delta(x_i)$ is different from the previously defined reconstruction error $\Delta_{X_i}$, where $\Delta_{X_i}=X_i-\hat{X}_i$. Specifically, $\Delta_{X_i}$ denotes the reconstruction error over all possible source output. It is a zero mean Gaussian random variable with distribution $\Delta_{X_i}\sim \mathcal{N}(0, \sigma^2_{\Delta_x})$, where $\sigma^2_{\Delta_x}=\min\{\mathsf{D}_x, \mathsf{D}_s-\sigma^2_{N_s}\}$. $\Delta(x_i)$ is a Gaussian random variable with distribution $\Delta(x_i)\sim \mathcal{N}(x_i, \sigma^2_{\hat{X}})$, which varies with the source output $x_i$. 

Based on $\Delta(x_i)$, we further introduce a random variable $W_2$ which is defined as
\begin{align}
    W_2=\sum_{i=1}^{k} \left(\frac{\Delta(x_i)-x_i}{\sigma_{\hat{X}}} + \frac{x_i}{\sigma_{\hat{X}}} \right)^2.
\end{align}
Noting that $(\Delta(x_i)-x_i)/\sigma_{\hat{X}}$ is a normal distribution, and 
\begin{align}
\sum_{i=1}^{k}  \left(\frac{x_i}{\sigma_{\hat{X}}}\right)^2=\frac{\sum_{i=1}^{k} x_i^2}{\sigma^2_{\hat{X}}}=\frac{\sigma^2_{X}w_0^k}{\sigma^2_{\hat{X}}} .   
\end{align}
Hence, $W_2$ is a chi-square random variable with $k$ degrees of freedom and noncentrality parameter as $\sigma^2_{X}w_0^k/\sigma^2_{\hat{X}}$. 
 Meanwhile, since 
 \begin{align}
    W_2= \sum_{i=1}^{k} \left(\frac{\Delta(x_i)-x_i}{\sigma_{\hat{X}}} + \frac{x_i}{\sigma_{\hat{X}}} \right)^2
    =\sum_{i=1}^{k} \left(\frac{\Delta(x_i)}{\sigma_{\hat{X}}} \right)^2=\frac{\sum_{i=1}^{k}(\Delta(x_i))^2}{\sigma^2_{\hat{X}}}, 
 \end{align}
hence the event
\begin{align}
    W_2\leq \frac{\mathsf{D}_x}{\sigma^2_{\hat{X}}} \label{equivalent ball}
\end{align}
 is an equivalent to the event $\sum_{i=1}^{k} (x_i-\hat{X}_i)^2 \leq \mathsf{D}_x$ for some given $x^k=(x_1,\dots,x_k)$. 
Therefore, to sum over all $\hat{x}^k=(\hat{x}_1,\dots,\hat{x}_k)\in B_{\mathsf{D}_x}\left(x^k\right)$ in  \eqref{type 3 gaussian} for some given $x^k$ can be replaced by the integral over all $w_2$ satisfying inequality 
\eqref{equivalent ball}. 

Then, combining the result of $P_{S^k}\left(B_{\mathsf{D}_s}\left(\hat{x}^k\right)\right)$ given in Lemma \ref{Proof_Lem_Bd}, we obtain 
{\small  \begin{align}
	\sum_{\hat{x}^k\in B_{\mathsf{D}_x}\left(X^k\right)} & \frac{P(\hat{x}^k)}{P_{\hat{X}^k}\left(B_{\mathsf{D}_x}\left(X^k\right)\right)}P\{\mathsf{d}_s(S^k,h(\hat{x}^k))>\mathsf{D}_s\} \\
	\leq &\mathbb{E}_{\sim W_{\frac{n}{P}}^n, W_{0}^n} \left[ \exp\left\{ - |U-\log \gamma |^+ \right\} \right]     \nonumber \\
	    & \left.+\mathbb{E}_{\sim W_0}  
		  \left[ \frac{1}{\rho_1 \left( W_0^{k}, \mathsf{D}_x, \sigma_x \right)} \int_{w_2 \leq {\mathsf{D}_x}/{\sigma^2_{\hat{X}}}} f_{W_{\sigma^2_{X}W_0^k/\sigma^2_{\hat{X}}}^{k}}(w_2)  
  \mathbb{E}_{\sim W_1} 	\left[1-  \rho_2(\sigma^2_{N_s}W_{1},\sigma^2_{\hat{X}}w_2,\mathsf{D}_s)  \right] \text{d}w_2 \right]  \right\},   \label{Equ_GSM_Inner_T3}
\end{align}}
and hence \eqref{Equ_GSM_Inner} is obtained by taking \eqref{Equ_GSM_Inner_T2} and \eqref{Equ_GSM_Inner_T3} into \eqref{leq-achievability bound-general-simplified}.

\section{Proof of Theorem \ref{Thm_GApprx_GMS}} \label{source dispersion linear proof}
	{ 	For the transmission  of GMS over Gaussian channels, $\hat{S}$ is a deterministic function of $\hat{X}$, i.e. $\hat{S}=h(\hat{X})$. Based on the achievability coding scheme, only the first-layer decoding is applied and the second layer decoding is omitted. 
	Thus, the terms in Theorem \ref{Thm_Gau_Ach1} satisfy:
	\begin{itemize}
		\item[(1)]  $V_1(\rho)$ and $V_2(\rho)$ are the channel dispersions given by
		\begin{align}
			V_1(\rho)=& {\rm{Var}}\big[\imath_{Y_1;Z}(Y_1;Z) \big]={\rm{Var}}\big[\imath_{Y;Z}(Y;Z) \big] =  V(\rho),\\
			V_2(\rho)=& 0,
		\end{align}
		and $ C(\rho)$ is the channel capacity;
		\item[(2)] $\tilde{V}(\rho)$ is the source dispersion given by
		\begin{align}
			\tilde{V}_1(\rho)=& {\rm{Var}}\big[\jmath_X(X; \hat{X})\big]=\tilde{V}(\rho), \\
			\tilde{V}_2(\rho)=& {\rm{Var}}\big[\jmath_X(X; \hat{S})|\hat{X}=\hat{x}\big]=0,
		\end{align}
		where the first equation follows from the calculation of $\tilde{V}(\rho)$ in Appendix F. 
		\item[(3)] 	$\upsilon_1 \triangleq \sqrt{n  V(\rho) + k \tilde{V}(\rho)}$ and $\upsilon_2=0$;
	\end{itemize} 
	Thus, the error probability satisfies (the second case in Theorem  \ref{Thm_Gau_Ach1})
	\begin{align}
		\epsilon 
		\leq Q\left(\frac{n C(\rho)-k R_{S,X}(\mathsf{D}_s,\mathsf{D}_x)+\theta(k)}{\sqrt{n  V(\rho) + k \tilde{V}(\rho)}}  \right).
	\end{align}  
	By Corollary \ref{Converse_GA_err}, the converse bound is given by
		\begin{align}
			\epsilon  \geq & Q\left( \frac{n C(\rho)-kR_{S,X}(\mathsf{D}_s,\mathsf{D}_x)-\theta(n)}{\sqrt{n V(\rho)+k\tilde{V}(\rho)} } 
			\right).\label{Equ_Corr1A}
		\end{align}
   
   In the following, we compute $ C(\rho)$, $ V(\rho)$, $\tilde{V}(\rho)$, and $R_{S,X}(\mathsf{D}_s,\mathsf{D}_x)$ for the transmission of a GMS over an AWGN channel. }
	According to the results in \cite{Kostina-lossyJSCC-13}, for jointly Gaussian sources, the minimum coding rate is obtained by jointly Gaussian distributed random variables $\hat{X}$ and $\hat{S}$, and there exists a Markov chain $S_i-X_i-\hat{X}_i-\hat{S}_i$. Therefore, the information density can be written as
	\begin{align}
		\imath_{X; \hat{X}, \hat{S}} (x_i; \hat{x}_i, \hat{s}_i)
		=&\log \left(\frac{P_{X|\hat{X}}(x_i| \hat{x}_i )}{P_X(x_i)}\right)  \nonumber\\
		=&\frac{1}{2} \log \frac{\sigma_x^2}{\mathbb{E} \left[ |x_i-\hat{x}_i |^2\right] } + \frac{1}{2} \left(\frac{|x_i|^2}{\sigma_x^2} -\frac{ |x_i-\hat{x}_i|^2}{\mathbb{E} \left[ |x_i-\hat{x}_i |^2\right] }  \right) \nonumber\\
		=& \frac{1}{2}  \log \frac{\sigma_x^2}{\mathbb{E} \left[ \mathsf{d}_x(x_i, \hat{x}_i) \right] } + \frac{1}{2}\left(\frac{|x_i|^2}{\sigma_x^2} -\frac{ \mathsf{d}_x(x_i, \hat{x}_i)}{\mathbb{E} \left[\mathsf{d}_x(x_i, \hat{x}_i) \right] }  \right)  \label{Equ_i},
	\end{align}

	Moreover, for the $\mathsf{d}()$-functional distortion constraint, we have   
	\begin{align}
		&\mathsf{d}(\mathbb{E}[\mathsf{d}_x({x}, {\hat{x}})], \mathbb{E}[\mathsf{d}_s({s}, {\hat{s}})])\nonumber\\
		=&\alpha \mathbb{E}[\mathsf{d}_x({x}, {\hat{x}})]
		+(1-\alpha) \mathbb{E}[\mathsf{d}_s({s}, {\hat{s}})]\nonumber\\
		\overset{(a)}{=}&\alpha \mathbb{E}[\mathsf{d}_x({x}, {\hat{x}})]
		+(1-\alpha)\left( \mathbb{E}[\mathsf{d}_x({x}, {\hat{x}})]+\sigma_{N_s}^2\right)\nonumber\\
		\leq & \alpha \mathsf{D}_x  + (1-\alpha)  \mathsf{D}_s \triangleq  \mathsf{D}_L, 
	\end{align}
	where step $(a)$ follows the fact proved in \cite{Kostina-lossyJSCC-13} that $\mathbb{E}[\mathsf{d}_s({s}, {\hat{s}})]=\mathbb{E}[\mathsf{d}_x({x}, {\hat{x}})]+\sigma_{N_s}^2$, and for convenience we denote the linear function results $\alpha \mathsf{D}_x  + (1-\alpha)  \mathsf{D}_s$ as $\mathsf{D}_L$.

	Hence, equivalently we have  
	\begin{align}
		\mathbb{E}[\mathsf{d}_x({x}, {\hat{x}})]
		\leq  \mathsf{D}_L -(1-\alpha) \sigma_{N_s}^2,  \label{Equ_Cons_Edxf}
	\end{align}
	and the equality holds when the minimum coding rate is achieved. 
	
	By \eqref{Equ_Cons_Edxf} we also can obtain that  
	\begin{align}
		R_{S,X}(\mathsf{D}_s,\mathsf{D}_x) 
		=& \inf_{P_{\hat{S}\hat{X}|X} }  \frac{1}{2}\log \frac{\sigma_x^2}{\mathbb{E} \left[ \mathsf{d}_x(x_i, \hat{x}_i) \right] } \\
		=& \frac{1}{2} \log \frac{\sigma_x^2}{ \mathsf{D}_L -(1-\alpha) \sigma_{N_s}^2}.
	\end{align}
	Hence, we have 
	\begin{align}
		\lambda^{\star} =& -\pdv{\mathbb{R}(\mathsf{d}_s,\mathsf{d}_x)}{\mathsf{D}_L} 
		=\frac{1}{2 \left( \mathsf{D}_L -(1-\alpha) \sigma_{N_s}^2 \right)}, \\
		\imath_{X; \hat{X}, \hat{S}} (x_i; \hat{x}_i, \hat{s}_i)
		=& \frac{1}{2} \log \frac{ \sigma_x^2}{ \mathsf{D}_L -(1-\alpha) \sigma_{N_s}^2 } +\frac{1}{2} \left(\frac{|x_i|^2}{\sigma_x^2} 	-\frac{\mathsf{d}_x({x}, {\hat{x}})}{ \mathsf{D}_L -(1-\alpha) \sigma_{N_s}^2 }  \right). 
	\end{align}
	Correspondingly, the $(\mathsf{d}_s,\mathsf{d}_x)$-tilted information $\jmath_{S^k,X^k}(s^k,x^k)$ is given by
	\begin{align}
		&\jmath_{S^k,X^k}(s^k,x^k) \nonumber\\
		=& \sum_{i=1}^{k} \jmath_{S,X}(s_i, x_i) \nonumber\\
		=&\sum_{i=1}^{k} \Big(\imath_{X; \hat{X}, \hat{S}} (x_i; \hat{x}_i, \hat{s}_i)
		+\lambda^{\star} \big(\mathsf{d}(\mathsf{d}_s(s_i,\hat{s}_i),\mathsf{d}_x(x_i,\hat{x}_i))-\mathsf{D}_L\big) \Big)\nonumber \\
		=&\sum_{i=1}^{k}\left[ \frac{1}{2} \log \frac{  \sigma_x^2}{ \mathsf{D}_L -(1-\alpha) \sigma_{N_s}^2} + \frac{1}{2}\left(\frac{|x_i|^2}{\sigma_x^2} -\frac{ \mathsf{d}_x({x}_i, {\hat{x}_i})}{ \mathsf{D}_L -(1-\alpha) \sigma_{N_s}^2}  \right) \right.\nonumber \\
		&\left. +\frac{1}{2( \mathsf{D}_L -(1-\alpha) \sigma_{N_s}^2)}  \big((1-\alpha) \mathsf{d}_s({s}_i, {\hat{s}_i})+\alpha \mathsf{d}_x({x}_i, {\hat{x}_i})-\mathsf{D}_L \big) \right] \nonumber\\ 
		=&\sum_{i=1}^{k}\left[ \frac{1}{2} \log \frac{  \sigma_x^2}{ \mathsf{D}_L -(1-\alpha) \sigma_{N_s}^2} 
		+\frac{1}{2} \left(\frac{|x_i|^2}{\sigma_x^2}
		+ \frac{ -(1-\alpha) \mathsf{d}_x({x}_i, {\hat{x}_i}) +(1-\alpha) \mathsf{d}_s({s}_i, {\hat{s}_i}) }{\mathsf{D}_L -(1-\alpha) \sigma_{N_s}^2 } 
		-\frac{ \mathsf{D}_L}{\mathsf{D}_L -(1-\alpha) \sigma_{N_s}^2} \right)  \right] \label{tilted information-gaussian-1}.
	\end{align}
	For convenience, we denote $\Delta_{X_i}=X_i-\hat{X}_i$ with realization $\delta_{X_i}$ as the difference between $\hat{X}_i$ and $X_i$. Hence, according to the setup in Section \ref{Sec_GMS_AWGN}, we can write $S_i-\hat{S}_i=X_i+N_{S_i}-\hat{X}_i=\Delta_{X_i}+N_{S_i}$ since $N_{S_i}$ is a zero mean Gaussian distributed random variable. Thus, the error of estimating $S_i$ equals the error of estimating $X_i$ plus the error of estimating $S_i$ with given $X_i$. Therefore,  \eqref{tilted information-gaussian-1} can be further simplified as follows that
	\begin{align}
		&\jmath_{S,X}(s_i, x_i) \nonumber\\
		=& \frac{1}{2} \log \frac{  \sigma_x^2}{ \mathsf{D}_L -(1-\alpha) \sigma_{N_s}^2} \nonumber\\
		&+\frac{1}{2} \left(\frac{X_i^2}{\sigma_x^2}
		+ \frac{ -(1-\alpha)\Delta_{X_i}^2 +(1-\alpha) (\Delta_{X_i}+N_{S_i})^2) }{\mathsf{D}_L -(1-\alpha) \sigma_{N_s}^2 } 
		-\frac{ \mathsf{D}_L}{\mathsf{D}_L -(1-\alpha) \sigma_{N_s}^2} \right).  \label{tilted information Gaussian}
	\end{align}
	
	Notice that only $X_i$,                                      
	$\Delta_{X_i}$ and $N_{S_i}$ vary with the value of  $(S_i,X_i,\hat{X}_i,\hat{S}_i)$, and other parameters can be regarded as constant numbers once the source distribution and channel distribution are given. Hence, the expected value of $\jmath_{S,X}(S_i, X_i)$ is given as
	\begin{align}
		\mu  =&  \mathbb{E} \big[\jmath_{S,X}(S_i,X_i) \big] \nonumber\\
		=&  \mathbb{E} \big[\jmath_{S,X}(S_1,X_1) \big] \nonumber\\
		= & \frac{1}{2} \log \frac{  \sigma_x^2}{ \mathsf{D}_L -(1-\alpha) \sigma_{N_s}^2} \nonumber\\
		&+\frac{1}{2}\mathbb{E} \left[\frac{|X_1|^2}{\sigma_x^2}
		+ \frac{ -(1-\alpha)\Delta_{X_1}^2 +(1-\alpha) (\Delta_{X_1}+N_{s_1})^2) }{\mathsf{D}_L -(1-\alpha) \sigma_{N_s}^2 } 
		-\frac{ \mathsf{D}_L}{\mathsf{D}_L -(1-\alpha) \sigma_{N_s}^2} \right]  \nonumber\\
		= & \frac{1}{2} \log \frac{  \sigma_x^2}{ \mathsf{D}_L -(1-\alpha) \sigma_{N_s}^2} \nonumber\\
		&+\frac{1}{2}\mathbb{E} \left[\frac{|X_1|^2}{\sigma_x^2}
		+ \frac{ 2(1-\alpha)\Delta_{X_1}N_{S_1}+ (1-\alpha)N^2_{S_1}}{\mathsf{D}_L -(1-\alpha) \sigma_{N_s}^2 } 
		-\frac{ \mathsf{D}_L}{\mathsf{D}_L -(1-\alpha) \sigma_{N_s}^2} \right]  \nonumber\\
		= & \frac{1}{2} \log \frac{  \sigma_x^2}{ \mathsf{D}_L -(1-\alpha) \sigma_{N_s}^2} +\frac{1}{2}\left(1 
		+ \frac{ (1-\alpha) \sigma_{N_s}^2}{\mathsf{D}_L -(1-\alpha) \sigma_{N_s}^2 } 
		-\frac{ \mathsf{D}_L}{\mathsf{D}_L -(1-\alpha) \sigma_{N_s}^2} \right).
		\label{Equ_NA_Cov_1}
	\end{align}
	Meanwhile, the second order moment of $\jmath_{S,X}(S_i,X_i)$ is given as
	\begin{align}
		&\text{Var}\left[\jmath_{S,X}(S_i,X_i) \right]\nonumber\\
		=&\frac{1}{4}\text{Var} \left[\frac{|X_1|^2}{\sigma_x^2}
		+ \frac{ 2(1-\alpha)\Delta_{X_1}N_{S_1}+ (1-\alpha)N^2_{S_1}}{\mathsf{D}_L -(1-\alpha) \sigma_{N_s}^2 } \right] \nonumber\\
		=&\frac{1}{4}\text{Var} \left[\frac{|X_1|^2}{\sigma_x^2}\right]
		+ \frac{1}{4}\text{Var} \left[\frac{ 2(1-\alpha)\Delta_{X_1}N_{S_1}+ (1-\alpha)N^2_{S_1}}{\mathsf{D}_L -(1-\alpha) \sigma_{N_s}^2 } \right] \nonumber\\
		=& \frac{1}{2}+\frac{4(1-\alpha)^2(\mathsf{D}_L -(1-\alpha) \sigma_{N_s}^2)\sigma_{N_s}^2+2(1-\alpha)^2\sigma_{N_s}^4}{4(\mathsf{D}_L -(1-\alpha) \sigma_{N_s}^2)^2 }.
	\end{align}

\section{Proof of Lemma \ref{Proof_Lem_Bd}}\label{Proof_Lem_Bd}
Fig. \ref{Fig_Bd_Inner} shows the geometry of the excess-distortion probability calculation for given $x$. 
Specifically, we have $S_i=x_i+N_s$ and $\hat{x}_i=x_i-\Delta(x_i)$, where both $S_i$ and $\hat{x}_i$ depend on the same $x_i$.
Hence, we consider an equivalent probability by removing the dependence on $x_i$, namely, the probability $P_{N_S}\left(B_{\mathsf{D}_s}\left(\delta(x)\right)\right)$. 

We first consider the case shown in Fig. \ref{Fig_Bd_Inner}-(a). Since $N_s$ and $\Delta(x)$ are independent random variables, we can apply  \cite[Thm. 37]{Kostina-fixed-12} (see also \cite{Wyner-analog-68,Sakrison-geometric-68,Kostina-lossyJSCC-13}). Specifically, when $\sqrt{k\mathsf{D}_s}-r_0\leq |n_s|\leq \sqrt{k\mathsf{D}_s}+r_0$, we have  
\begin{align}
	P_{S^k}\left(B_{\mathsf{D}_s}\left(\hat{x}^k\right)\right)
	=& P_{N_S}\left(B_{\mathsf{D}_s}\left(\delta(x)\right)\right) \nonumber\\
	\geq & \frac{\Gamma\left(\frac{k}{2}+1\right)}{\sqrt{\pi}k\Gamma\left(\frac{k-1}{2}+1\right)} (1-\cos^2 \theta)^{\frac{k-1}{2}}, 
\end{align}
where
\begin{align}
	\cos \theta =& \frac{r^2+r_0^2-k\mathsf{D}_s}{2rr_0},
\end{align}
$r=\sqrt{\sum_{i=1}^{k} n_{s_i}^2}$ and $r_{0}=|x^k-\hat{x}^k|=\sigma_{\hat{X}}\sqrt{w_2}$.	
Define $\sum_{i=1}^{k} n_{s_i}^2 \triangleq w_1 \sigma^2_{N_s}$, where $w_1$ is the realization of a chi-square random variable with $k$ degrees of freedom and zero noncentrality. 
Then $\cos \theta$ can be further written as
\begin{align}
	\cos \theta 
	= & \frac{w_1 \sigma^2_{N_s}+w_2 \sigma^2_{\hat{X}}-\mathsf{D}_s}{2 \sqrt{w_1 \sigma^2_{N_s}} \sqrt{w_2 \sigma^2_{\hat{X}}}}.
\end{align}
Therefore, $P\{\mathsf{d}_s(S^k,h(\hat{x}^k))>\mathsf{D}_s\}$ can be bounded as follows
\begin{align}
	P\{\mathsf{d}_s(S^k,h(\hat{x}^k))>\mathsf{D}_s\}
	\leq  &1- P\{\mathsf{d}_s(S^k,h(\hat{x}^k))\leq \mathsf{D}_s\} \\
	\leq &
	1-  \frac{\Gamma\left(\frac{k}{2}+1\right)}{\sqrt{\pi}k\Gamma\left(\frac{k-1}{2}+1\right)} 
	\left(1-  \left(  \frac{ \sigma^2_{N_s}w_1 +\sigma^2_{\hat{X}}w_2- \mathsf{D}_s}{2 \sqrt{ \sigma^2_{N_s}w_1  } \sqrt{  \sigma^2_{\Delta_X}w_2}} \right)^2  \right)^{\frac{k-1}{2}}  .
	\label{Equ_GSM_Inner_rhoB}
\end{align}

\end{appendices}

\bibliographystyle{IEEEtran}
\bibliography{sample}

\end{document}

%% file: ChannelModel.tex
	\begin{tikzpicture}[
		outer/.style={draw=gray,dashed,fill=black!1,thick,inner sep=3pt},
		box/.style={draw,thick,minimum width=10mm,minimum height=8mm,text width =20mm,align=flush center},
		box1/.style={draw,fill=gray!30,thick,minimum width=10mm,minimum height=8mm,text width =20mm,align=flush center},
		arrow/.style={draw, single arrow,thick,minimum height=8mm, minimum width=8mm, single arrow head extend=2mm,anchor=west}]
		
		\node (s1) at (0,0) [] {\Large$S$};
		\node (s2) at (2.2,0) [box] {\Large$P_{X|S}$};
		\node (s3) at (5.5,0) [box1] {\Large$f_{{Y}|{X}}$};
		\node (s4) at (8.7,0) [box] {\Large$P_{Z|Y}$};
		\node (s) at (11.1,0) [] {\Large$Z$};
		\node (s5) at (13.1,1) [box1] {\Large$g_{\hat{S}|Z}{(Z)}$};
		\node (s6) at (13,-1) [box1] {\Large$g_{\hat{X}|Z}{(Z)}$};
		\draw [thick,shorten >=1mm,shorten <=1mm,->](s1.east)->(s2.west)node [midway,sloped,above,align=center] {};
		\draw [thick,shorten >=1mm,shorten <=1mm,->](s2.east)->(s3.west)node [midway,sloped,above,align=center] {\Large$X$};
		\draw [thick,shorten >=1mm,shorten <=1mm,->](s3.east)->(s4.west)node [midway,sloped,above,align=center] {\Large$Y$};
		\draw [thick,shorten >=1mm,shorten <=1mm,->](s4.east)->(s.west)node [midway,sloped,above,align=center] {};
		\draw [thick,shorten >=1mm,shorten <=1mm,->](s)|-(s5.west)node [midway,sloped,above,align=center] {};
		\draw [thick,shorten >=1mm,shorten <=1mm,->](s)|-(s6.west)node [midway,sloped,above,align=center] {};
		\node (s7) at (16,1) [] {\Large$\hat{S}$};
		\node (s8) at (15.4,-1) [] {\Large$\hat{X}$};
		\draw [thick,shorten >=1mm,shorten <=1mm,->](s5.east)->(s7.west)node [midway,sloped,above,align=center] {};
		\draw [thick,shorten >=1mm,shorten <=1mm,->](s6.east)->(s8.west)node [midway,sloped,above,align=center] {};
		
		\node (dx) at (7,-2) [] {\Large$\mathsf{d}(x,\hat{x})\leq \mathsf{D}_x$};
		\draw [dashed, -latex] (3.7,-0.2) |- (dx.west);
		\draw [dashed, -latex] (s8) |- (dx.east);
		\node (ds) at (7,-2.7) [] {\Large$\mathsf{d}(s,\hat{s})\leq \mathsf{D}_s$};
		\draw [dashed, -latex] (s1) |- (ds.west);
		\draw [dashed, -latex] (s7) |- (ds.east);
	\end{tikzpicture}
	

%% file: ChannelModel2.tex
	\begin{tikzpicture}[
		outer/.style={draw=gray,dashed,fill=black!1,thick,inner sep=3pt},
		box/.style={draw,thick,minimum width=10mm,minimum height=8mm,text width =20mm,align=flush center},
		box1/.style={draw,fill=gray!30,thick,minimum width=10mm,minimum height=8mm,text width =20mm,align=flush center},
		arrow/.style={draw, single arrow,thick,minimum height=8mm, minimum width=8mm, single arrow head extend=2mm,anchor=west}]
		
		\node (s1) at (0,0) [] {\Large$S^k$};
		\node (s2) at (2.2,0) [box] {\Large$P_{X^k|S^k}$};
		\node (s3) at (5.5,0) [box1] {\Large$f_{Y^n|X^k}$};
		\node (s4) at (8.7,0) [box] {\Large$P_{Z^n|Y^n}$};
		\node (s) at (11.1,0) [] {\Large$Z^n$};
		\node (s5) at (13.1,1) [box1] {\Large$g_{\hat{S}|Z}^k{(Z^n)}$};
		\node (s6) at (13,-1) [box1] {\Large$g_{\hat{X}|Z}^k{(Z^n)}$};
		\draw [thick,shorten >=1mm,shorten <=1mm,->](s1.east)->(s2.west)node [midway,sloped,above,align=center] {};
		\draw [thick,shorten >=1mm,shorten <=1mm,->](s2.east)->(s3.west)node [midway,sloped,above,align=center] {\Large$X^k$};
		\draw [thick,shorten >=1mm,shorten <=1mm,->](s3.east)->(s4.west)node [midway,sloped,above,align=center] {\Large$Y^n$};
		\draw [thick,shorten >=1mm,shorten <=1mm,->](s4.east)->(s.west)node [midway,sloped,above,align=center] {};
		\draw [thick,shorten >=1mm,shorten <=1mm,->](s)|-(s5.west)node [midway,sloped,above,align=center] {};
		\draw [thick,shorten >=1mm,shorten <=1mm,->](s)|-(s6.west)node [midway,sloped,above,align=center] {};
		\node (s7) at (16,1) [] {\Large$\hat{S}^k$};
		\node (s8) at (15.4,-1) [] {\Large$\hat{X}^k$};
		\draw [thick,shorten >=1mm,shorten <=1mm,->](s5.east)->(s7.west)node [midway,sloped,above,align=center] {};
		\draw [thick,shorten >=1mm,shorten <=1mm,->](s6.east)->(s8.west)node [midway,sloped,above,align=center] {};
		
		\node (dx) at (7,-2) [] {\Large$\mathsf{d}(x^k,\hat{x}^k)\leq \mathsf{D}_x$};
		\draw [dashed, -latex] (3.7,-0.2) |- (dx.west);
		\draw [dashed, -latex] (s8) |- (dx.east);
		\node (ds) at (7,-2.7) [] {\Large$\mathsf{d}(s^k,\hat{s}^k)\leq \mathsf{D}_s$};
		\draw [dashed, -latex] (s1) |- (ds.west);
		\draw [dashed, -latex] (s7) |- (ds.east);
	\end{tikzpicture}
	

%% file: sample.bib
@ARTICLE{Csisza1974,
	title = {On an extremum problem of information theory},
	author  = {Csisz\'{a}r},
	journal = {  Studia Scientiarum Mathematicarum Hungarica},
	volume  = {9},
	number  =  {1},
	pages   = { 57-71},
	year    = {1974}
}

@ARTICLE{Yang25,
	author={Yang, Huiyuan and Shi, Yuxuan and Shao, Shuo and Yuan, Xiaojun},
	journal={IEEE Transactions on Communications}, 
	title={Indirect Lossy Source Coding With Observed Source Reconstruction: Nonasymptotic Bounds and Second-Order Asymptotics}, 
	year={2025},
	volume={73},
	number={6},
	pages={4241-4256},
	doi={10.1109/TCOMM.2024.3507636}}

@ARTICLE{Liu-semantic-2022,
  author={Liu, Jiakun and Shao, Shuo and Zhang, Wenyi and Poor, H. Vincent},
  journal={IEEE Transactions on Communications}, 
  title={An Indirect Rate-Distortion Characterization for Semantic Sources: General Model and the Case of Gaussian Observation}, 
  year={2022},
  volume={70},
  number={9},
  pages={5946-5959},
  doi={10.1109/TCOMM.2022.3194978}}

@article{Liu-representinglearning-2021,
  author       = {Kang Liu and
                  Dong Liu and
                  Li Li and
                  Ning Yan and
                  Houqiang Li},
  title        = {Semantics-to-Signal Scalable Image Compression with Learned Revertible
                  Representations},
  journal      = {Int. J. Comput. Vis.},
  volume       = {129},
  number       = {9},
  pages        = {2605--2621},
  year         = {2021},
  url          = {https://doi.org/10.1007/s11263-021-01491-7},
  doi          = {10.1007/S11263-021-01491-7},
  bibsource    = {dblp computer science bibliography, https://dblp.org}
}

@ARTICLE{Zhang-imagesemantic-2023,
  author={Zhang, Hongwei and Shao, Shuo and Tao, Meixia and Bi, Xiaoyan and Letaief, Khaled B.},
  journal={IEEE Journal on Selected Areas in Communications}, 
  title={Deep Learning-Enabled Semantic Communication Systems With Task-Unaware Transmitter and Dynamic Data}, 
  year={2023},
  volume={41},
  number={1},
  pages={170-185},
  doi={10.1109/JSAC.2022.3221991}}

@ARTICLE{Stavrou-ba-2023,
  author={Stavrou, Photios A. and Kountouris, Marios},
  journal={IEEE Transactions on Communications}, 
  title={The Role of Fidelity in Goal-Oriented Semantic Communication: A Rate Distortion Approach}, 
  year={2023},
  volume={71},
  number={7},
  pages={3918-3931},
  doi={10.1109/TCOMM.2023.3274122}}

@ARTICLE{Shi-errorexponent-2023,
  author={Shi, Yuxuan and Shao, Shuo and Wu, Yongpeng and Zhang, Wenjun and Xia, Xiang-Gen and Xiao, Chengshan},
  journal={IEEE Transactions on Wireless Communications}, 
  title={Excess Distortion Exponent Analysis for Semantic-Aware MIMO Communication Systems}, 
  year={2023},
  volume={22},
  number={9},
  pages={5927-5940},
  doi={10.1109/TWC.2023.3238463}}

@INPROCEEDINGS{Wang-semanticseperation-2022,
  author={Wang, Yizhu and Guo, Tao and Bai, Bo and Han, Wei},
  booktitle={2022 IEEE Information Theory Workshop (ITW)}, 
  title={The Estimation-Compression Separation in Semantic Communication Systems}, 
  year={2022},
  volume={},
  number={},
  pages={315-320},
  doi={10.1109/ITW54588.2022.9965794}}

@ARTICLE{Kostina-lossyJSCC-13,
	author={Kostina, Victoria and Verdú, Sergio},
	journal={IEEE Transactions on Information Theory}, 
	title={Lossy Joint Source-Channel Coding in the Finite Blocklength Regime}, 
	year={2013},
	volume={59},
	number={5},
	pages={2545-2575},
	doi={10.1109/TIT.2013.2238657}}

@ARTICLE{Kostina-fixed-12,
		author={Kostina, Victoria and Verdu, Sergio},
		journal={IEEE Transactions on Information Theory}, 
		title={Fixed-Length Lossy Compression in the Finite Blocklength Regime}, 
		year={2012},
		volume={58},
		number={6},
		pages={3309-3338},
		doi={10.1109/TIT.2012.2186786}}

@ARTICLE{Wyner-analog-68,
			author={Wyner, A. D.},
			journal={The Bell System Technical Journal}, 
			title={Communication of analog data from a Gaussian source over a noisy channel}, 
			year={1968},
			volume={47},
			number={5},
			pages={801-812},
			doi={10.1002/j.1538-7305.1968.tb00062.x}}

@ARTICLE{Sakrison-geometric-68,
			author={Sakrison, D.},
			journal={IEEE Transactions on Information Theory}, 
			title={A geometric treatment of the source encoding of a Gaussian random variable}, 
			year={1968},
			volume={14},
			number={3},
			pages={481-486},
			doi={10.1109/TIT.1968.1054145}}

@ARTICLE{Kostina-isc-2016,
  author={Kostina, Victoria and Verdú, Sergio},
  journal={IEEE Transactions on Information Theory}, 
  title={Nonasymptotic Noisy Lossy Source Coding}, 
  year={2016},
  volume={62},
  number={11},
  pages={6111-6123},
  doi={10.1109/TIT.2016.2562008}}

@INPROCEEDINGS{Ulger-23,
  author={{\"{U}}lger, Oğuzhan Kubilay and Erkip, Elza},
  booktitle={2023 59th Annual Allerton Conference on Communication, Control, and Computing (Allerton)}, 
  title={Single-Shot Lossy Compression for Joint Inference and Reconstruction}, 
  year={2023},
  volume={},
  number={},
  pages={1-8},
  doi={10.1109/Allerton58177.2023.10313437}}

@ARTICLE{Yury-10,
	author={Polyanskiy, Yury and Poor, H. Vincent and Verdu, Sergio},
	journal={IEEE Transactions on Information Theory}, 
	title={Channel Coding Rate in the Finite Blocklength Regime}, 
	year={2010},
	volume={56},
	number={5},
	pages={2307-2359},
	doi={10.1109/TIT.2010.2043769}}

@ARTICLE{Gunduz-23,
  author={Gündüz, Deniz and Qin, Zhijin and Aguerri, Inaki Estella and Dhillon, Harpreet S. and Yang, Zhaohui and Yener, Aylin and Wong, Kai Kit and Chae, Chan-Byoung},
  journal={IEEE Journal on Selected Areas in Communications}, 
  title={Beyond Transmitting Bits: Context, Semantics, and Task-Oriented Communications}, 
  year={2023},
  volume={41},
  number={1},
  pages={5-41},
  doi={10.1109/JSAC.2022.3223408}}

@ARTICLE{Zhou-19,
  author={Zhou, Lin and Motani, Mehul},
  journal={IEEE Transactions on Information Theory}, 
  title={Non-Asymptotic Converse Bounds and Refined Asymptotics for Two Source Coding Problems}, 
  year={2019},
  volume={65},
  number={10},
  pages={6414-6440},
  doi={10.1109/TIT.2019.2920893}}

@ARTICLE{Polyanskiy-10,
  author={Polyanskiy, Yury and Poor, H. Vincent and Verdu, Sergio},
  journal={IEEE Transactions on Information Theory}, 
  title={Channel Coding Rate in the Finite Blocklength Regime}, 
  year={2010},
  volume={56},
  number={5},
  pages={2307-2359},
  doi={10.1109/TIT.2010.2043769}}

@ARTICLE{Tan15,
  	author={Tan, Vincent Yan Fu and Tomamichel, Marco},
  	journal={IEEE Transactions on Information Theory}, 
  	title={The Third-Order Term in the Normal Approximation for the AWGN Channel}, 
  	year={2015},
  	volume={61},
  	number={5},
  	pages={2430-2438},
  	doi={10.1109/TIT.2015.2411256}}

@article{Collins2019coherent,
		title={Coherent Multiple-Antenna Block-Fading Channels at Finite Blocklength},
		author={Collins, Austin and  Polyanskiy, Yury},
		journal={IEEE Trans. Inf. Theory},
		volume={65},
		number={1},
		pages={380--405},
		year={2019},
		publisher={IEEE}
	}

@article{yang2014quasi,
		title={Quasi-static multiple-antenna fading channels at finite blocklength},
		author={Yang, Wei and Durisi, Giuseppe and Koch, Tobias and Polyanskiy, Yury},
		journal={IEEE Trans. Inf. Theory},
		volume={60},
		number={7},
		pages={4232--4265},
		year={2014},
		publisher={IEEE}
	}

@article{lancho2019single,
		title={On single-antenna {R}ayleigh block-fading channels at finite blocklength},
		author={Lancho, Alejandro and Koch, Tobias and Durisi, Giuseppe},
		journal={IEEE Trans. Inf. Theory},
		year={2019},
		publisher={IEEE}
	}

@ARTICLE{Tan21,
		author={Sakai, Yuta and Yavas, Recep Can and Tan, Vincent Y. F.},
		journal={IEEE Transactions on Information Theory}, 
		title={Third-Order Asymptotics of Variable-Length Compression Allowing Errors}, 
		year={2021},
		volume={67},
		number={12},
		pages={7708-7722},
		doi={10.1109/TIT.2021.3117591}}

@ARTICLE{Albert14,
			author={Scarlett, Jonathan and Martinez, Alfonso and Fabregas, Albert Guillen i},
			journal={IEEE Transactions on Information Theory}, 
			title={Mismatched Decoding: Error Exponents, Second-Order Rates and Saddlepoint Approximations}, 
			year={2014},
			volume={60},
			number={5},
			pages={2647-2666},
			doi={10.1109/TIT.2014.2310453}}
